%% file: main.tex
\documentclass{egpubl}
\usepackage{egsgp19}
 
\SpecialIssuePaper         

\usepackage[T1]{fontenc}
\usepackage{dfadobe}  

\usepackage{cite}  
\BibtexOrBiblatex
\electronicVersion
\PrintedOrElectronic
\ifpdf \usepackage[pdftex]{graphicx} \pdfcompresslevel=9
\else \usepackage[dvips]{graphicx} \fi

\usepackage{egweblnk} 

\ifpdf \usepackage[pdftex]{graphicx} \pdfcompresslevel=9
\else \usepackage[dvips]{graphicx} \fi

\usepackage{t1enc,dfadobe}

\usepackage{multirow}
\usepackage{amscd, amsmath, amssymb,isomath}
\usepackage{dsfont} 
\newcommand{\matr}[1]{\mathbfit{#1}}			                								
\renewcommand{\vec}[1]{\mathbf{#1}}                             					
\newcommand{\tran}{^{\mkern-1.5mu {T}}}	                        					
\DeclareMathOperator*{\argmin}{arg\,min}		                							
\usepackage{overpic}
\usepackage{pgfplots}
\usepackage{pgf,tikz}
\usepackage{bm}
\usepackage{wrapfig}

\usepackage{soul}
\usepackage{xcolor}

\newcommand{\myparagraph}[1]{{\vspace{2mm} \noindent \textbf{#1.}}}

\newtheorem{theorem}{Theorem}
\newtheorem{lemma}{Lemma}
\newtheorem{definition}[theorem]{Definition}

\usepackage[noabbrev,capitalize]{cleveref}

\renewcommand{\hl}[1]{#1}

\title[Structured Regularization of Functional Map Computations]%
      {Structured Regularization of Functional Map Computations}

\author[J. Ren, M. Panine, P. Wonka, \& M. Ovsjanikov]
{\parbox{\textwidth}{\centering Jing Ren$^{1}$, 
        Mikhail Panine$^{2}$\orcid{0000-0001-7946-0584}, 
        Peter Wonka$^{1}$\orcid{0000-0003-0627-9746},
        and Maks Ovsjanikov$^{2}$\orcid{0000-0002-5867-4046}
        }
        \\
{\parbox{\textwidth}{\centering $^1$KAUST, \quad
$^2$LIX, \'Ecole Polytechnique, CNRS 
}
}
}

\begin{document}


\maketitle
\begin{abstract}
  We consider the problem of non-rigid shape matching using the functional map
  framework. Specifically, we analyze a commonly used approach for regularizing functional maps,
  which consists in penalizing the failure of the unknown map to commute with the Laplace-Beltrami
  operators on the source and target shapes. We show that this approach has certain undesirable
  fundamental theoretical limitations, and can be undefined even for trivial maps
  in the smooth setting. Instead we propose a novel, theoretically well-justified approach for
  regularizing functional maps, by using the notion of the \emph{resolvent} of the Laplacian
  operator. In addition, we provide a natural one-parameter family of regularizers, that can be
  easily tuned depending on the expected approximate isometry of the input shape pair. We show on a
  wide range of shape correspondence scenarios that our novel regularization leads to an improvement in the quality of the estimated functional, and ultimately pointwise
  correspondences before and after commonly-used refinement techniques.
\end{abstract}  

\input{sec_introduction}

\input{sec_relatedWork}

\input{sec_background}
\label{sec:background}

\input{sec_method}

\input{sec_results}

\input{sec_limitation_futureWork_conclusion}


\paragraph*{Acknowledgement}
The authors would like to thank the anonymous reviewers for their valuable comments and helpful suggestions.
Parts of this work were supported by the KAUST OSR Awards No. CRG2017-3426 and CRG2018-3730, a gift from the NVIDIA Corporation, and the ERC Starting Grant No. 758800 (EXPROTEA).

\bibliographystyle{eg-alpha}
\bibliography{reference}
\appendix
\input{sec_appendix_HilbertSchmidt.tex}

\input{sec_appendix_heatMask.tex}
\input{sec_appendix_descSize.tex}
\input{sec_appendix_stability.tex}
\end{document}

%% file: sec_introduction.tex
\section{Introduction} Shape matching is a fundamental problem in geometry processing and computer graphics, with
applications ranging from shape interpolation \cite{heeren2016splines} to shape exploration \cite{huang2014functional} and statistical shape analysis \cite{Bogo:CVPR:2014}.

In this paper, we concentrate on the functional maps framework for computing correspondences between shapes, which has proven to be especially useful when dealing with non-rigid and especially near-isometric shape pairs. This approach, originally introduced in \cite{ovsjanikov2012functional} and subsequently extended in multiple follow-up works, e.g. \cite{kovnatsky2013coupled,rodola2017partial,huang2014functional}, is based on optimizing for a linear mapping between function spaces defined on the shapes, which can be conveniently encoded as a matrix. One attractive feature of this approach is that many desirable geometric objectives for the unknown pointwise correspondence, such as preservation of geodesic distances, can be conveniently encoded as constraints on the matrices representing functional maps, and often lead to simple convex optimization problems.

The majority of objectives when optimizing for a functional map between a pair of shapes consist of a) preservation of some pre-computed descriptors and b) a functional map regularization term, based on promoting some desired global map properties. This latter, regularization term is especially important for ensuring the overall global consistency of the computed map.

Common strategies for regularization include commutativity with the Laplacian operator, exploited in the original article \cite{ovsjanikov2012functional}, a mask promoting a slanted diagonal in the case of partial correspondences \cite{rodola2017partial}, or orthonormality of the functional map, which corresponds to local \emph{area preservation} by the underlying pointwise map \cite{ovsjanikov2012functional,kovnatsky2013coupled,rustamov13}.

\begin{figure}[!t]
\vspace{6pt}
    \centering
    \begin{overpic}
    [trim=0cm 0cm 0cm 0cm,clip,width=1\columnwidth, grid=false]{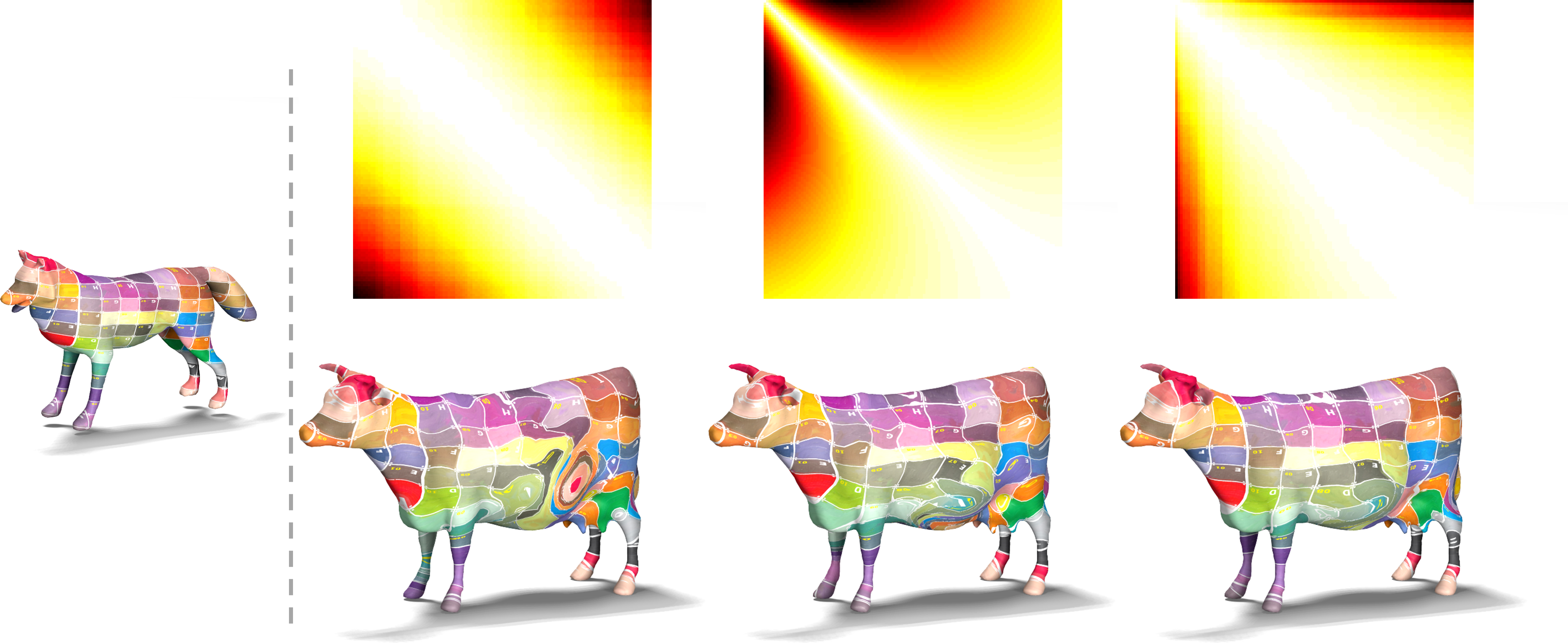}
    \put(3,28){Source}
    \put(25,42){Standard}
    \put(52,42){Slanted}
    \put(80,42){\textbf{Ours}}
    \put(99,40.5){\rotatebox{-90}{Regularizing}}
    \put(95,35){\rotatebox{-90}{mask}}
    \put(99,18){\rotatebox{-90}{Pointwise}}
    \put(95,15){\rotatebox{-90}{map}}
    
  \end{overpic}
    \caption{Comparison of different masks. Standard: the Laplacian commutativity operator can be equivalently formulated as a penalty or regularizing mask. Slanted: the weight mask proposed in~\cite{rodola2017partial} designed to promote a slanted structure. Ours: the mask proposed in this paper based on the resolvent of the Laplacian operator. The penalty masks of these three methods are visualized in the first row for an example pair from the SHREC dataset, and the corresponding optimized point-wise maps from the functional map pipeline are shown in the second row. \vspace{-2mm}     \label{fig:teaser_v2} }
    \vspace{-0.2cm}
\end{figure}
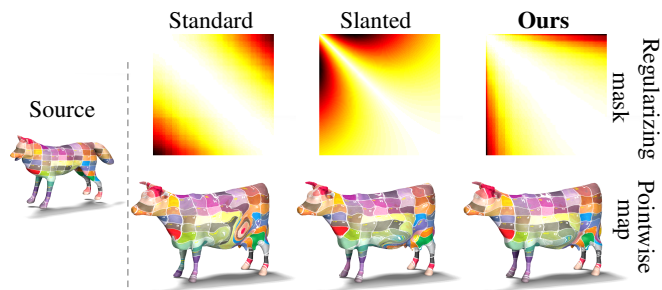

Unfortunately, although  functional map regularization has been instrumental in obtaining high-quality results on some shape correspondence benchmarks, these regularization terms often lack a theoretical foundation and indeed in some cases result in optimization objectives which would not be well-defined in the case of smooth surfaces.


In this paper, we introduce a novel way of regularizing functional maps, based on the concept of the resolvent of the Laplacian operator. Unlike the original approach, based on commutativity with the Laplacians, our method has a natural theoretically well-defined analogue in the case of smooth surfaces, and is guaranteed to be bounded even in the full (infinite dimensional) basis case. Moreover, our approach provides a simple one-parameter family of regularizers, that can be tuned depending on the expected approximate deviation from isometry in a given shape pair. Finally, we demonstrate on a wide range of challenging settings that our approach  leads to a quantitative and qualitative improvement for the computed functional and eventually pointwise maps, without incurring any additional computational complexity (see Fig.~\ref{fig:teaser_v2} as an example).

%% file: sec_relatedWork.tex
\section{Related Work}
Shape matching in its full generality is \hl{an} extremely well-studied area of geometry processing and
computer graphics, and its full overview is beyond the scope of our paper. Therefore, below we only
review the most closely related methods based primarily on the functional maps framework. We refer
the interested readers to recent surveys including
\cite{van2011survey,tam2013registration,biasotti2016recent} for an in-depth treatment of other
shape matching approaches.

\myparagraph{Functional Maps} Our approach fits within the functional map framework, which was
originally introduced in \cite{ovsjanikov2012functional} for solving non-rigid shape matching
problems, and extended significantly in follow-up works, including
\cite{kovnatsky2013coupled,aflalo2013spectral,kovnatsky2015functional,rodola2017partial,ezuz2017deblurring,burghard2017embedding} among many others
(see also \cite{ovsjanikov2017course} for an overview). The key observation in these techniques
is that it is often easier to estimate correspondences between real-valued functions, rather than
points on the shapes. This is because functions have a natural vector space structure, and moreover
functional transformations are most often linear, which means that functional maps can be encoded,
and thus optimized for, as small matrices in a reduced functional basis.

As observed by several works in this domain,
\cite{kovnatsky2013coupled,rustamov13,rodola2017partial,burghard2017embedding} many natural properties on the
underlying pointwise correspondences can be expressed as objectives on functional maps. Most notably,
this includes: orthonormality of functional maps, which corresponds to the local area-preservation
nature of pointwise correspondences \cite{ovsjanikov2012functional,kovnatsky2013coupled,rustamov13}; preservation
of inner products of gradients of functions, which corresponds to conformal maps \cite{rustamov13,burghard2017embedding,wang2018vector};
preservation of \emph{pointwise products} of functions, which corresponds to functional maps arising
from point-to-point correspondences \cite{nogneng2017informative,nogneng2018improved}; slanted diagonal structure of
functional maps, which corresponds to correspondences between partial shapes
\cite{rodola2017partial,litany2017fully}. Similarly, several other regularizers have been proposed, including using
robust norms and matrix completion techniques \cite{kovnatsky2013coupled,kovnatsky2015functional}, exploiting the relation between functional maps in different directions \cite{eynard2016coupled}, the map adjoint  \cite{huang2017adjoint}, and powerful
cycle-consistency constraints \cite{huang2014functional} in the context of shape collections, among many
others. More recently constraints on functional maps have been introduced to promote
\emph{continuity} of the recovered pointwise correspondence \cite{poulenard2018topological}, maps between curves defined on
shapes \cite{gehre2018interactive}, kernel-based techniques aimed at extracting more information
from given descriptor constraints \cite{wang2018kernel}, and an approach for incorporating
\emph{orientation} information into the functional map infererence pipeline \cite{ren2018continuous} among
others.

Among all of these, perhaps the most widely-used building block for regularizing functional maps,
introduced in \cite{ovsjanikov2012functional} and extended in follow-up works, including
\cite{wang2013image,rodola2017partial,litany2016non,litany2017fully}, is based on the commutativity with the Laplacian operators, which implies
a diagonal (or slanted diagonal in the case of partial correspondence) structure for functional
maps. To promote this structure, the most common method (see also Chapter 2.4 in
\cite{ovsjanikov2017course}) consists in adding an energy to the functional map estimation
pipeline, which penalizes the failure of the unknown functional map to commute with the
Laplace-Beltrami operators on the source and target shapes. Conveniently, this term still leads to a
convex optimization problem during functional map estimation. Moreover, as observed in
\cite{ovsjanikov2012functional}, for functional maps arising from pointwise ones, this term is zero
if and only if the map preserves geodesic distances exactly. 

Unfortunately, although functional map regularization via commutativity with the Laplace-Beltrami
operators has been instrumental in obtaining high quality results in challenging benchmarks, the
exact properties of this regularization are not well-understood. In particular, as we show below,
the commonly-used energy is not bounded in the full (infinite-dimensional) basis in general. Instead, our novel regularizer, although based on a
similar underlying principle, overcomes this limitation, and both leads to a theoretically better
justified energy, and a practical improvement on a range of challenging datasets.



\myparagraph{Optimal Transport} We also note briefly that other commonly-used relaxations for
matching problems include those based on optimal transport,
e.g. \cite{solomon2016entropic,mandad2017variance,vestner2017product}.  These techniques have
recently gained prominence especially due to the computational advances for adressing large-scale
transport problems, primarily using the Sinkhorn method with entropic regularization
\cite{cuturi2013sinkhorn,solomon2015convolutional}. Furthermore, other techniques that exploit the
formalism of optimal transport, for solving assignment problems include the recent variants
of Product Manifold Filter using Gaussian and Heat Kernels
\cite{mandad2017variance,vestner2017product}. Interestingly, these latter methods argue that
near-isometric shape matching can be better addressed via preservation of general functional
\emph{kernels} rather than preservation of e.g. geodesic distances. Our modification of the
regularization term of functional maps can also be seen through a similar light, as we show that a
more theoretically justified functional operator leads to an improvement of the overall structural
properties of the map, which eventually result in more accurate functional and pointwise maps.


%% file: sec_background.tex
\section{Background}
Our work is based on the functional map representation and the estimation pipeline. Below we review
the basic notions and the main steps for estimating a map between a pair of shapes using this
framework, and refer the interested reader to a recent set of course notes
\cite{ovsjanikov2017course} for a more in-depth discussion.

\myparagraph{Basic Pipeline} Given a pair of shapes, $S_1,S_2$ represented as triangle meshes, and
containing, respectively, $n_1$ and $n_2$ vertices, the general pipeline for computing a map between
them using the functional map representation, consists of the following main steps:
\begin{enumerate}
\item Compute a small set of $k_1 << n_1$ and $k_2 << n_2$ basis functions on each shape. The most common choice consists in using the first $k$ eigenfunctions of the
  Laplace-Beltrami operator of each shape, although other bases derived from  the Hamiltonian operator \cite{choukroun2018hamiltonian} and more localized basis functions \cite{neumann2014compressed,melzi2018localized} have also been used.
\item Compute a set of descriptor functions on each shape, that are expected
  to be approximately preserved by the unknown map. Store their coefficients in the corresponding
  bases as columns of matrices $A_1, A_2$.
\item Compute the optimal \emph{functional map} $C$ by solving the following optimization problem:
\begin{align}
C_{\text{opt}} = \argmin_{\matr{C}_{12}} E_{\text{desc}}\big(\matr{C}_{12}\big) + \alpha E_{\text{reg}}\big(\matr{C}_{12}\big) 
\end{align}
where the first term aims at the descriptor preservation: $E_{\text{desc}}\big(\matr{C}_{12}\big)
= \big\Vert \matr{C}_{12} \matr{A}_1 - \matr{A}_2\big\Vert^2$, whereas the second term regularizes
the map by promoting the correctness of its overall structural properties. As mentioned above, the
most common approach consists of penalizing the failure of the unknown
functional map to commute with the Laplace-Beltrami operators, which can be written as:
        \begin{align}
        \label{eq:energy:laplacian_comm}
        E_{\text{reg}}(C_{12}) = E_{\text{comm}}\big(\matr{C}_{12}\big) = \big\Vert \matr{C}_{12}\Delta_1 - \Delta_2\matr{C}_{12} \big\Vert^2    
        \end{align}
         where $\Delta_1$ and $\Delta_2$ are the Laplace-Beltrami operators of the two shapes expressed in the respective bases. Here, and throughout the rest of our paper, unless stated otherwise $\Vert \,\cdot \,\Vert$ denotes the matrix Frobenius norm.
\item Convert the functional map $C$ to a point-to-point map, for example using an iterative approach,
  such as the Iterative Closest Point (ICP) in the spectral embedding, or using other more advanced
  techniques \cite{rodola-vmv15,ezuz2017deblurring}.
\end{enumerate}

One of the attractive properties of this pipeline is that the functional map computation in step 3 leads to a simple (convex) least squares optimization with a relatively small number of unknowns, independent
of the number of points on the shapes. This step has been further extended e.g. using more powerful descriptor preservation constraints via commutativity \cite{nogneng2017informative}, or using manifold optimization \cite{kovnatsky2016madmm} among many others (see also Chapter 3 in \cite{ovsjanikov2017course}).

%% file: sec_method.tex
\section{Functional Map Regularization}
\begin{figure}[!t]
  \vspace{0.5cm}
  \centering
  \begin{overpic}
  [trim=0cm 0cm 0cm 0cm,clip,width=0.95\linewidth, grid=false]{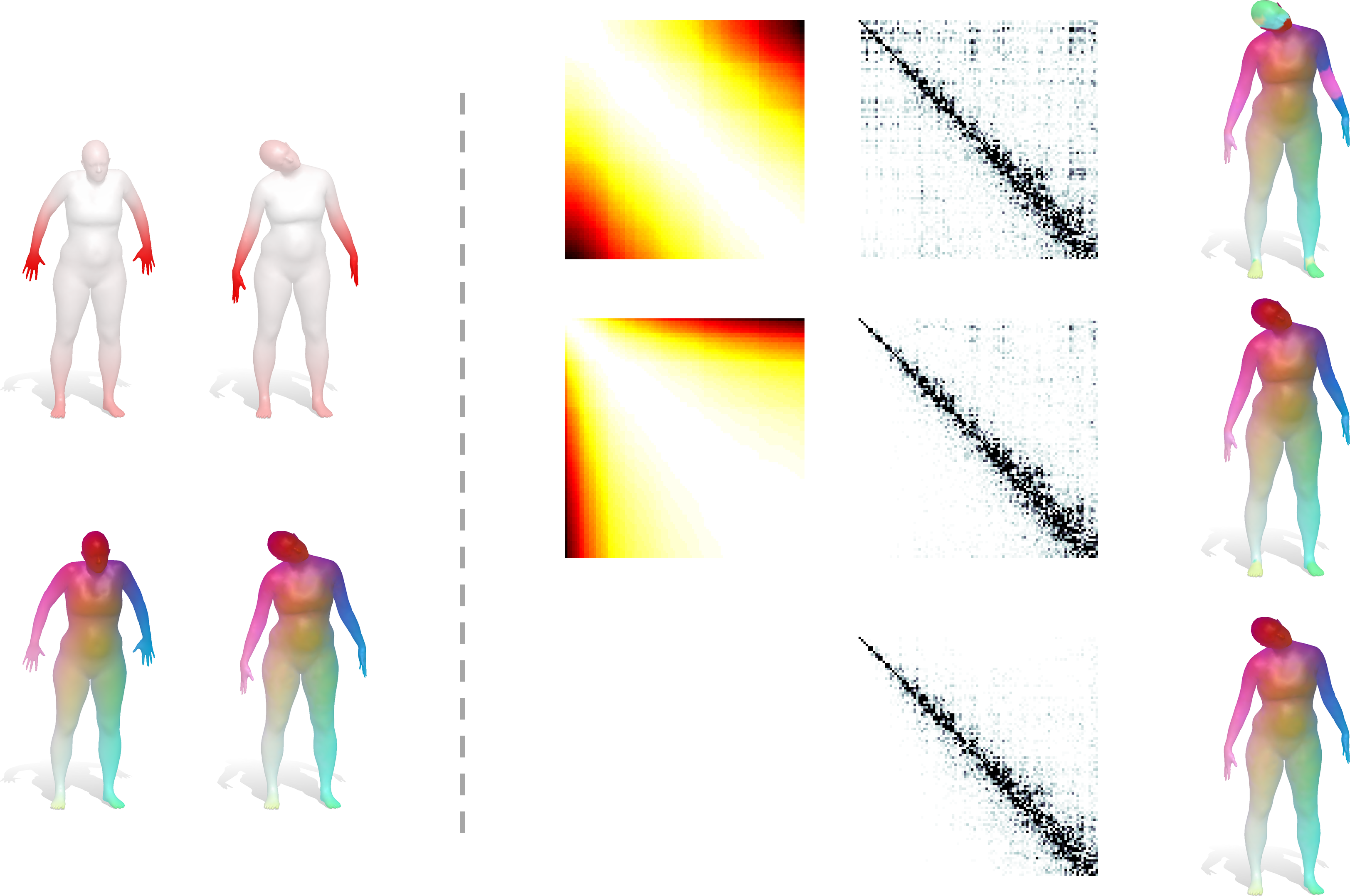}
  \put(2,60){Given descriptor}
  \put(2,30){Ground-truth map}
  \put(101,21){\rotatebox{-90}{Ground-truth}}
  \put(101,41){\rotatebox{-90}{\textbf{Ours}}}
  \put(101,63){\rotatebox{-90}{Standard}}
  
  \put(45, 68.5){Mask}
  \put(64, 70){Functional}
  \put(68, 67){map}
  \put(88, 70){Pointwise}
  \put(92, 67) {map}
  \end{overpic}
\caption{Given a single pair of WKS descriptors, we optimize a 100-by-100 functional map using the 
standard mask and our resolvent mask and compare to ground truth. The converted point-wise maps are shown on the right. We can see that the functional map stemming from our resolvent mask has less noise than the functional map computed with the standard mask. Also, our resolvent mask leads to a 
point-wise map with better quality. }
\label{fig:mtd:example}
\vspace{-0.3cm}
\end{figure}
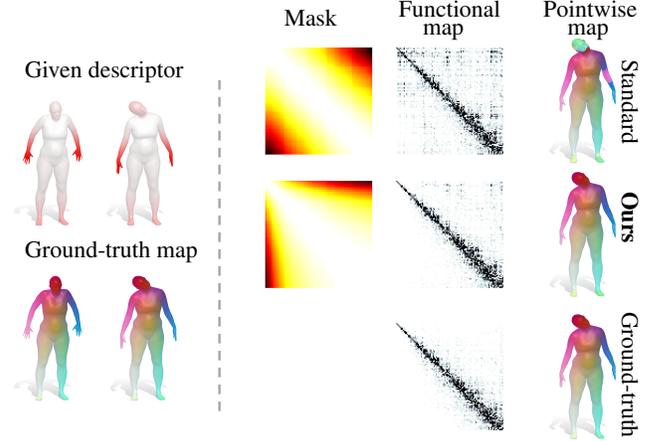

Our main goals are to analyze the commonly-used functional map regularization term, to bring attention to some of its
theoretical limitations, and to propose a novel regularizer with better theoretical properties, which lead to practical
improvements. Therefore, in this in this section, we first consider the standard Laplacian-commutativity term, and then
introduce our new regularizer based on the resolvent operator of the Laplacian.

\subsection{Reformulation of the Laplacian-Commutativity term}
As mentioned above, the Laplacian-commutativity term given in Eq. \eqref{eq:energy:laplacian_comm} was first introduced
to promote approximately isometric correspondences. If the functional map is expressed in the basis given by the
eigenfunctions of the Laplace-Beltrami operator, and letting $\Lambda_1$ and $\Lambda_2$ represent vectors that store
the eigenvalues of the Laplacians of shape $S_1$ and $S_2$ respectively, this term can be equivalently reformulated as:
\begin{equation} \label{eq:comm_as_mask}
\begin{split}
E_{\text{comm}}\big(\matr{C}_{12}\big) 
 &= \big\Vert \matr{C}_{12}\Delta_1 - \Delta_2\matr{C}_{12} \big\Vert^2 \\
 &= \big\Vert\matr{C}_{12}\text{diag}\big(\Lambda_1\big) - \text{diag}\big(\Lambda_2\big)\matr{C}_{12}\big\Vert^2\\
 &= \big\Vert \matr{C}_{12} \odot \big(\mathbf{1}_{k_2}\Lambda_1\tran\big) - \big(\Lambda_2\mathbf{1}_{k_1}\tran\big)\odot \matr{C}_{12} \big\Vert^2\\
 &= \big\Vert \big(\mathbf{1}_{k_2}\Lambda_1\tran - \Lambda_2\mathbf{1}_{k_1}\tran \big) \odot \matr{C}_{12}\big\Vert^2\\
 &\triangleq E_{\text{mask}}\big(\matr{C}_{12}\big) = \sum_{ij} \left[\matr{M}_{LB} \right]_{ij} \left[ \matr{C}_{12} \right]_{ij}^2~, 
\end{split}\end{equation}
where $\odot$ is the entry-wise matrix multiplication, $\mathbf{1}_{k}$ is a $k$-dim all-ones vector and $[A]_{ij}$ denotes the $(i,j)$ entry of a matrix $A$. 

In other words, the regularization term  $E_\text{comm}$ can be written as a product between the matrix $M_{LB}$, which
we call the penalty \emph{mask} matrix and the squares of the entries of the unknown functional map $C_{12}$. In the case of $E_\text{comm}$, the
entries of $M_{LB}$ are given as $M_{LB}(i,j) = (\Lambda_2(i)-\Lambda_1(j))^2$. Fig.~\ref{fig:mtd:example} shows a heat
map of this matrix (see the first row in the ``Mask'' column).

Unfortunately, this simple formulation has certain fundamental undesirable properties. In particular, in the case of smooth surfaces, Laplace-Beltrami operators are \emph{unbounded} operators on square-integrable functions \cite{minakshisundaram1949some}.
Consequently, in general, the energy term $\| C_{12}\Delta_{1} - \Delta_{2}C_{12} \|^{2}$ is undefined on smooth surfaces.
As a simple example, consider the situation where $C_{12} = Id$, the identity operator, and the two surfaces are scaled versions of each other (i.e., $\Delta_2 = c\Delta_1$ for some constant $c \neq 1$). In this case, $\| C_{12}\Delta_{1} - \Delta_{2}C_{12} \|^{2} = |1-c|^2 \| \Delta_1 \|^2$ is infinite. 

That being said, the ill-definiteness of $\| C_{12}\Delta_{1} - \Delta_{2}C_{12} \|^{2}$ is not a mere question of scale. Recall that, by Weyl's law \cite{dodziuk1981eigenvalues}, large Laplacian eigenvalues of surfaces can be estimated in terms of the surface area $S$ as follows:
\begin{equation}
    \lambda_k \sim \frac{4 \pi}{S}k~.
\end{equation}
\noindent Thus, by rescaling the surfaces such that their areas match guarantees that their eigenvalues have comparable
asymptotic growth. This, however, is not sufficient to make $\| C_{12}\Delta_{1} - \Delta_{2}C_{12} \|^{2}$ well-defined. This can be seen from another simple, if slightly artificial, example. Consider a sphere ($\mathcal{M}_1$) and a flat square torus ($\mathcal{M}_2$), both of unit
area. While these surfaces are not diffeomorphic, they are among the few whose Laplace-Beltrami eigenvalues can be explicitly computed (see \cite{sauvigny2006partial1,sauvigny2006partial2} among many others), which is why we consider them here. Said spectra, as well as the corresponding Weyl estimate are illustrated \hl{in} Fig. \ref{fig:mtd:sphere_torus}.
For simplicity's sake, we once again pick a functional map of the form $C_{12} = Id$. 
We illustrate the spectra, as well as the corresponding $\| C_{12}\Delta_{1} - \Delta_{2}C_{12} \|^{2}$ on Fig.  \ref{fig:mtd:spectra_combined}. Notice the divergence of the Laplacian commutator energy. On the same figure, we illustrate its proposed replacement, defined in the next section. Note its rapid convergence.

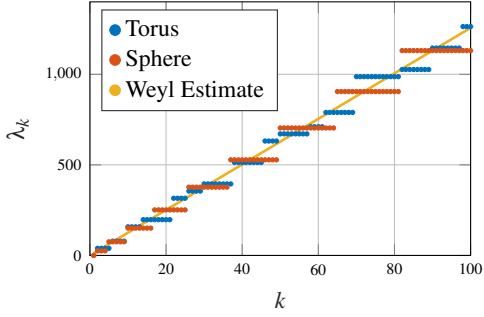
\begin{figure}[!t]
\centering
\input{figures_tex/SphereTorusSpectra.tex}
\caption{The spectra of a sphere and a square torus of unit area, as well as the corresponding Weyl estimate. }
\label{fig:mtd:sphere_torus}
\end{figure}


In addition to being ill-defined in the continuous setting, the Laplacian commutativity mask has another significant
problem. Namely, it penalizes the high frequency entries of functional maps in the same way as the low frequency ones,
despite the greater instability of the former, and in this way fails to exhibit the funnel-like structure observed in
ground-truth functional maps (see e.g. Fig.~\ref{fig:mtd:example} above or Figure 4 in
\cite{ovsjanikov2012functional}). Recall that a mask serves as a penalty during the optimization of the functional map
matrix.  Thus, the mask and the ground truth functional map should be complementary, in the sense that for the regions
of ground-truth functional maps with smaller (resp. large) values, the mask should add more (resp. less) penalty.

We illustrate this phenomenon in Fig.~\ref{fig:res:gt_and_masks}, which shows the average of the squared values of functional maps constructed from the ground-truth correspondences using 250 eigenfunctions over 100 FAUST non-isometric shape pairs. We then compare its structure to the Laplacian commutativity mask (labeled ``Standard''), averaged over the same shape pairs. We perform the same computation for the heuristic slanted mask introduced in \cite{rodola2017partial} and, finally, for our proposed resolvent-based mask defined in the next section. Notice that our mask better reproduces the funnel-like structure of the ground-truth maps.

\begin{figure}[!t]
    \centering
    \input{figures_tex/spectra_combined.tex}
    \caption{The Laplacian and Resolvent commutator energies for $C_{12} = Id$ computed using the lowest $k$ eigenvalues
      of the unit area sphere and torus (see Fig. \ref{fig:mtd:sphere_torus}). Note the rapid convergence of $\|
      C_{12}R(\Delta_{1}) - R(\Delta_{2})C_{12} \|^{2}$ and the divergence of $\| C_{12}\Delta_{1} - \Delta_{2}C_{12}
      \|^{2}$. The spectra are rescaled with respect to the largest eigenvalue with $k = 100$, as described in
      Eq.~\eqref{eq:mtd:rescaling}.} 
    \label{fig:mtd:spectra_combined}
    \vspace{-0.3cm}
\end{figure}
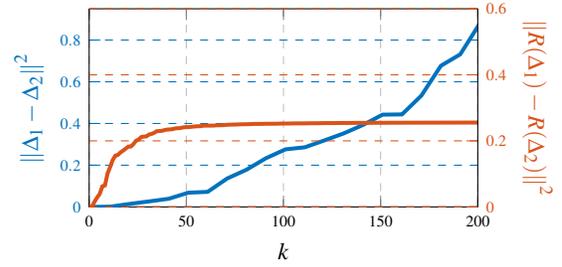


\subsection{A novel regularization based on the Resolvent}

In this section, we propose an alternative form for the Laplacian commutativity regularizer that overcomes the problems identified in the previous section. The resulting regularizer is better from both the theoretical and practical standpoints.

The first issue of the original Laplacian-commutativity term arises from the fact that the Laplace-Beltrami operator is
unbounded. We thus propose to replace it with a meaningful bounded operator, namely its
resolvent. Below, we give a brief overview and refer the reader to \cite{reedsimon1} for the detailed
functional-analytic underpinnings of the following discussion.

\myparagraph{Resolvent} Let $A: \mathcal{D} \to \mathcal{H}$ be a densely-defined closed  operator on some Hilbert space $\mathcal{H}$ with domain $\mathcal{D} \subset \mathcal{H}$. Let $\rho(A) \subset \mathbb{C}$ be the set of all complex numbers $\mu$ such that $R_{\mu}(A) = \left( A - \mu Id \right)^{-1}$ is defined and bounded. The set $\rho(A)$ and the operator $R_{\mu}(A)$ are known as the resolvent set of $A$ and the resolvent (operator) of $A$ at $\mu$, respectively. The resolvent set $\rho(A)$ is the complement of the spectrum of $A$ in the complex plane.

The general idea of considering the resolvent of an unbounded operator rather than the operator itself comes from the
notion of norm-resolvent convergence, which is used to study the convergence of unbounded self-adjoint
operators \cite{reedsimon1}. In that sense, our choice to use the resolvent of the Laplace-Beltrami operator is a natural one. Note that here, closedness is a technical condition used in the definition of the resolvent. In particular, it is satisfied by self-adjoint operators, such as the ones we consider.

Before proceeding further, we slightly generalize our problem. In what follows we will consider powers of the Laplacian
rather than the Laplacian itself. Specifically, we will use $\Delta^{\gamma}$ for some $\gamma > 0$ rather than
$\Delta$. As explained below, the parameter $\gamma$ controls the funnel-like structure of the resulting mask. Thus, this parameter takes on an important role in the numerical tests reported later in this paper.

Now since the Laplace-Beltrami operator $\Delta$ is positive and self-adjoint, its spectrum is contained in the non-negative real line. The same also holds for $\Delta^\gamma$. Thus, we are guaranteed that $R_{\mu}(\Delta)$ is a well-defined bounded operator for any complex $\mu$ not in the non-negative real line.
 
\noindent For our purposes, the resolvent of $\Delta^{\gamma}$ is expressed as
\begin{equation}
    R(\Delta^{\gamma}) = \left( \Delta^{\gamma} - (a + i b) Id \right)^{-1}~,
\end{equation}

\noindent where $i$ is the imaginary unit and $a,b \in \mathbb{R}$. Our proposal is to use the resolvent of the Laplacian instead of the Laplacian in the commutator term. We thus define a new energy term,
\begin{equation}
\label{eq:res_penalty}
    E_{\text{resolvent}}\big(\matr{C}_{12}\big) = \big\Vert \matr{C}_{12} R \left( \Delta_1^{\gamma} \right) - R \left( \Delta_2^{\gamma} \right) \matr{C}_{12} \big\Vert^2~.
\end{equation}

\noindent Before proceeding further, recall that we only need to consider bounded functional maps $C_{12}$, since functional maps that arise as pullbacks of \hl{diffeomorphisms} are bounded. This last fact is shown in appendix \ref{appendixPullback}, for completeness.
Our proposal is motivated by the following result.


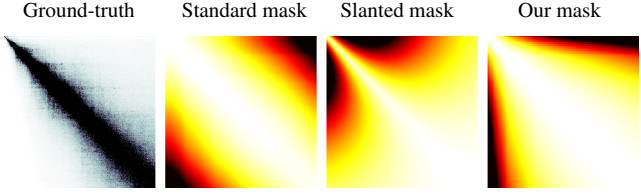
\begin{figure}[!t]
    \vspace{12pt}
    \centering
    \begin{overpic}
    [trim=0cm 0cm 0cm 0cm,clip,width=1\columnwidth, grid=false]{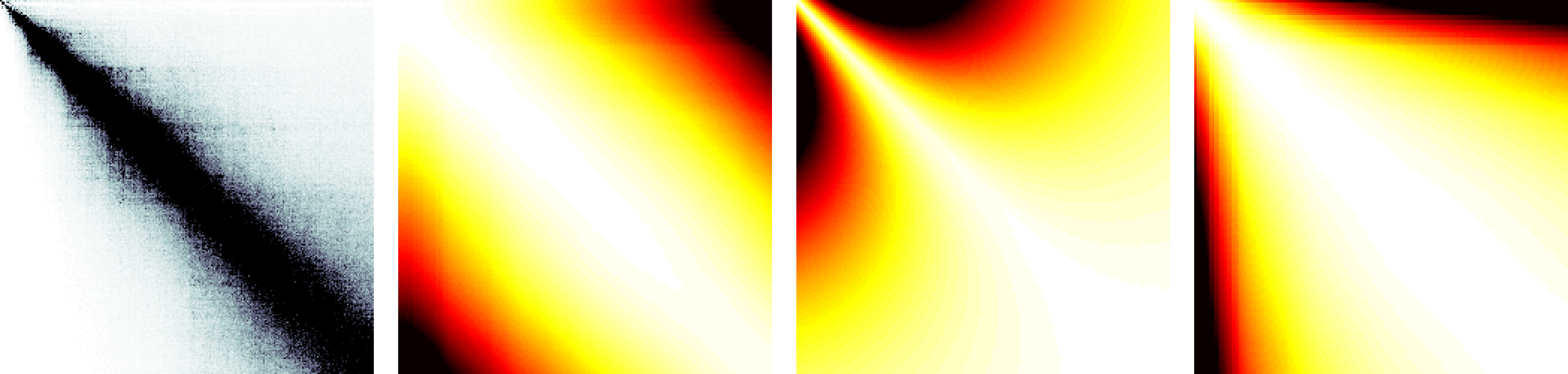}
    \put(3,27){\footnotesize{Ground-truth}}
    \put(28,27){\footnotesize{Standard mask}}
    \put(53,27){\footnotesize{Slanted mask}}
    \put(81,27){\footnotesize{Our mask }}
  \end{overpic}
    \caption{The average squared ground-truth functional map and three penalty masks over 100 FAUST non-isometric shape pairs.
  }
    \label{fig:res:gt_and_masks}
    \vspace{-0.3cm}
\end{figure}

\begin{theorem}[Bounded Resolvent Commutativity]\label{Thm:BoundedResolvent}
Let $C_{12}$ be a bounded functional map. Then, in the operator norm,
\begin{equation}
    \left\| \matr{C}_{12} R \left( \Delta_1^{\gamma} \right) - R \left( \Delta_2^{\gamma} \right) \matr{C}_{12}  \right\|^{2}_{op} < \infty~.
\end{equation}
\end{theorem}
\begin{proof}
The result directly follows from the fact that products and linear combinations of bounded operators are bounded. As defined above, $R(\Delta^{\gamma})$ is bounded, which concludes the proof.
\end{proof}

As explained above, there is no analogous guarantee for the commonly used Laplacian commutativity, since the Laplace-Beltrami operator itself is unbounded. In contrast, our resolvent-based commutativity is always well-defined, and leads to a commutator with bounded operator norm, even in the case of smooth surfaces.

\hl{
Notice that Theorem {\ref{Thm:BoundedResolvent}} holds in the operator norm, rather than the more convenient Frobenius norm, which we use to define the energy in Equation {\eqref{eq:res_penalty}}. Our usage of the Frobenius norm can be justified in two ways. First, a Frobenius norm version of Theorem {\ref{Thm:BoundedResolvent}} holds for $\gamma >1/2$. This is shown in Lemma {\ref{LemmaBoundedEnergy}} of Appendix {\ref{appendixHilbertSchmidt}}. Thus, for $\gamma > 1/2$, the usage of the Frobenius norm is perfectly justified in the smooth setting. For lower values of $\gamma$, we invoke the fact that, in practice, we work on a finite dimensional vector space. Since all norms are equivalent on finite dimensional vector spaces, we are free to replace the operator norm with the Frobenius one. In that case, we lose the guarantee that the energy $E_{resolvent}$ makes sense in the full Laplace-Beltrami basis in the smooth setting. Yet, as shown in the experiments reported below, this does not seem to cause any issue.}


 Similarly to the usual Laplacian commutativity term $E_{\text{comm}}$, $E_{\text{resolvent}}$ can also be expressed as a mask matrix. This can be done as follows. Notice that $R(\Delta^{\gamma})$ is diagonal in the Laplacian eigenbasis. Moreover, if $\Delta$ has eigenvalues $\{ \lambda_{n} \}_{n=0}^{\infty}$, then, by matrix inversion, the eigenvalues $\{ r_n \}_{n=0}^{\infty}$ of $R(\Delta^{\gamma})$ will be given by:
 
\begin{equation} \label{eq:mtd:resolvent_eigenvalues}
    r_n = \frac{1}{\lambda^{\gamma}_n - a - ib}  = \frac{\lambda^{\gamma}_n - a}{(\lambda^{\gamma}_n - a)^2 + b^2} + \frac{b}{(\lambda^{\gamma}_n - a)^2 + b^2}i~.
\end{equation}
 
\noindent In practice, we choose $a = 0$ and $b = 1$, which yields 
 
 \begin{equation}
    r_n = \frac{\lambda^{\gamma}_n}{(\lambda^{\gamma}_n )^2 + 1} + \frac{1}{(\lambda^{\gamma}_n )^2 + 1}i~.
\end{equation}

\noindent Since the square of the Frobenius norm independently considers the real and imaginary parts of a matrix, we can re-express Equation \eqref{eq:res_penalty} in terms of two new matrices $M_{Re}$ and $M_{Im}$, defined below.

\begin{equation}
\begin{aligned}
    E_{\text{resolvent}}\big(\matr{C}_{12}\big) &= \big\Vert \matr{M_{Re}} \odot \matr{C}_{12}\big\Vert^2 + \big\Vert \matr{M_{Im}} \odot \matr{C}_{12}\big\Vert^2\\
    &= \sum_{ij} \left[M_{Re} \right]_{ij}^2 \left[ C_{12} \right]_{ij}^2 + \sum_{ij} \left[M_{Im} \right]_{ij}^2 \left[ C_{12} \right]_{ij}^2\\
    &= \sum_{ij} \left( \left[\matr{M}_{Re}\right]_{ij}^2 + \left[M_{Im} \right]_{ij}^2 \right) \left[ C_{12} \right]_{ij}^2
\end{aligned}
\label{eq:mtd:resolvent2masks}
\end{equation}

\noindent The matrices $\matr{M}_{Re}$ and $\matr{M}_{Im}$ correspond to the real and imaginary parts of the eigenvalues of the resolvent, respectively. Explicitly, these matrices are given by

\begin{align}
    \label{eq:mtd:def:mask_real} &\matr{M}_{Re}(i,j) = \frac{\Lambda_2(i)^{\gamma}}{\Lambda_2(i)^{2\gamma} + 1} - \frac{\Lambda_1(j)^{\gamma}}{\Lambda_1(j)^{2\gamma} + 1}\\
    \label{eq:mtd:def:mask_imag} &\matr{M}_{Im}(i,j) =  \frac{1}{\Lambda_2(i)^{2\gamma} + 1} - \frac{1}{\Lambda_1(j)^{2\gamma} + 1}
\end{align}

\noindent These matrices are not quite masks in the sense of Eq. \eqref{eq:comm_as_mask}. As per Equation \eqref{eq:mtd:resolvent2masks}, the above two matrices can be combined into a single mask 
\begin{equation}
    \label{eq:mtd:res_mask} \matr{M}_{res}(i,j) = \matr{M}_{Re}(i,j)^2 + \matr{M}_{Im}(i,j)^2~.
\end{equation}

\noindent The split of $\matr{M}_{res}$ into $\matr{M}_{Re}$ and $\matr{M}_{Im}$ will be \hl{revisited} later, when we explore beyond the established theory and consider a mask constructed from weighted combinations of these two matrices.\\

\subsection{Rescaling the Spectra}

As it stands now, we observe that in practice the mask defined in equation \eqref{eq:mtd:res_mask} decays too quickly as the Laplacian eigenvalues grow.
\hl{In other words, the mask is not sufficiently sensitive to the higher frequencies. This is due to the scale introduced by the parameter $b$ found in the definition of the resolvent. There are a few equivalent ways of addressing this issue.}


 \hl{Our approach is as follows.} We begin by computing the $k$ lowest eigenvalues of the considered Laplacians. Then, both spectra are rescaled according to the rule

\begin{equation} \label{eq:mtd:rescaling}
    \Lambda_{i} ~\longmapsto~ \frac{\Lambda_{i}}{ \max \big( \max(\Lambda_1), \max(\Lambda_2) \big)}~.
\end{equation}

\noindent In other words, the spectra of the source and target shape are rescaled by the same factor, such that the largest considered eigenvalue (over both spectra) becomes equal to $1$. From Equation \eqref{eq:mtd:resolvent_eigenvalues}, one can see that this rescaling can be absorbed into a choice of $b$ and a change in the weight of the resolvent energy in the overall energy. Thus, conceptually, the spectral rescaling is equivalent to a choice of the parameter $b$ used in the definition of the resolvent. Alternatively, recall that rescaling the spectrum is equivalent to rescaling the surface. 
\hl{Consequently, this rescaling does not affect the theoretical guarantees of the previous section.}

\noindent Returning to the previously discussed example of the sphere and torus of unit area (see Fig.~\ref{fig:mtd:sphere_torus}), we illustrate the resolvent energy computed by the above procedure in Fig.~\ref{fig:mtd:spectra_combined}. Note the rapid convergence of the resolvent commutator energy (in red) and the divergence of the Laplacian commutator energy (in blue).

\subsection{Mask Structure as a Function of $\gamma$} \label{sec:mask_structure}
\begin{figure}[!t]
\centering
\includegraphics[width = 0.98\columnwidth]{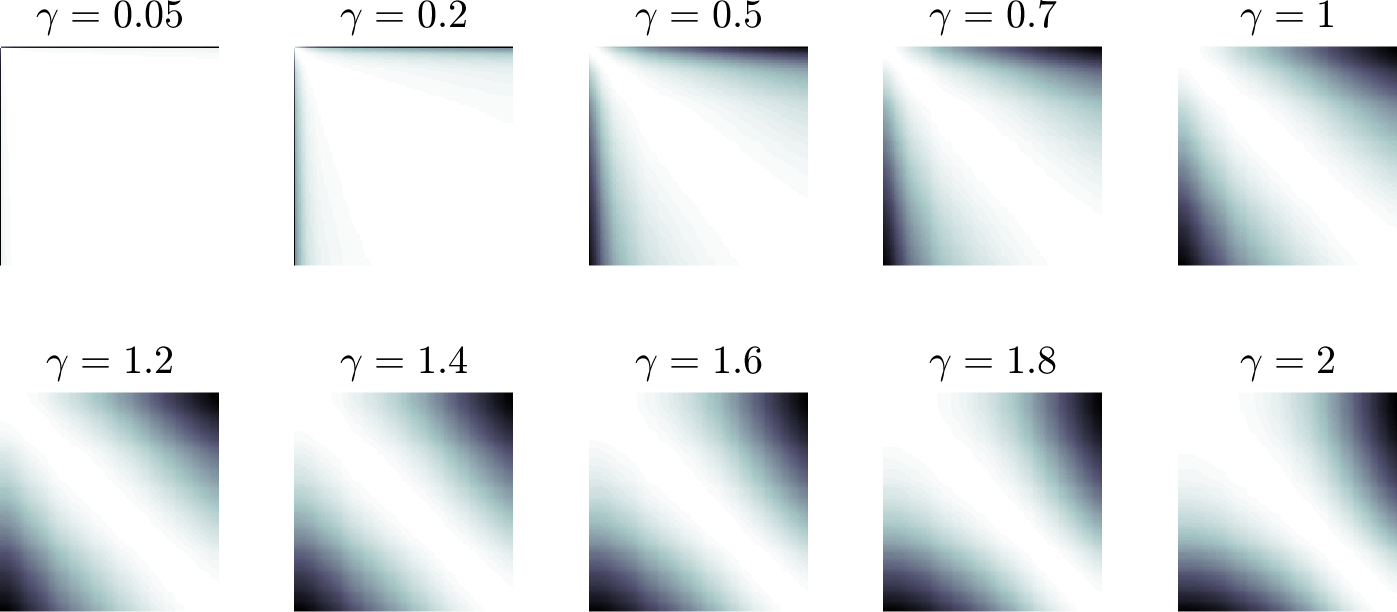}
\caption{The resolvent mask introduced in Eq.~\eqref{eq:mtd:res_mask} with different $\gamma$. The darker the region, the larger the penalty. Here, $\Delta_1$ and $\Delta_2$ are near-isospectral, which explains the approximately symmetric structure of the mask.}
\label{fig:mtd:mask_gamma}
\vspace{-0.3cm}
\end{figure}

Having defined the way to compute the resolvent mask, we are now ready to explore the way in which tuning the parameter $\gamma$ controls its the funnel-like structure.

Fig.~\ref{fig:mtd:mask_gamma} illustrates the resolvent mask for different values of $\gamma$. The darker the region, the more penalized the corresponding entry in the functional map $C_{12}$ would be. Notice that the funnel-like structure of the mask changes with different values of $\gamma$.

From Fig. \ref{fig:mtd:mask_gamma}, we see that the behaviour of the mask as a function of $\gamma$ can be separated into two regimes.
The first  corresponds to $\gamma \in (0,1]$. There, increasing $\gamma$ results in a narrowing of the funnel-like shape of the mask. Thus, the larger $\gamma$ is, the more the mask penalizes functional maps that take eigenfunctions of $\Delta_{1}$ to eigenfunctions of $\Delta_{2}$ with distant eigenvalues. Correspondingly, small values of $\gamma$ are more lax in that regard, allowing for maps between eigenfunctions with quite different eigenvalues.
The former choice of $\gamma$ is appropriate when the shapes under consideration are approximately isometric, as one then expects the eigenfunctions and eigenvalues of both surfaces to be roughly the same. The latter choice is sound for shape pairs that are further away from isometry.

The second regime corresponds to $\gamma > 1$. There, we observe an inversion of the funnel-like structure, with the reverse funnel shape being more and more pronounced as $\gamma$ increases. This results in a low penalty for maps between low frequency eigenfunctions, which does not respect the shape of the ground-truth maps (see Fig. \ref{fig:res:gt_and_masks}).

Later in this paper, we report empirically obtained optimal values for $\gamma$ on a benchmark dataset (see Fig.~\ref{fig:res:diffSigma}).

\subsection{Proposed Functional Map Energy}
In summary, we propose the following energy to compute a functional map for a pair of shapes, with a set of given descriptors:
\begin{equation}
\label{eq:energy:complete}
    E\big(\matr{C}_{12}\big) = \alpha_1E_{\text{desc}} + \alpha_2 E_{\text{mult}} + \alpha_3 E_{\text{orient}} + \alpha_4 E_{\text{resolvent}}
\end{equation}
where $E_{\text{desc}}$ is the descriptor-preserving term defined in Sec.~\ref{sec:background}, $E_{\text{mult}}$ is the descriptor-commutativity term defined as $E_{\text{mult}} = \sum_{i} \big\Vert \matr{C}_{12}\matr{D}_{1i}^{\text{mult}} - \matr{D}_{2i}^{\text{mult}}\matr{C}_{12}\big\Vert^2$ introduced in~\cite{nogneng2017informative}, $E_{\text{orient}}$ is the orientation-preserving term introduced in~\cite{ren2018continuous}. $E_{\text{resolvent}}$ is the our new term introduced in Eq. \eqref{eq:res_penalty} and discussed above. As mentioned above, when functional maps are expressed in the Laplace-Beltrami eigenbasis, this term can be  written via a penalty using the mask matrix given in Eq. \eqref{eq:mtd:res_mask} .

Before proceeding to a more extensive evaluation of the proposed energy, in Fig.~\ref{fig:mtd:example} above, we provide an example of a functional map obtained using this energy and compare it to one obtained using the standard Laplacian commutator regularizer. Fig.~\ref{fig:mtd:example} also shows the resolvent mask with $\gamma = 0.5$ (second row), compared to the standard mask (first row).
The rightmost column shows the quality of pointwise maps, recovered from the functional maps (shown in the second rightmost column) using both our and the standard Laplacian regularizers.


%% file: figures_tex/SphereTorusSpectra.tex
%
%
\definecolor{mycolor1}{rgb}{0.00000,0.44700,0.74100}%
\definecolor{mycolor2}{rgb}{0.85000,0.32500,0.09800}%
\definecolor{mycolor3}{rgb}{0.92900,0.69400,0.12500}%
\pgfplotsset{
compat=1.11,
legend image code/.code={
\draw[mark repeat=2,mark phase=2]
plot coordinates {
(0cm,0cm)
(0.15cm,0cm)        
(0.3cm,0cm)         
};%
}
}
\begin{tikzpicture}

\begin{axis}[%
width=0.6\linewidth,
height=0.4\linewidth,
at={(0.772in,0.516in)},
scale only axis,
xmin=0,
xmax=100,
xlabel style={font=\color{white!15!black}},
xlabel={$k$},
ymin=0,
ymax=1400,
ylabel style={font=\color{white!15!black}},
ylabel={$\lambda_k$},
xmajorgrids,
ymajorgrids,
every x tick label/.append style={font=\color{black}, font=\tiny},
every y tick label/.append style={font=\color{black}, font=\tiny},
axis background/.style={fill=white},
legend style={at={(0.03,0.97)}, anchor=north west, legend cell align=left, align=left, draw=white!15!black}
]
\addlegendimage{only marks, mark=*,  mark options={solid, mycolor1}}
\addlegendimage{only marks, mark=*,  mark options={solid, mycolor2}}
\addlegendimage{only marks, mark=*,  mark options={solid, mycolor3}}
\addplot [color=mycolor1, draw=none, mark size=0.8pt, mark=*, mark options={solid, mycolor1}]
  table[row sep=crcr]{%
1	0\\
2	39.4784176043574\\
3	39.4784176043574\\
4	39.4784176043574\\
5	39.4784176043574\\
6	78.9568352087149\\
7	78.9568352087149\\
8	78.9568352087149\\
9	78.9568352087149\\
10	157.91367041743\\
11	157.91367041743\\
12	157.91367041743\\
13	157.91367041743\\
14	197.392088021787\\
15	197.392088021787\\
16	197.392088021787\\
17	197.392088021787\\
18	197.392088021787\\
19	197.392088021787\\
20	197.392088021787\\
21	197.392088021787\\
22	315.827340834859\\
23	315.827340834859\\
24	315.827340834859\\
25	315.827340834859\\
26	355.305758439217\\
27	355.305758439217\\
28	355.305758439217\\
29	355.305758439217\\
30	394.784176043574\\
31	394.784176043574\\
32	394.784176043574\\
33	394.784176043574\\
34	394.784176043574\\
35	394.784176043574\\
36	394.784176043574\\
37	394.784176043574\\
38	513.219428856647\\
39	513.219428856647\\
40	513.219428856647\\
41	513.219428856647\\
42	513.219428856647\\
43	513.219428856647\\
44	513.219428856647\\
45	513.219428856647\\
46	631.654681669719\\
47	631.654681669719\\
48	631.654681669719\\
49	631.654681669719\\
50	671.133099274076\\
51	671.133099274076\\
52	671.133099274076\\
53	671.133099274076\\
54	671.133099274076\\
55	671.133099274076\\
56	671.133099274076\\
57	671.133099274076\\
58	710.611516878434\\
59	710.611516878434\\
60	710.611516878434\\
61	710.611516878434\\
62	789.568352087149\\
63	789.568352087149\\
64	789.568352087149\\
65	789.568352087149\\
66	789.568352087149\\
67	789.568352087149\\
68	789.568352087149\\
69	789.568352087149\\
70	986.960440108936\\
71	986.960440108936\\
72	986.960440108936\\
73	986.960440108936\\
74	986.960440108936\\
75	986.960440108936\\
76	986.960440108936\\
77	986.960440108936\\
78	986.960440108936\\
79	986.960440108936\\
80	986.960440108936\\
81	986.960440108936\\
82	1026.43885771329\\
83	1026.43885771329\\
84	1026.43885771329\\
85	1026.43885771329\\
86	1026.43885771329\\
87	1026.43885771329\\
88	1026.43885771329\\
89	1026.43885771329\\
90	1144.87411052637\\
91	1144.87411052637\\
92	1144.87411052637\\
93	1144.87411052637\\
94	1144.87411052637\\
95	1144.87411052637\\
96	1144.87411052637\\
97	1144.87411052637\\
98	1263.30936333944\\
99	1263.30936333944\\
100	1263.30936333944\\
};
\addlegendentry{Torus}

\addplot [color=mycolor2, draw=none, mark size=0.8pt, mark=*, mark options={solid, mycolor2}]
  table[row sep=crcr]{%
1	0\\
2	25.1327412287183\\
3	25.1327412287183\\
4	25.1327412287183\\
5	75.398223686155\\
6	75.398223686155\\
7	75.398223686155\\
8	75.398223686155\\
9	75.398223686155\\
10	150.79644737231\\
11	150.79644737231\\
12	150.79644737231\\
13	150.79644737231\\
14	150.79644737231\\
15	150.79644737231\\
16	150.79644737231\\
17	251.327412287183\\
18	251.327412287183\\
19	251.327412287183\\
20	251.327412287183\\
21	251.327412287183\\
22	251.327412287183\\
23	251.327412287183\\
24	251.327412287183\\
25	251.327412287183\\
26	376.991118430775\\
27	376.991118430775\\
28	376.991118430775\\
29	376.991118430775\\
30	376.991118430775\\
31	376.991118430775\\
32	376.991118430775\\
33	376.991118430775\\
34	376.991118430775\\
35	376.991118430775\\
36	376.991118430775\\
37	527.787565803085\\
38	527.787565803085\\
39	527.787565803085\\
40	527.787565803085\\
41	527.787565803085\\
42	527.787565803085\\
43	527.787565803085\\
44	527.787565803085\\
45	527.787565803085\\
46	527.787565803085\\
47	527.787565803085\\
48	527.787565803085\\
49	527.787565803085\\
50	703.716754404114\\
51	703.716754404114\\
52	703.716754404114\\
53	703.716754404114\\
54	703.716754404114\\
55	703.716754404114\\
56	703.716754404114\\
57	703.716754404114\\
58	703.716754404114\\
59	703.716754404114\\
60	703.716754404114\\
61	703.716754404114\\
62	703.716754404114\\
63	703.716754404114\\
64	703.716754404114\\
65	904.77868423386\\
66	904.77868423386\\
67	904.77868423386\\
68	904.77868423386\\
69	904.77868423386\\
70	904.77868423386\\
71	904.77868423386\\
72	904.77868423386\\
73	904.77868423386\\
74	904.77868423386\\
75	904.77868423386\\
76	904.77868423386\\
77	904.77868423386\\
78	904.77868423386\\
79	904.77868423386\\
80	904.77868423386\\
81	904.77868423386\\
82	1130.97335529233\\
83	1130.97335529233\\
84	1130.97335529233\\
85	1130.97335529233\\
86	1130.97335529233\\
87	1130.97335529233\\
88	1130.97335529233\\
89	1130.97335529233\\
90	1130.97335529233\\
91	1130.97335529233\\
92	1130.97335529233\\
93	1130.97335529233\\
94	1130.97335529233\\
95	1130.97335529233\\
96	1130.97335529233\\
97	1130.97335529233\\
98	1130.97335529233\\
99	1130.97335529233\\
100	1130.97335529233\\
};
\addlegendentry{Sphere}

\addplot [color=mycolor3, line width=1.0pt]
  table[row sep=crcr]{%
1	12.5663706143592\\
2	25.1327412287183\\
3	37.6991118430775\\
4	50.2654824574367\\
5	62.8318530717959\\
6	75.398223686155\\
7	87.9645943005142\\
8	100.530964914873\\
9	113.097335529233\\
10	125.663706143592\\
11	138.230076757951\\
12	150.79644737231\\
13	163.362817986669\\
14	175.929188601028\\
15	188.495559215388\\
16	201.061929829747\\
17	213.628300444106\\
18	226.194671058465\\
19	238.761041672824\\
20	251.327412287183\\
21	263.893782901543\\
22	276.460153515902\\
23	289.026524130261\\
24	301.59289474462\\
25	314.159265358979\\
26	326.725635973338\\
27	339.292006587698\\
28	351.858377202057\\
29	364.424747816416\\
30	376.991118430775\\
31	389.557489045134\\
32	402.123859659494\\
33	414.690230273853\\
34	427.256600888212\\
35	439.822971502571\\
36	452.38934211693\\
37	464.955712731289\\
38	477.522083345649\\
39	490.088453960008\\
40	502.654824574367\\
41	515.221195188726\\
42	527.787565803085\\
43	540.353936417444\\
44	552.920307031804\\
45	565.486677646163\\
46	578.053048260522\\
47	590.619418874881\\
48	603.18578948924\\
49	615.752160103599\\
50	628.318530717959\\
51	640.884901332318\\
52	653.451271946677\\
53	666.017642561036\\
54	678.584013175395\\
55	691.150383789754\\
56	703.716754404114\\
57	716.283125018473\\
58	728.849495632832\\
59	741.415866247191\\
60	753.98223686155\\
61	766.54860747591\\
62	779.114978090269\\
63	791.681348704628\\
64	804.247719318987\\
65	816.814089933346\\
66	829.380460547705\\
67	841.946831162065\\
68	854.513201776424\\
69	867.079572390783\\
70	879.645943005142\\
71	892.212313619501\\
72	904.77868423386\\
73	917.34505484822\\
74	929.911425462579\\
75	942.477796076938\\
76	955.044166691297\\
77	967.610537305656\\
78	980.176907920015\\
79	992.743278534375\\
80	1005.30964914873\\
81	1017.87601976309\\
82	1030.44239037745\\
83	1043.00876099181\\
84	1055.57513160617\\
85	1068.14150222053\\
86	1080.70787283489\\
87	1093.27424344925\\
88	1105.84061406361\\
89	1118.40698467797\\
90	1130.97335529233\\
91	1143.53972590668\\
92	1156.10609652104\\
93	1168.6724671354\\
94	1181.23883774976\\
95	1193.80520836412\\
96	1206.37157897848\\
97	1218.93794959284\\
98	1231.5043202072\\
99	1244.07069082156\\
100	1256.63706143592\\
};
\addlegendentry{Weyl Estimate}

\end{axis}
\end{tikzpicture}%

%% file: figures_tex/spectra_combined.tex
\definecolor{mycolor1}{rgb}{0.00000,0.44700,0.74100}%
\definecolor{mycolor2}{rgb}{0.85000,0.32500,0.09800}%
\pgfplotsset{scaled x ticks=false}
\pgfplotsset{
compat=1.11,
legend image code/.code={
\draw[mark repeat=2,mark phase=2]
plot coordinates {
(0cm,0cm)
(0.15cm,0cm)        
(0.3cm,0cm)         
};%
}
}
\pgfplotsset{major grid style={dashed, color=mycolor1}}
\pgfplotsset{minor grid style={dashed}}
\begin{tikzpicture}
\begin{axis}[
width=0.8\linewidth,
height=0.5\linewidth,
xmin = 0, xmax = 200,
ymin = 0, ymax = 0.95,
xminorgrids,
ymajorgrids,
minor xtick={0,50,100,150,200,300,400,500},
axis y line* = left, 
every x tick label/.append style={font=\color{black}, font=\tiny},
every y tick/.style={mycolor1},
every y tick label/.append style={font=\tiny\color{mycolor1}},
xlabel = $k$,
ylabel = {\textbf{\textcolor{mycolor1}{$\left\| \Delta_{1} - \Delta_{2} \right\|^{2}$}}}]
\addplot [color=mycolor1, line width=1.5pt, forget plot]
  table[row sep=crcr]{%
1	0\\
11	0.00129051138572211\\
21	0.0145489841883741\\
31	0.026551341223116\\
41	0.0391593688787657\\
51	0.0680611557954472\\
61	0.0721717584067752\\
71	0.136074913675638\\
81	0.178393524420097\\
91	0.233411614800431\\
101	0.275929722547108\\
111	0.285862373263507\\
121	0.318572173344735\\
131	0.352372404719342\\
141	0.393192715331743\\
151	0.442350364487649\\
161	0.443503691175125\\
171	0.533932381424256\\
181	0.676500404562021\\
191	0.731655994688857\\
201	0.879096780652419\\
211	0.912341954574395\\
221	0.929411686374682\\
231	1.10072375398098\\
241	1.16801642609089\\
251	1.31101251268708\\
261	1.54693831083222\\
271	1.57132613755329\\
281	1.60882993171213\\
291	1.78441717554811\\
301	1.89271970095496\\
311	1.9336360814387\\
321	2.02370499694806\\
331	2.26048517905683\\
341	2.37064651238419\\
351	2.40423908566473\\
361	2.7035989103248\\
371	2.8913923076957\\
381	2.92623595898531\\
391	3.0262198134099\\
401	3.24357950056958\\
411	3.44984382433861\\
421	3.58282452002651\\
431	3.64590397308569\\
441	4.12530084977979\\
451	4.16670876433906\\
461	4.19593568345372\\
471	4.20550437306758\\
481	4.23666158090296\\
491	4.68991441495585\\
501	4.79115555599688\\
511	4.81582778706538\\
521	5.1122748111106\\
531	5.6678215783255\\
541	5.985892882828\\
551	6.1560598335763\\
561	6.21144754014584\\
571	6.4101441733746\\
581	6.7842853498473\\
591	7.1342011549261\\
601	7.23587549622059\\
611	7.3559522763662\\
621	7.68762600759269\\
631	8.28410237003425\\
641	8.56481641143904\\
651	8.60012325623878\\
661	8.62602167870734\\
671	8.78305312301675\\
681	9.33991976817448\\
691	9.39870745398813\\
701	9.44332201164487\\
711	9.47467086104665\\
721	9.63283455025229\\
731	10.1842003462367\\
741	10.7863610960845\\
751	11.1312969966578\\
761	11.1392476681359\\
771	11.2644177327028\\
781	11.7177446903587\\
791	12.464430362414\\
801	12.6366704992272\\
811	12.6588430717346\\
821	12.6817743728067\\
831	12.8152187667259\\
841	13.2319364368177\\
851	14.1562850192012\\
861	14.4689460250302\\
871	14.5360069696155\\
881	14.6007553376464\\
891	15.4287598255345\\
901	16.7575159833113\\
911	17.1123902481673\\
921	17.388698979415\\
931	17.5058133850443\\
941	17.5237296730449\\
951	17.7113182930044\\
961	18.4973856777065\\
971	19.5084681363479\\
981	19.7543075761236\\
991	19.7815880064666\\
1001	19.8606601254901\\
1011	20.3624081973454\\
1021	21.1955590616545\\
1031	22.2931165859828\\
1041	23.0523393842217\\
1051	23.1305947062002\\
1061	23.1485597290322\\
1071	23.2048989743235\\
1081	23.3507619755498\\
1091	24.0317807409746\\
1101	24.7924463172034\\
1111	24.929143144395\\
1121	24.9686765107886\\
1131	25.104697858544\\
1141	25.4839836298206\\
1151	26.0867333437381\\
1161	27.0869860744979\\
1171	27.8413062157914\\
1181	28.4247959377401\\
1191	28.574441691982\\
1201	28.6056885537642\\
1211	29.1776700473456\\
1221	30.0821216753524\\
1231	30.908442730914\\
1241	31.2660819945232\\
1251	31.3375172790372\\
1261	31.3442507950644\\
1271	31.3945566583615\\
1281	31.5319540610502\\
1291	32.1789216004002\\
1301	33.5237574731663\\
1311	34.6621022856911\\
1321	34.8565385486148\\
1331	34.8585367789877\\
1341	34.8685230031604\\
1351	34.8785092273332\\
1361	35.5148994019049\\
1371	36.7256391438132\\
1381	37.0493704855713\\
1391	37.3021065152293\\
1401	37.3652941085868\\
1411	37.4029957656439\\
1421	37.4542739864006\\
1431	37.594201035982\\
1441	38.3315937129624\\
1451	39.7981752729133\\
1461	40.7071695875525\\
1471	41.1336296449038\\
1481	41.306874796425\\
1491	41.6910659731677\\
1501	42.5293039421369\\
1511	43.4819725596485\\
1521	44.9626762713299\\
1531	46.268464735013\\
1541	47.1025302715065\\
1551	47.4638819050804\\
1561	47.5878856114646\\
1571	47.7066866952668\\
1581	47.8220193640143\\
1591	48.5170959344501\\
1601	49.5530283284558\\
1611	50.7347686035265\\
1621	51.1284872111095\\
1631	51.1686595681242\\
1641	51.1849535916792\\
1651	51.2798930314905\\
1661	51.8713058036246\\
1671	52.602902742299\\
1681	53.7784896805946\\
1691	55.273878143361\\
1701	55.871408292898\\
1711	56.0237710355551\\
1721	56.1194209030795\\
1731	56.1577793015556\\
1741	56.1637305122001\\
1751	56.5866383781331\\
1761	57.7991879308367\\
1771	59.0950494027615\\
1781	60.0459342104146\\
1791	60.317417963013\\
1801	60.3554449096206\\
1811	60.451261205105\\
1821	60.6762132990916\\
1831	61.0448922419551\\
1841	61.9345954814487\\
1851	63.317319652397\\
1861	65.9590128411025\\
1871	66.9454179473415\\
1881	67.5981392778097\\
1891	67.759394546428\\
1901	67.8218219980424\\
1911	68.4095352499214\\
1921	69.1340874914922\\
1931	70.1841055848692\\
1941	71.9058936098149\\
1951	72.1754253142982\\
1961	72.4449570187814\\
1971	72.5962269165954\\
1981	72.6321228048818\\
1991	72.64746899031\\
};
\end{axis}
\begin{axis}[
width=0.8\linewidth,
height=0.5\linewidth,
     xmin = 0, xmax = 1500,
     ymin = 0, ymax = 0.52,
     hide x axis,
     hide y axis]
     \addplot [color=mycolor2, line width=1.5pt, forget plot]
  table[row sep=crcr]{%
1	0\\
11	0.00128423527165796\\
21	0.0138317126601529\\
31	0.0246065337501504\\
41	0.0344224677269446\\
51	0.054020075566186\\
61	0.0564648810101053\\
71	0.0873076366288413\\
81	0.104677221166008\\
91	0.123066353541208\\
101	0.134443566458105\\
111	0.1365424600453\\
121	0.142924325440007\\
131	0.147757850879005\\
141	0.152885596618246\\
151	0.158353766187226\\
161	0.158452104272914\\
171	0.164878445590922\\
181	0.17378246468664\\
191	0.176565216064563\\
201	0.183313986175051\\
211	0.184523192832404\\
221	0.185117283816775\\
231	0.189933846625974\\
241	0.1915285188861\\
251	0.194454534361383\\
261	0.198732668556403\\
271	0.19909652045925\\
281	0.199629504633693\\
291	0.201730617173992\\
301	0.20291299988457\\
311	0.203283724681235\\
321	0.204090574177411\\
331	0.205818187222561\\
341	0.206560869181086\\
351	0.206759282999454\\
361	0.208471149514169\\
371	0.20933311593851\\
381	0.209490289267956\\
391	0.2098933331919\\
401	0.210716894202493\\
411	0.211367050144928\\
421	0.211778729805071\\
431	0.211957204736612\\
441	0.213195952536952\\
451	0.213282892314158\\
461	0.213344137106804\\
471	0.21336369013556\\
481	0.21342583981797\\
491	0.21417606729142\\
501	0.214329171355123\\
511	0.214363947391112\\
521	0.214757541857828\\
531	0.215474002829973\\
541	0.215823428919617\\
551	0.216008802405134\\
561	0.216065286585921\\
571	0.216260137645759\\
581	0.216584872390676\\
591	0.216862453389493\\
601	0.216941123208345\\
611	0.217027270884066\\
621	0.217260445484451\\
631	0.217639526654321\\
641	0.217801116917898\\
651	0.217820737075172\\
661	0.217835055643422\\
671	0.217916870460151\\
681	0.218182672674769\\
691	0.218206933161979\\
701	0.218225211523308\\
711	0.218237658021692\\
721	0.218298917175088\\
731	0.218499540410973\\
741	0.218693637861051\\
751	0.218803661430228\\
761	0.218806037302647\\
771	0.218842725481468\\
781	0.218970963693391\\
791	0.219162099721526\\
801	0.219203643148814\\
811	0.219208681814061\\
821	0.219213847569863\\
831	0.219243131957757\\
841	0.21933354776158\\
851	0.219506565143476\\
861	0.219563296502901\\
871	0.219575186569998\\
881	0.219586109181762\\
891	0.219719366405679\\
901	0.21993047583623\\
911	0.219979783096825\\
921	0.220018055433187\\
931	0.220034058547408\\
941	0.220036497754121\\
951	0.220060267943656\\
961	0.220159121884385\\
971	0.220270633438638\\
981	0.22029738111368\\
991	0.220300225941159\\
1001	0.220308201594532\\
1011	0.220357547602423\\
1021	0.220438869154532\\
1031	0.220536361683155\\
1041	0.220601030982278\\
1051	0.220607459299135\\
1061	0.22060889477552\\
1071	0.220613362981814\\
1081	0.220624865972119\\
1091	0.220674805394117\\
1101	0.220725590109751\\
1111	0.220734460870188\\
1121	0.220736939806931\\
1131	0.220745380755992\\
1141	0.220768560143655\\
1151	0.220805175835962\\
1161	0.220861729216294\\
1171	0.220901414831316\\
1181	0.220932003343551\\
1191	0.220939730816814\\
1201	0.220941314700215\\
1211	0.220968899071004\\
1221	0.221012100333916\\
1231	0.221049005443413\\
1241	0.221063790492132\\
1251	0.221066703323706\\
1261	0.221066974185506\\
1271	0.221068958199959\\
1281	0.22107435366881\\
1291	0.221099224001018\\
1301	0.221146903524336\\
1311	0.221185576908913\\
1321	0.221192000499298\\
1331	0.22119206430183\\
1341	0.221192383156415\\
1351	0.221192702010999\\
1361	0.221212284055225\\
1371	0.221248716655634\\
1381	0.221257312540748\\
1391	0.22126401101925\\
1401	0.221265664450644\\
1411	0.221266645252204\\
1421	0.22126795801816\\
1431	0.221271512822949\\
1441	0.221289834362761\\
1451	0.221323191003575\\
1461	0.221343090692663\\
1471	0.221352301858875\\
1481	0.221356006360166\\
1491	0.22136383404917\\
1501	0.221380752813791\\
1511	0.221399929790487\\
1521	0.221429425940354\\
1531	0.221452833356035\\
1541	0.221467659378158\\
1551	0.221474000701677\\
1561	0.221476152321545\\
1571	0.22147818042131\\
1581	0.221480126174338\\
1591	0.221491622529831\\
1601	0.221508341002395\\
1611	0.221525654160543\\
1621	0.221531339513706\\
1631	0.221531906878232\\
1641	0.221532136557226\\
1651	0.221533432342135\\
1661	0.221541478796138\\
1671	0.221551390744131\\
1681	0.221567196702829\\
1691	0.221585304334179\\
1701	0.221592407049335\\
1711	0.221594190812697\\
1721	0.221595307039365\\
1731	0.221595754666451\\
1741	0.221595823113076\\
1751	0.221600581307767\\
1761	0.221613982233759\\
1771	0.221627418534723\\
1781	0.221636810590712\\
1791	0.221639454045783\\
1801	0.221639814313817\\
1811	0.221640718744165\\
1821	0.221642828781958\\
1831	0.221646264150783\\
1841	0.221654478455926\\
1851	0.22166674666938\\
1861	0.221688875197804\\
1871	0.221696961145072\\
1881	0.221702284351614\\
1891	0.221703588123514\\
1901	0.221704078764381\\
1911	0.221708619699286\\
1921	0.221714199213072\\
1931	0.221722242404352\\
1941	0.221734707090668\\
1951	0.221736519114805\\
1961	0.221738331138942\\
1971	0.22173934444575\\
1981	0.221739583345108\\
1991	0.221739685210745\\
};
\end{axis}
\pgfplotsset{every axis y label/.append style={rotate=180,yshift=0cm}}
\pgfplotsset{major grid style={dashed, mycolor2}}
\begin{axis}[
width=0.8\linewidth,
height=0.5\linewidth,
every x tick label/.append style={font=\color{black}, font=\tiny},
every y tick label/.append style={font=\tiny\color{mycolor2}},
every y tick/.style={mycolor2},
xmin=0, xmax=200,
ymin=0, ymax=0.6,
hide x axis,
axis y line*=right,
ymajorgrids,
ylabel={\textcolor{mycolor2}{$\left\| R(\Delta_{1}) - R(\Delta_{2}) \right\|^{2}$}}
]
\end{axis}
\end{tikzpicture}

%% file: sec_results.tex
\begin{figure*}[!ht]
\centering
\input{figures_tex/diff_fMapSize_normalized_faust.tex}
\input{figures_tex/diff_fMapSize_normalized_faust_direct.tex}
\input{figures_tex/diff_fMapSize_normalized_tosca.tex}
\input{figures_tex/diff_fMapSize_normalized_tosca_direct.tex}
\caption{Changing the functional map size. We randomly select 50 FAUST non-isometric pairs and 50 TOSCA isometric pairs. For each of the pairs, we only use one pair of WKS \cite{aubry2011wave} descriptors. We then optimize for a \hl{functional} map with different size ranging from 20 to 250 using different Laplacian mask terms. The solid lines represent the initialization with different masks and the dashed lines are the ICP-refined results. We report the per-vertex and the direct error measure. We can see that the proposed mask is much more stable than the standard and the slanted mask as the functional map size increases. }
\label{fig:res:diff_fmapsize}
\vspace{-0.5cm}
\end{figure*}
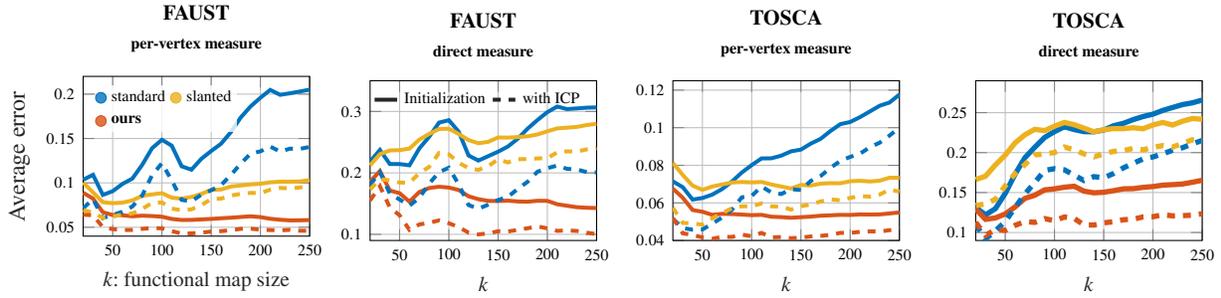

\section{Results}
\label{sec:results}
We tested the proposed approach using a MATLAB-based implementation, which we adapt from the state-of-the-art functional map approach in \cite{ren2018continuous}. Here, we first describe the benchmark datasets and the baseline methods. 

\myparagraph{Datasets}
We use the two datasets introduced in~\cite{ren2018continuous}. These consist of shapes from the FAUST~\cite{Bogo:CVPR:2014} and TOSCA~\cite{tosca} datasets, which were remeshed so that the shapes have different triangulations, and are no longer in one-to-one correspondence, making the matching more challenging and realistic. Specifically, we include 300 FAUST shape pairs and 284 TOSCA pairs for evaluation.

\myparagraph{Baselines}
To evaluate our new regularizer, we compare to the following two masks:
\begin{itemize}
    \item "standard": the standard Laplacian-commutativity mask, which is defined as $\matr{M}_{LB}(i,j) = \big(\Delta_2(i) - \Delta_1(j)\big)$ as discussed before. 
    \item "slanted": the heuristic slanted diagonal penalty mask proposed in~\cite{rodola2017partial}, which is defined as
    $$\matr{M}(i, j) = \exp{\left(-\eta \sqrt{i^2 + j^2}\right)}\left\Vert \frac{\vec{n}}{\Vert \vec{n} \Vert} \times \left((i,j)\tran - \vec{p}\right)\right\Vert$$
    where $\vec{p} = (1,1)\tran$, and $\vec{n} = (1,r/k)\tran$ is the line direction with slope $r/k$, where $r$ is the estimated rank of the functional map, and $k$ is the size of the square functional map. This weight matrix is originally applied to partial shape matching in~\cite{rodola2017partial}, and we use it as a mask matrix to regularize the functional map. In our tests, $\eta$ is set to the default value 0.03 as suggested in the original paper. This mask is illustrated on Fig. \ref{fig:res:gt_and_masks}. Note that it exhibits the desired funnel-like structure for the lower part of the spectrum, but not the upper part.
\end{itemize}

We also compare three different settings:
\begin{itemize}
    \item \textbf{Initialization}. The wave kernel signatures (WKS) \cite{aubry2011wave} are used to construct $E_{\text{desc}}, E_{\text{mult}}$, and $E_{\text{orient}}$ in Eq. \eqref{eq:energy:complete}. Then we optimize the functional map w.r.t. the energy defined in Eq.~\eqref{eq:energy:complete} with three different masks, namely, the standard, slanted, and our resolvent mask.
    \item \textbf{ICP refinement}. After initialization, we use ICP~\cite{ovsjanikov2012functional} to refine the computed functional maps.
    \item \textbf{BCICP refinement}. After initialization, we use the recently proposed BCICP algorithm~\cite{ren2018continuous} to refine the computed functional and pointwise maps, using the open-source implementation and default parameters provided by the authors.
\end{itemize}

\myparagraph{Measurements}
In our experiments, we measured the quality of the functional maps and the recovered point-wise maps:
\begin{itemize}
    \item \textbf{Point-wise maps}. Since most shapes contain left-right symmetries, which are indistinguishable for intrinsic methods, in each dataset, we considered both the ground-truth direct and symmetric correspondences. To measure the accuracy of a computed map, we used the following measures:
	\begin{itemize}
	\item \textbf{per-vertex error}: for each vertex we accept the ground-truth direct and symmetric correspondences and take the minimum as the error of this vertex. This measure reflects the accuracy of the map regardless of the symmetry.  
	\item \textbf{direct error}: we compute the average per-vertex error to the direct ground-truth correspondences only. This measure reflects both the accuracy and the smoothness of the map.
	\end{itemize} 
    \item \textbf{Functional maps}. We can evaluate a functional map by the quality of its recovered point-wise map, or by measuring the penalty from a given mask. Specifically, for a given functional map $\matr{C}$ and a  mask matrix $\matr{M}$, we can measure the total penalty as 
    $\sum_{ij} \left[\matr{M} \right]_{ij} \left[ \matr{C} \right]_{ij}^2$~. 
\end{itemize}

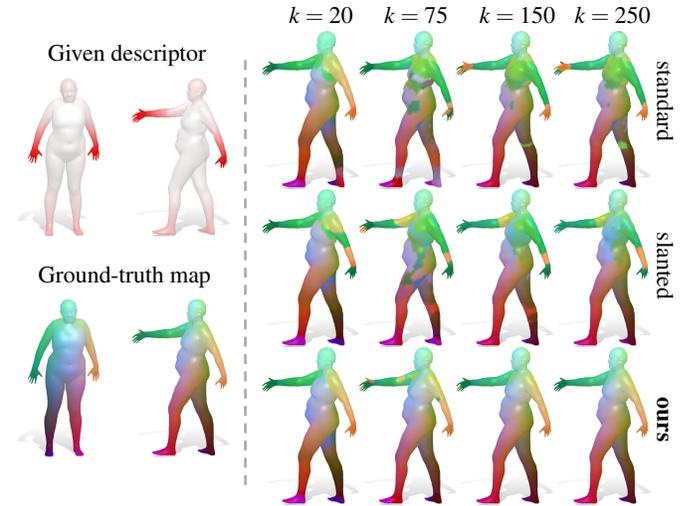
\begin{figure}[!t]
  \begin{overpic}
  [trim=0cm 0cm 0cm 0cm,clip,width=1\linewidth, grid=false]{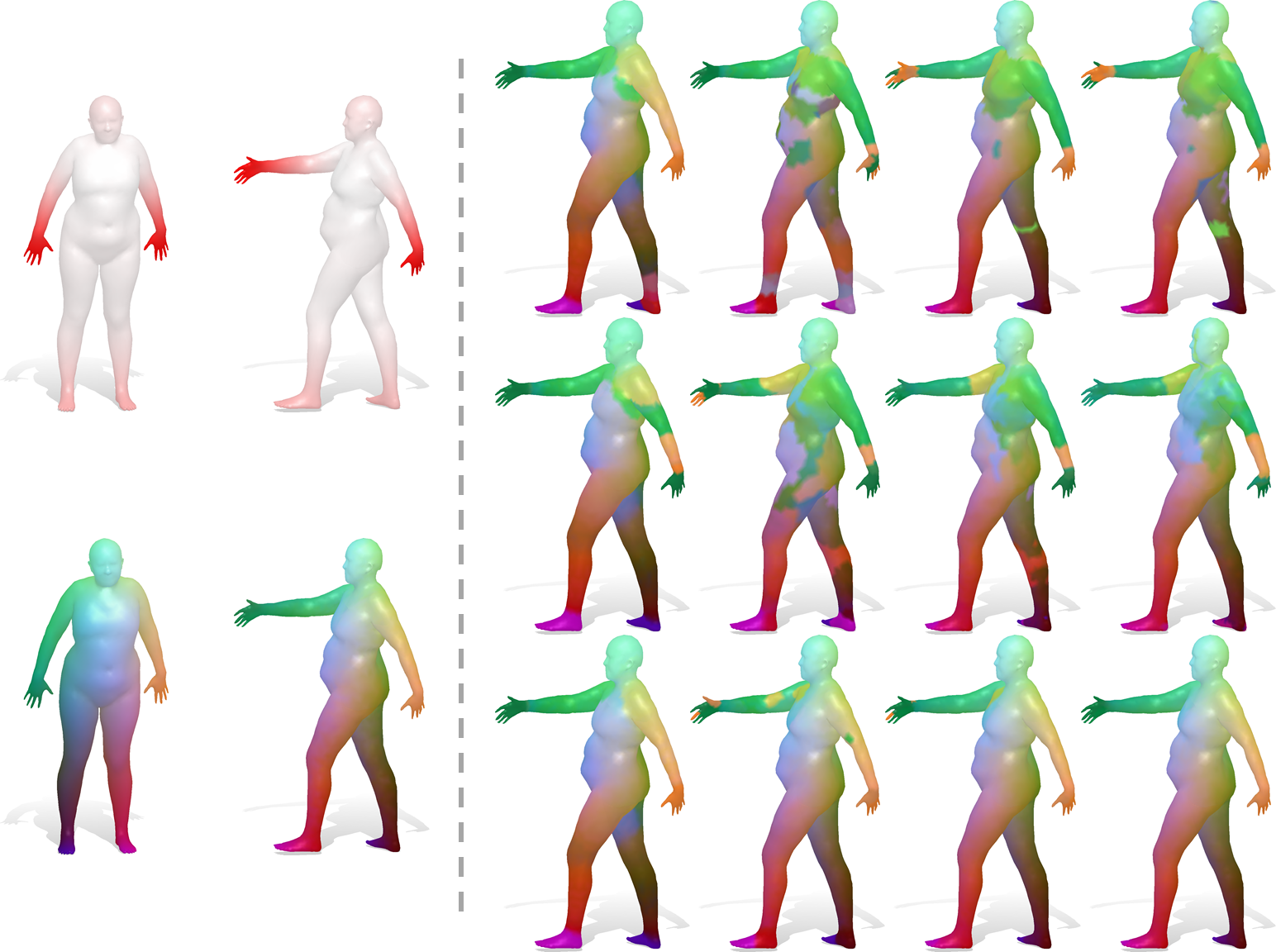}
  \put(5,70){Given descriptor}
  \put(3.5,35){Ground-truth map}
  \put(101,17){\rotatebox{-90}{\textbf{ours}}}
  \put(101,43){\rotatebox{-90}{slanted}}
  \put(101,70){\rotatebox{-90}{standard}}
  
  \put(43,76){$k=20$}
  \put(58,76){$k=75$}
  \put(73,76){$k=150$}
  \put(88,76){$k=250$}
  \end{overpic}
  \caption{Given one pair of WKS descriptors, as visualized on the left top, we use different Laplacian mask terms to optimize for a functional map with size $k$. The ground-truth map is visualized on the bottom left. We can see that our mask is much more stable over different size $k$.}
  \label{fig:res:diff_fmapsize:eg}
  \vspace{-0.5cm}
\end{figure}

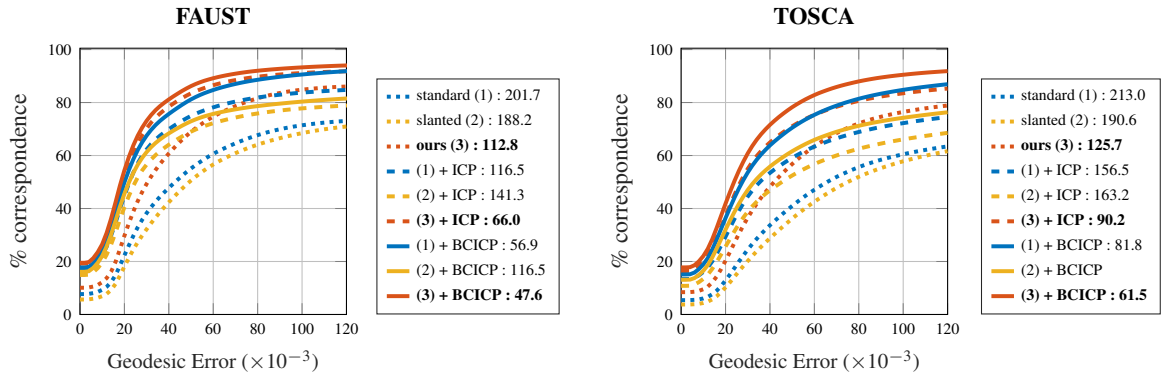
\begin{figure*}[!t]
\centering
\input{figures_tex/summaryCurve_norm_faust_combined.tex}
\input{figures_tex/summaryCurve_norm_tosca_new.tex}
\caption{For each dataset, we compare the quality of the functional maps between the standard mask (blue curves), the slanted mask (yellow curves), and our resolvent mask (red curves) with the same set of parameters, where a single pair of descriptors is used to optimize a 100-by-100 functional map. The comparison is made in three different settings: comparing the initialization directly (dotted lines), comparing the initialization with ICP refinement (dashed lines), and with the BCICP refinement (solid lines). 
Specifically, these curves are measured on 300 FAUST shape pairs and 284 TOSCA pairs. The average direct errors are reported in the legends.}
\label{fig:res:err_summary}
\end{figure*}

\subsection{Effect of functional map size}
For the functional maps pipeline, a set of corresponding descriptors is given as input. We then optimize 
a $k \times k$ functional map by minimizing an objective function based on Eq.  \eqref{eq:energy:complete}. 
Therefore, we have to solve for $k^2$ variables. If $k$ is a smaller value, e.g., $k < 50$, the optimization problem
is easier to solve since the number of variables is small. However, if $k$ is too small, the information that is encoded
into the optimized functional map is limited to the low-resolution of the spectrum of the shapes, and hence it will be hard
to transfer detailed information. On the other hand, if $k$ is too large, solving the optimization problem is potentially more time-consuming. Even more importantly, 
this optimization problem can become under-constrained when the number of variables exceeds the constraints stemming from the input descriptors.
In this case, we need effective regularizers to regularize the functional maps. In real experiments, the choice of the parameter $k$ is a key hyper parameter. 

To quantify the stability and the effectiveness of the proposed resolvent mask, we conduct the following test: for each of the test pairs, we only use one pair of corresponding WKS descriptors. We then fix this descriptor pair and optimize for functional maps with different sizes ranging from 20 to 250. We randomly select 50 pairs of FAUST and 50 pairs of TOSCA, and report the average error over the tested shape pairs w.r.t. different functional map size in Fig.~\ref{fig:res:diff_fmapsize}. We can see that the standard mask fails to regularize the functional map with a large size: the average per-vertex error is three or four times larger than ours. At the same time, the slanted mask has a better performance than the standard one in the per-vertex measure. However, the slanted mask has large direct errors, which suggests that the smoothness is not well preserved. 
We believe this is due to the fact that the orientation-preserving regularizer starts to fail to disambiguate the symmetry as the functional map size $k$ increases. In this case, the slanted mask cannot help the orientation-preserving operator, while our mask can strengthen the functionality of the orientation-preserving operator and leads to maps with much lower per-vertex and direct error.
In summary, even with limited constraints from a single pair of WKS descriptors, increasing the number of variables does not significantly affect the performance of our mask. Fig.~\ref{fig:res:diff_fmapsize:eg} shows an illustrative example.

This test shows that our mask is much more stable and can better regularize larger functional maps even in very challenging cases with little input information or constraints. Also, it suggests that with this new mask, we no longer need to tune the parameter $k$ as much as needed by the standard mask to achieve a better result.

\subsection{Evaluation on shape matching}

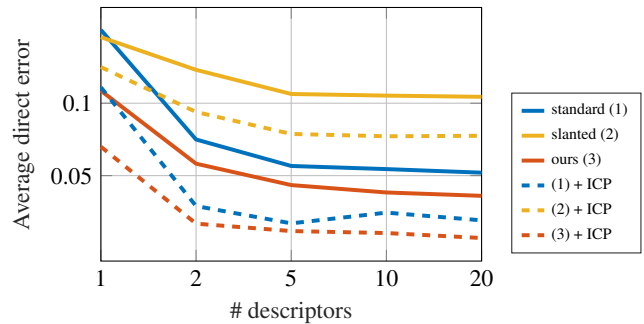
\begin{figure}
    \centering
    \input{figures_tex/diff_desc_faust_direct.tex}
    \caption{Changing the number of input descriptors. We use different numbers of descriptors to optimize a 100-by-100 functional map with different masks. The results are reported for 300 FAUST shape pairs. As the number of descriptors increases, the results improve for all the masks.}
    \label{fig:res:diff_desc}
\end{figure}

\begin{figure}[!t]
  \vspace{0.5cm}
  \centering
  \begin{overpic}
  [trim=0cm 0cm 0cm 0cm,clip,width=1\linewidth, grid=false]{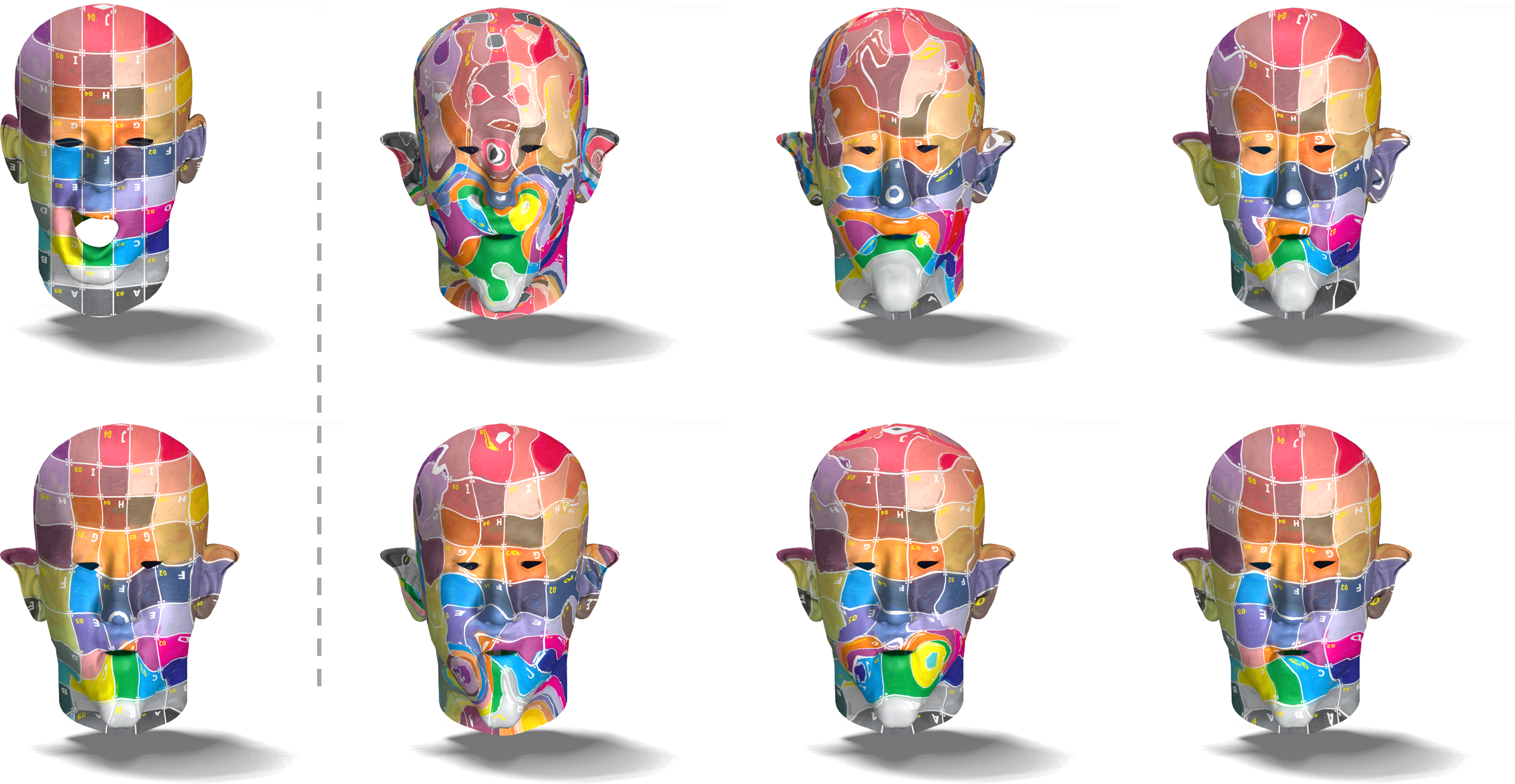}
  \put(1.8,52.5){Source}
  \put(25.5,52.5){Standard}
  \put(53,52.5){Slanted}
  \put(81,52.5){\textbf{Ours}}
  
  \put(0,25){Ground-truth}
  \put(96,50){\rotatebox{-90}{Initialization}}
  \put(96,17){\rotatebox{-90}{+ ICP}}
  
  \end{overpic}
  \caption{Example. Comparing the quality of the initial maps and after ICP refinement via texture transfer. (One pair of descriptors is used to optimize a 50-by-50 functional map)}
  \label{fig:res:eg:face}
\end{figure}
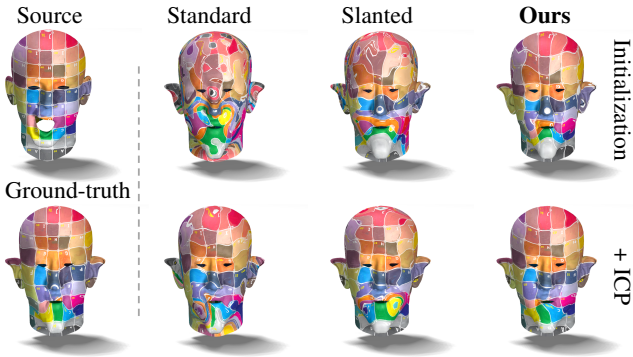

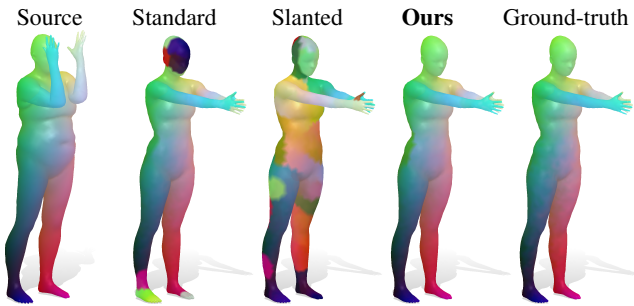
\begin{figure}[!t]
  \vspace{0.5cm}
  \centering
  \begin{overpic}
  [trim=0cm 0cm 0cm 0cm,clip,width=1\linewidth, grid=false]{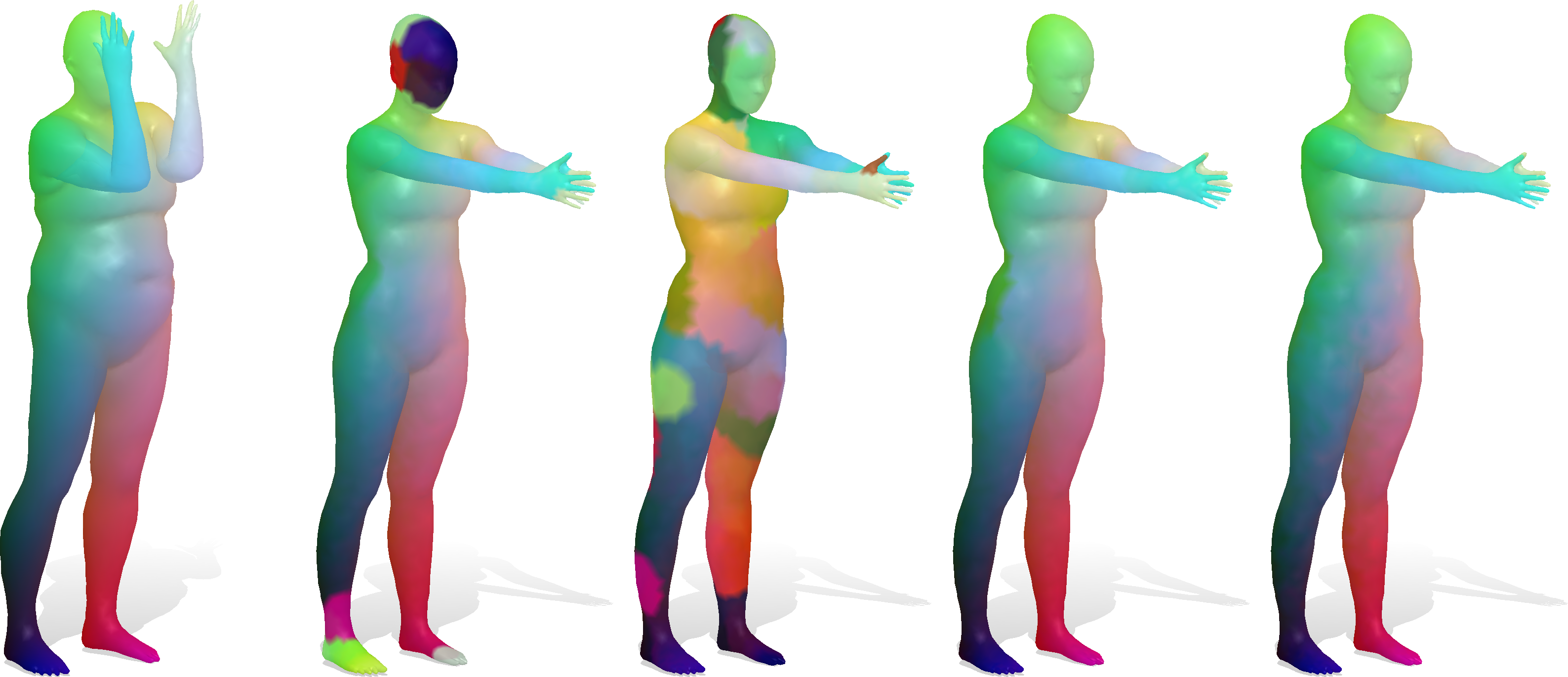}
  \put(1.8,44){Source}
  \put(20,44){Standard}
  \put(42,44){Slanted}
  \put(62.5,44){\textbf{Ours}}
  \put(78.5,44){Ground-truth}
  
  \end{overpic}
  \caption{Example. Comparing the quality of the ICP refined maps via color transfer. (One pair of descriptors is used to optimize a 100-by-100 functional map)}
  \label{fig:res:eg:faust}
  \vspace{-0.2cm}
\end{figure}

In this experiment, we compare our resolvent mask to the standard and the slanted mask on a larger FAUST and TOSCA dataset in three settings: the initialization, with ICP refinement, and with BCICP refinement. 
To make a fair comparison between different masks, the weights for different terms (the $\alpha_i$ in Eq.~\eqref{eq:energy:complete}) are fixed across different test pairs and different test masks.

In this test, for each test pair, we use one pair of WKS descriptors to optimize a 100-by-100 functional map. 
As reported in Fig.~\ref{fig:res:err_summary}, our mask leads to 40.4\%, 43.1\%, and 16.3\% improvement w.r.t the settings of the initialization, with ICP refinement, and with BCICP refinement respectively, over the best of the standard and the slanted mask on the FAUST dataset regarding the average direct error. Similarly, for the TOSCA dataset, ours achieves 34.2\%, 42.3\%, and 24.8\% improvement respectively over the best of the standard and the slanted mask.

Fig.~\ref{fig:res:diff_desc} shows a test on the FAUST dataset, where we use different number of carefully curated input
descriptors based on WKS \cite{aubry2011wave} using parameters from \cite{ren2018continuous} to compute 100-by-100
functional maps. As the number of input descriptors increases, more constraints are added to regularize the functional
map. Therefore, for all the masks, and all the test settings, the results improve.  Observe that when the number of
variables is small (as shown in Fig.~\ref{fig:res:diff_fmapsize}  for small $k$) or the number of
descriptor constraints is large (as shown in Fig.~\ref{fig:res:diff_desc} for large descriptor number), all 
masks perform well since the problem is well-constrained. However, when this is not the case, our resolvent mask can still
regularize the functional map better than the other two.  For completeness,  in
Appendix~\ref{appendix:descSize} we also include results with the BCICP refinement. We remark that this refinement is very computationally and memory intensive due, in part, to requiring all-pairs geodesic distances, but can, as such, improve upon even very poor quality maps.

\begin{figure}[!t]
  \vspace{0.5cm}
  \centering
  \begin{overpic}
  [trim=0cm 0cm 0cm 0cm,clip,width=1\linewidth, grid=false]{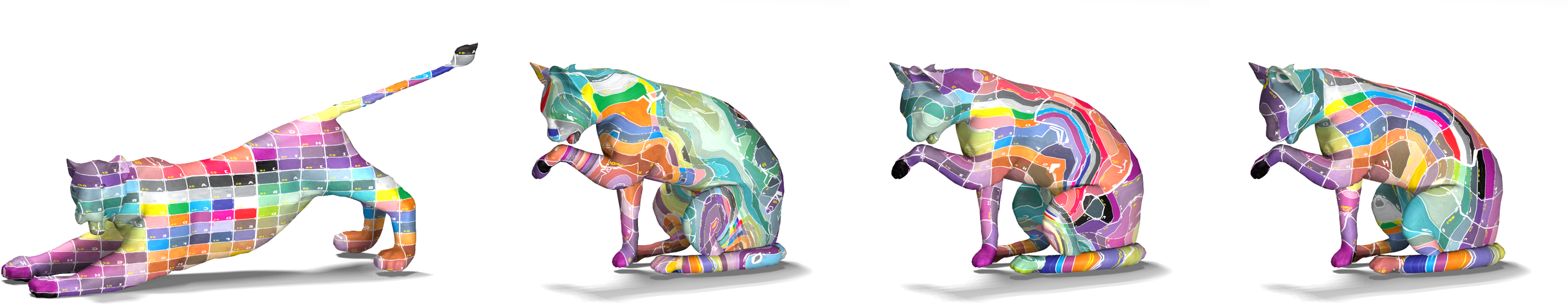}
  \put(9,17){Source}
  \put(35,17){Standard}
  \put(60,17){Slanted}
  \put(84,17){\textbf{Ours}}
  
  \end{overpic}
  \caption{Example. Here we show a challenging pair of a lion and a cat and the results are refined by BCICP. We can see that, when the initial maps are in a low quality and the BCICP refinement fails to improve the initial maps, our new mask still gives a more reasonable map. (Twenty pairs of descriptors are used to optimize a 50-by-50 functional map)}
  \label{fig:res:eg:cat_lion}
\end{figure}
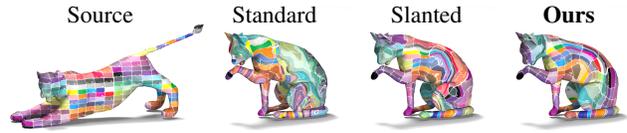

Fig.~\ref{fig:res:eg:face} shows a qualitative example of using a single descriptor pair to optimize for a 100$\times$100 functional map. The first row shows the quality of the initial maps with different masks, and the second rows the quality of the corresponding maps refined by ICP. We can see that our resolvent mask outperforms the other two masks. Also the quality of our initial map is close to the ICP refined map, which shows that with our mask, we do not rely on the post-processing refinement as much as the other two. Fig.~\ref{fig:res:eg:faust} shows another example of the results refined by ICP. Fig.~\ref{fig:res:eg:cat_lion} shows a challenging non-isometric example with results refined after BCICP.

\subsection{Analysis of the complex resolvent mask}
As shown in Fig.~\ref{fig:res:gt_and_masks}, the funnel-shape of our resolvent mask aligns well with the ground-truth functional map, which leads to a better performance over the standard and the slanted mask. To further analyze the properties of our mask, we also conduct the following experiments: We test the range of the parameter $\gamma$, and the relative weight between the complex and the real part to construct the resolvent mask. Moreover, we investigate the correlation between the mask penalty added on a functional map and the corresponding recovered point-wise map.

\subsubsection{Different parameters for the complex resolvent mask}
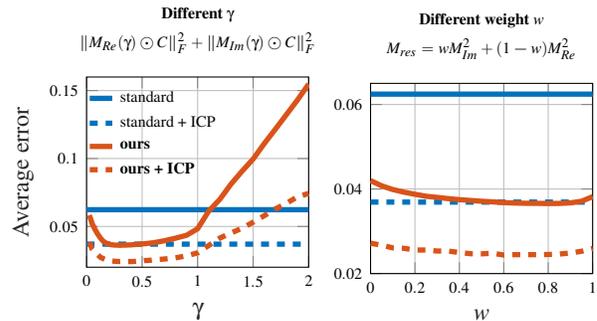
\begin{figure}[!t]
\centering
\input{figures_tex/diff_sigma_normalized_pervtx.tex}
\input{figures_tex/diff_weight_normalized_pervtx.tex}
\caption{Left: changing the parameter $\gamma$ (as defined in Eq.~\eqref{eq:res_penalty}); Right: changing the relative weight $w$ between the imaginary part and the real part (see the original definition in Eq.~\eqref{eq:mtd:res_mask}). The results are based on~50 TOSCA pairs.}
\label{fig:res:diffSigma}
\end{figure}

In this section, we empirically tune the parameters in our resolvent based mask. First, we explore different values of $\gamma$. Recall that $\gamma$ controls the funnel-like structure of the mask. Thus, it is expected that tuning $\gamma$ can influence the functional map quality.

On Fig. \ref{fig:res:diffSigma}, we report results for 100 FAUST pairs, where the result of the standard mask is colored blue, and ours is colored red. As is shown on the left, when $\gamma$ lies between~0 and~1, our mask always outperforms the standard mask over both the initialization (solid lines) and after the ICP refinement (the dashed lines).

Note that the case $\gamma = 1$ corresponds to the resolvent of the Laplacian. Thus, the fact that our mask outperforms the standard one for $\gamma = 1$ experimentally justifies the usage of the resolvent of the Laplacian, rather than the Laplacian itself. Note also that, as suggested in Sec. \ref{sec:mask_structure}, the resolvent mask performs poorly for $\gamma > 1$.

Finally, we explore the relative contribution of the real and imaginary parts of the resolvent mask in order to analyze the utility of the two components of our mask construction.  For this, we analyze the accuracy of the computed pointwise maps when allowing $M_{Re}$ and $M_{Im}$ to have different weights, $(1-w)$ and $w$ for $w \in [0,1]$, respectively. As shown in Fig.~\ref{fig:res:diffSigma} (right), the convex shape of the red curve (where the weight $w$ changes) demonstrates that both the real part and the imaginary part contribute to the improvement over the standard mask. Note that our mask with any convex combination between the real part and the imaginary parts outperforms the standard mask. In practice we always use the equal weight $w=0.5$.


\subsubsection{Correlation between the mask penalty and map accuracy}

\begin{figure}[!t]
\centering
\input{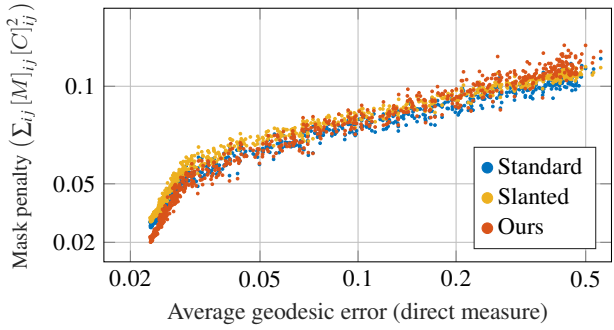}
\caption{Correlation between the mask penalty on a functional map and the accuracy of the corresponding point-wise map. For a TOSCA isometric pair, we sample 700 functional maps of size 50-by-50 with different quality (normalized to the same scale). We measure the penalty of different masks, as an indicator of the quality of the functional map. We also measure the average geodesic error (w.r.t. the direct measurement) of the recovered point-wise map of the corresponding functional map. The correlation between the functional map penalty and the point-wise map quality is visualized as a scatter plot with 700 samples. Compared to the standard and the slanted mask, our resolvent mask applies a smaller penalty to functional maps with good underlying point-wise maps and more heavily penalizes functional maps with bad underlying point-wise maps.}
\vspace{-0.5cm}
\label{fig:res:err_correlation}
\end{figure}

To justify that our resolvent mask is a better regularizer than the standard and the slanted ones, we also measure how the mask penalty relates to the accuracy of the corresponding point-wise map shown in Fig.~\ref{fig:res:err_correlation}. In this experiment, we generate 700 different point-wise maps with different levels of accuracy, then convert them to a functional map representation and measure the penalty added by different masks w.r.t. the average geodesic error computed from the pointwise maps. Each scatter point in Fig.~\ref{fig:res:err_correlation} shows such a test sample. 

We can observe that, compared to the standard and the slanted mask, the new mask induces a lower penalty on functional maps with a good quality (i.e., smaller average geodesic error), and penalizes a functional map with larger error more heavily. 
This further confirms that using our resolvent mask is more likely to produce a better functional, and ultimately
better pointwise map.

\subsection{Parameters}

\begin{figure}[!t]
    \centering
    \input{figures_tex/maskWeight_pervtx.tex}
    \input{figures_tex/maskWeight_direct.tex}
    \caption{Changing the relative weight $\alpha_4$ of the mask term. We choose 68 different values for $\alpha_4$ in the range of 0 and 10 for this test.}
    \label{fig:res:change_alpha4}
\end{figure}
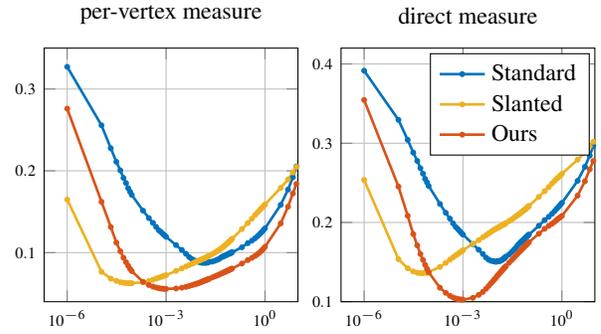

\begin{table}[!t]
\caption{Average error on 284 TOSCA pairs of different masks, each mask using its own optimal weight from Fig.~\ref{fig:res:change_alpha4}. Specifically, we set $\alpha_4$ to $10^{-2}$, $10^{-4}$, $10^{-3}$ for the standard, slanted, and our mask respectively.}
\label{tb:res:opti_mask_weight}
\centering
\footnotesize
\begin{tabular}{ccccccc}
\hline
\multirow{3}{*}{\begin{tabular}[c]{@{}c@{}}Mask\\ type\end{tabular}} & \multicolumn{6}{c}{Average error \small{$(\times 10^{-3})$}} \\ \cline{2-7} 
 & \multicolumn{3}{c|}{per-vertex measure} & \multicolumn{3}{c}{direct measure} \\ \cline{2-7} 
 & \footnotesize Ini & \footnotesize{+ ICP} & \multicolumn{1}{c|}{\footnotesize + BCICP} &\footnotesize Ini &\footnotesize + ICP &\footnotesize + BCICP \\ \hline
standard & 87.26 & 65.97 & \multicolumn{1}{c|}{41.22} & 178.8 & 130.8 & 76.87 \\
slanted & 64.66 & 50.76 & \multicolumn{1}{c|}{37.45} & 166.9 & 131.5 & 100.0 \\
\textbf{ours} & \textbf{55.52} & \textbf{43.82} & \multicolumn{1}{c|}{\textbf{32.48}} & \textbf{124.5} & \textbf{90.6} & \textbf{62.33} \\ \hline
\end{tabular}
\vspace{-0.2cm}
\end{table}

In our tests, we use $\gamma = 0.5$ and $w = 0.5$ to construct our resolvent mask, and use the default value $\eta = 0.03$ to construct the slanted mask. When comparing the three masks in the initialization setting, i.e., to optimize the energy defined in~\eqref{eq:energy:complete}, the weights $\alpha_i$ are set to the same values as reported in~\cite{ren2018continuous} \hl{but in a relative way (see Appendix~{\ref{appendix:descSize}} for more details)}. As for the comparison of the ICP and BCICP refinement settings, the same default parameters and the same number of iterations are used for different masks. 

Fig.~\ref{fig:res:change_alpha4} shows the results of the changing the weight of the mask term, i.e., $\alpha_4$, while keeping the rest weights fixed on 50 TOSCA shape pairs. We can see that, for all different choices of the weight, our resolvent mask is always better than the standard mask. When the weight is small enough, it seems the slanted mask can achieve a faster decrease of the error than ours. However, the range of the effective weight is smaller than ours besides that the construction is purely heuristic and lacks a theoretical justification.

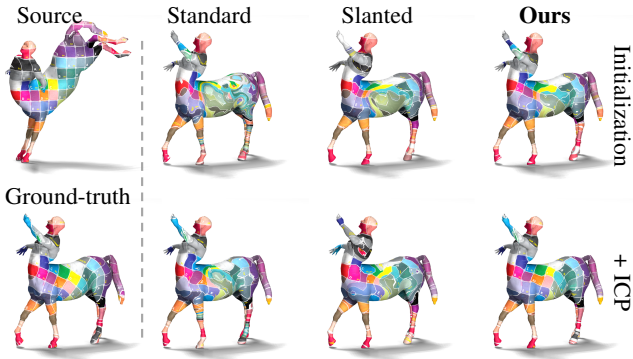
\begin{figure}[!t]
  \vspace{0.5cm}
  \centering
  \begin{overpic}
  [trim=0cm 0cm 0cm 0cm,clip,width=1\linewidth, grid=false]{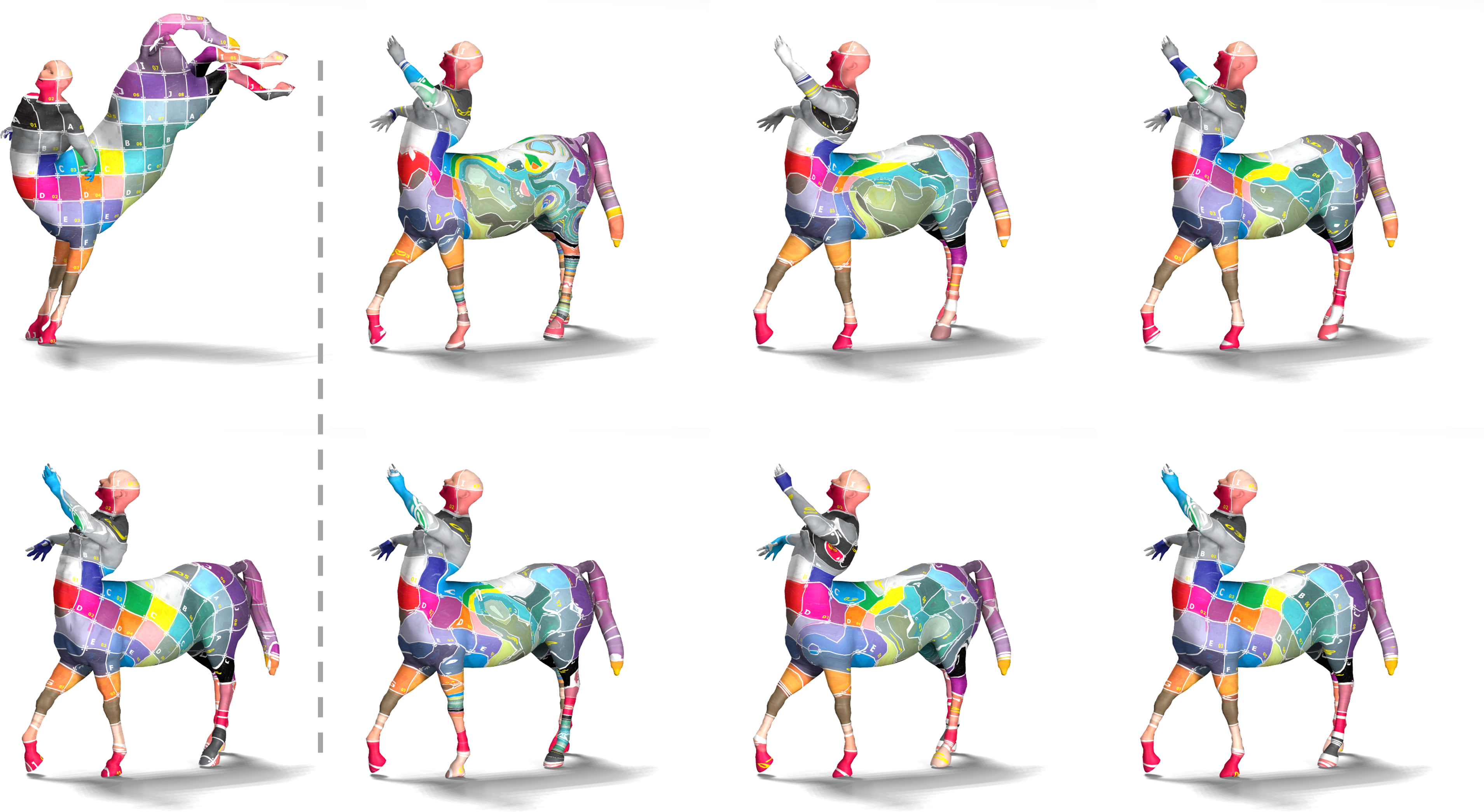}
  \put(1.8,53.5){Source}
  \put(25.5,53.5){Standard}
  \put(53,53.5){Slanted}
  \put(81,53.5){\textbf{Ours}}
  
  \put(0,25){Ground-truth}
  \put(96,50){\rotatebox{-90}{Initialization}}
  \put(96,17){\rotatebox{-90}{+ ICP}}
  \end{overpic}
  \caption{Example. Comparing the quality of the initial maps computed with different masks and after ICP refinement via texture transfer. Each mask uses its own optimal weight from Fig.~\ref{fig:res:change_alpha4}.}
  \label{fig:res:eg:centaur}
  \vspace{-0.2cm}
\end{figure}

On the other hand, Fig.~\ref{fig:res:change_alpha4} also suggests that different masks have their own preferred choice of the weight $\alpha_4$. Therefore, instead of fixing $\alpha_4$, we set it independently w.r.t. the optimal values on a subset reflected in Fig.~\ref{fig:res:change_alpha4}. Specifically, for the standard mask, we set $\alpha_4 = 10^{-2}$, for the slanted mask, we set $\alpha_4 = 10^{-4}$, and for ours, we set $\alpha_4 = 10^{-3}$. The corresponding average errors on the complete TOSCA dataset are reported in Table~\ref{tb:res:opti_mask_weight}. We can see that, even when the parameters are carefully tuned for the other two masks, our resolvent mask still achieves the best accuracy. Fig.~\ref{fig:res:eg:centaur} shows a qualitative example.

\subsection{Application to non-isometric shape pairs}
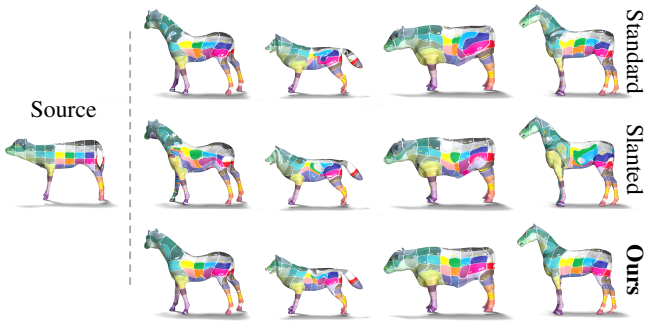
\begin{figure}[!th]
    \centering
    \begin{overpic}
    [trim=0cm 0cm 0cm 0cm,clip,width=1\columnwidth, grid=false]{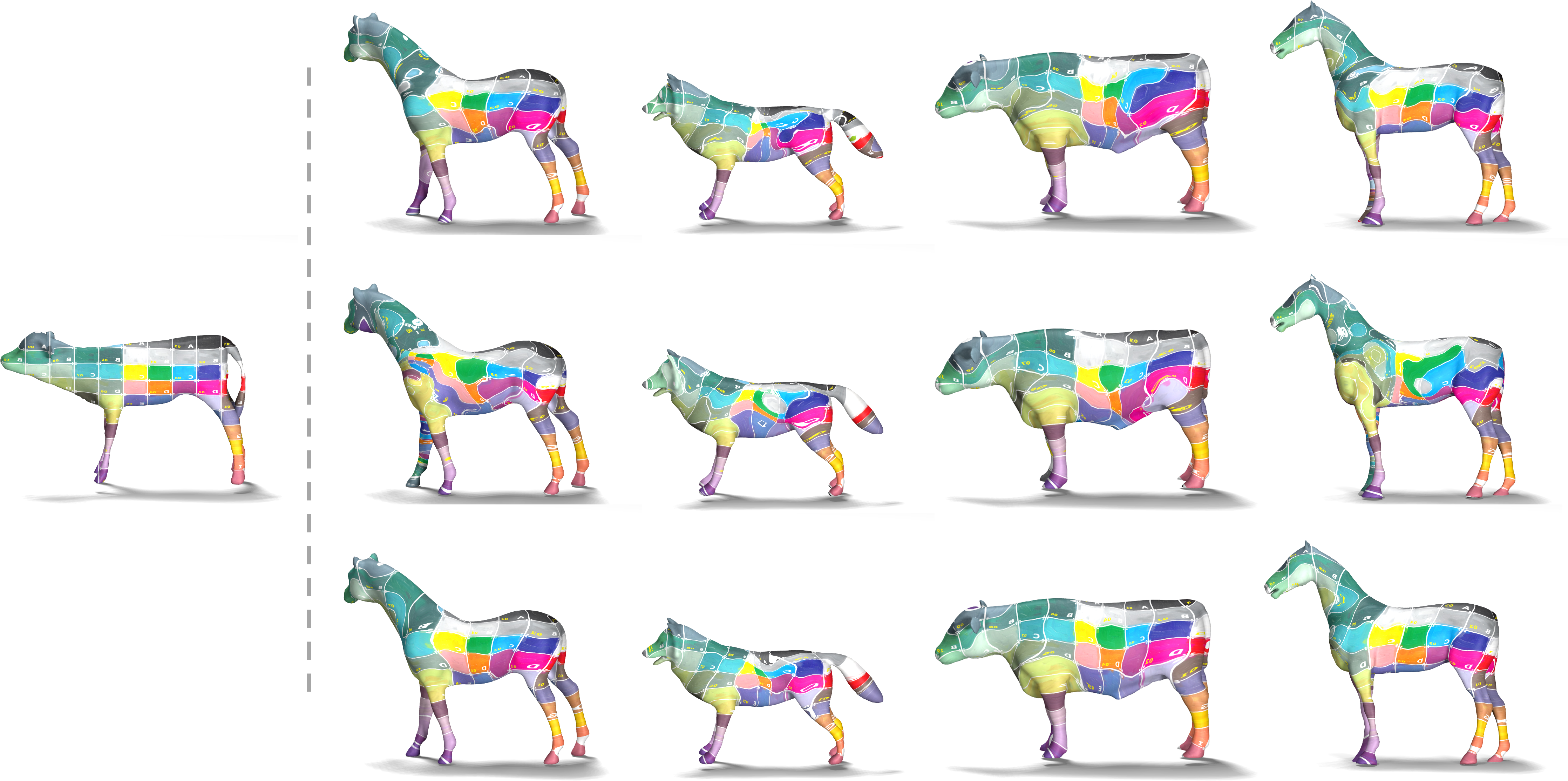}
    \put(4,32){Source}
    \put(98,49){\rotatebox{-90}{Standard}}
    \put(98,31){\rotatebox{-90}{Slanted}}
    \put(98,12){\rotatebox{-90}{\textbf{Ours}}}
    
  \end{overpic}
    \caption{Example. Comparing the quality of the initial maps computed with different masks of four SHREC shape pairs via texture transfer.}
    \label{fig:res:shrec:texture}
    \vspace{-0.2cm}
\end{figure}

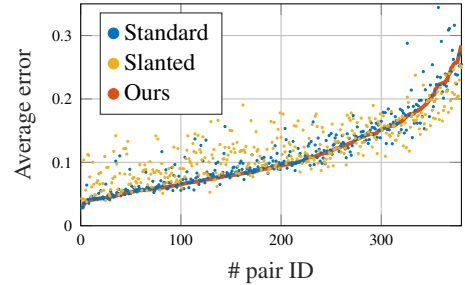
\begin{figure}[!t]
    \centering
    \input{figures_tex/res_shrec_aveErr.tex}
    \caption{For 380 SHREC shape pairs, we plot the average error for each pair with different masks, where the result of the standard, slanted and our mask is colored in blue, yellow, and red respectively. Therefore, in the case where the red curve is below the blue or the yellow points, our mask leads to a lower error for these pairs.}
    \label{fig:res:shrec:error}
\end{figure}

\begin{figure}[!t]
    \centering
    \begin{overpic}
    [trim=0cm 1cm 0cm 0cm,clip,width=1\columnwidth, grid=false]{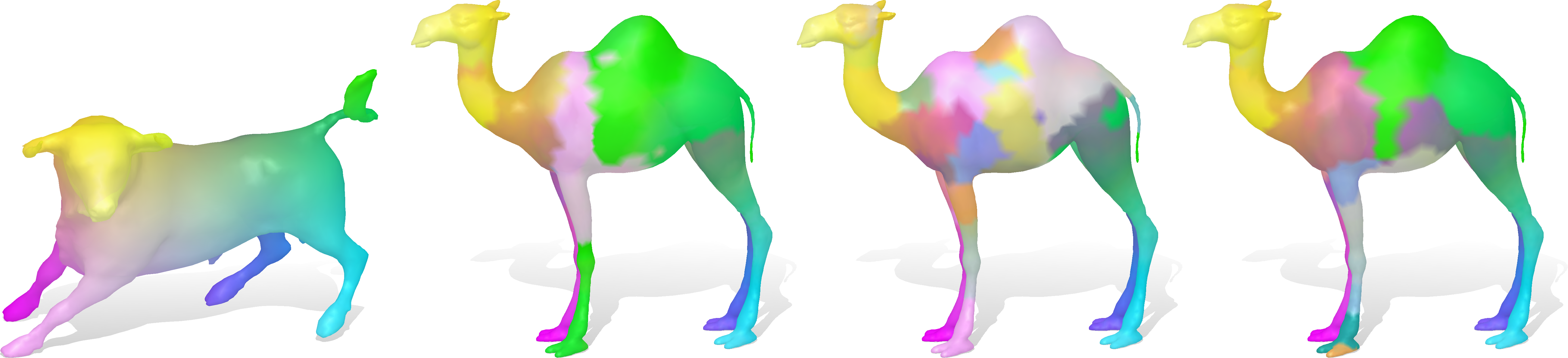}
    \put(7,24){Source}
    \put(33,24){Standard}
    \put(59,24){Slanted}
    \put(85,24){\textbf{Ours}}
  \end{overpic}
    \caption{Failure case. Here we show a challenging shape pair from SHREC, where all the three masks fail to produce a good map.}
    \label{fig:res:failure}
    \vspace{-0.2cm}
\end{figure}
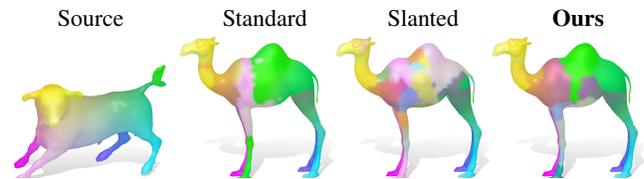

To show the usefulness of our resolvent mask on non-isometric shape pairs, we test the 20 FourLeg shapes from the SHREC 2007 dataset~\cite{giorgi2007shape}. 
Specifically, we use 10 pairs of descriptors constructed from 4 landmarks to optimize a 120-by-120 functional map using the standard, the slanted, and our resolvent mask.
For a fair comparison, we use the same parameters as the previous tests: we set $\gamma = 0.5$ for our resolvent mask, and set $\eta = 0.03$ for the slanted mask.
Fig.~\ref{fig:res:shrec:texture} shows a qualitative example. 
We can see that our resolvent mask gives the best initialization. 
We also measure the accuracy of the 380 point-wise maps among the 20 shapes on the given landmarks (21 ground-truth landmarks are given for each shape).
The average error for the standard, slanted, and our mask is 0.114, 0.122, and 0.107 respectively. 
Fig.~\ref{fig:res:shrec:error} reports the average error for each shape pair. 
We observe that resolvent mask gives a better result than the standard mask on 68.4\% out of 380 pairs, and outperforms the slanted mask on 69.5\% pairs.
The limited quantitative improvement of our mask over the other two masks is due to: 
(1) for the shape pairs where our mask significantly outperforms the other two, e.g., as shown in Fig.~\ref{fig:res:shrec:texture}, the map quality is measured on only 21 landmarks, where the average error does not fully reflect our improvement.
(2) for some extremely challenging pairs, e.g., a failure case shown in Fig.~\ref{fig:res:failure}, all the maps from different masks have a poor quality, and thus the ``relative improvement'' on the average error is not informative. 
Moreover, we also computed the average error in the case where the mask term is removed from the total energy, i.e., set $\alpha_4 = 0$. In this case, the average error over the complete dataset is 0.112. We can see that, the standard and the slanted mask can have a negative effective in this case, while our resolvent mask still works to improve the map quality to some extent.

%% file: figures_tex/diff_fMapSize_normalized_faust.tex
%
%
\definecolor{mycolor1}{rgb}{0.00000,0.44706,0.74118}%
\definecolor{mycolor2}{rgb}{0.92941,0.69412,0.12549}%
\definecolor{mycolor3}{rgb}{0.85098,0.32549,0.09804}%
\pgfplotsset{
compat=1.11,
legend image code/.code={
\draw[mark repeat=2,mark phase=2]
plot coordinates {
(0cm,0cm)
(0.0cm,0cm)        
(0.3cm,0cm)         
};%
}
}
\begin{tikzpicture}

\begin{axis}[%
width=0.17\linewidth,
height=0.12\linewidth,
at={(2.706in,0.598in)},
scale only axis,
xmin=20,
xmax=250,
every x tick label/.append style={font=\color{black}, font=\tiny},
every y tick label/.append style={font=\color{black}, font=\tiny},
xlabel style={font=\color{white!15!black}},
xlabel={\small $k$: functional map size},
ymin=0.04,
ymax=0.22,
ylabel style={font=\color{white!15!black}},
ylabel={Average error},
axis background/.style={fill=white},
title style={font=\bfseries,align=center},
title={\small FAUST \\ \tiny per-vertex measure},
yticklabel style={
    /pgf/number format/fixed,
    /pgf/number format/precision=5
},
scaled y ticks=false,
xtick={50,100,150,200,250},
xmajorgrids,
ymajorgrids,
legend style={at={(0.68,0.65)}, anchor=south east, legend cell align=left, align=left, draw=white!15!black},
legend style={draw=none, legend columns=2, fill opacity=0.8, text opacity = 1, draw opacity=1},
legend style={inner sep=0pt}
]
\addlegendimage{only marks, mark=*,  mark options={solid, mycolor1}}
\addlegendimage{only marks, mark=*,  mark options={solid, mycolor2}}
\addlegendimage{only marks, mark=*,  mark options={solid, mycolor3}}

\addplot [color=mycolor1, line width=1.5pt]
  table[row sep=crcr]{%
20	0.1035\\
30	0.10912\\
40	0.086479\\
50	0.090658\\
60	0.098886\\
70	0.10486\\
80	0.12037\\
90	0.13908\\
100	0.14857\\
110	0.14193\\
120	0.1187\\
130	0.11486\\
140	0.12709\\
150	0.1359\\
160	0.14407\\
170	0.15719\\
180	0.17367\\
190	0.18612\\
200	0.19594\\
210	0.20476\\
220	0.19923\\
230	0.20094\\
240	0.20293\\
250	0.20507\\
};
\addlegendentry{\scriptsize standard}

\addplot [color=mycolor2, line width=1.5pt]
  table[row sep=crcr]{%
20	0.10037\\
30	0.08935\\
40	0.078217\\
50	0.077128\\
60	0.077742\\
70	0.080276\\
80	0.085478\\
90	0.087028\\
100	0.088431\\
110	0.08338\\
120	0.082206\\
130	0.084191\\
140	0.087536\\
150	0.092107\\
160	0.092546\\
170	0.095862\\
180	0.097079\\
190	0.098296\\
200	0.099022\\
210	0.10065\\
220	0.10135\\
230	0.10134\\
240	0.1016\\
250	0.10316\\
};
\addlegendentry{\scriptsize slanted}

\addplot [color=mycolor3, line width=1.5pt]
  table[row sep=crcr]{%
20	0.08889\\
30	0.082336\\
40	0.066915\\
50	0.063822\\
60	0.063505\\
70	0.062332\\
80	0.063067\\
90	0.062642\\
100	0.061891\\
110	0.059365\\
120	0.058197\\
130	0.058316\\
140	0.058832\\
150	0.059375\\
160	0.060128\\
170	0.06077\\
180	0.061592\\
190	0.062359\\
200	0.061455\\
210	0.059559\\
220	0.058472\\
230	0.057733\\
240	0.057842\\
250	0.058123\\
};
\addlegendentry{\scriptsize \textbf{ours}}

\addplot [color=mycolor1, dashed, line width=1.5pt]
  table[row sep=crcr]{%
20	0.071443\\
30	0.081741\\
40	0.060398\\
50	0.063931\\
60	0.067576\\
70	0.073423\\
80	0.085843\\
90	0.10863\\
100	0.12216\\
110	0.10283\\
120	0.08084\\
130	0.079169\\
140	0.086772\\
150	0.091691\\
160	0.097423\\
170	0.10885\\
180	0.12171\\
190	0.13292\\
200	0.13556\\
210	0.14089\\
220	0.1353\\
230	0.13856\\
240	0.13832\\
250	0.14055\\
};

\addplot [color=mycolor2, dashed, line width=1.5pt]
  table[row sep=crcr]{%
20	0.069425\\
30	0.066474\\
40	0.060717\\
50	0.062141\\
60	0.064091\\
70	0.065543\\
80	0.070697\\
90	0.077452\\
100	0.078027\\
110	0.070998\\
120	0.069181\\
130	0.070191\\
140	0.074927\\
150	0.081819\\
160	0.082273\\
170	0.085694\\
180	0.087559\\
190	0.088727\\
200	0.088536\\
210	0.092346\\
220	0.094314\\
230	0.093846\\
240	0.094805\\
250	0.096678\\
};

\addplot [color=mycolor3, dashed, line width=1.5pt]
  table[row sep=crcr]{%
20	0.064166\\
30	0.063501\\
40	0.050166\\
50	0.04791\\
60	0.046949\\
70	0.046971\\
80	0.048194\\
90	0.049204\\
100	0.048763\\
110	0.045694\\
120	0.042837\\
130	0.043097\\
140	0.043871\\
150	0.045197\\
160	0.045639\\
170	0.046602\\
180	0.04823\\
190	0.048422\\
200	0.047289\\
210	0.046838\\
220	0.046544\\
230	0.047095\\
240	0.046103\\
250	0.046006\\
};
\end{axis}
\end{tikzpicture}%

%% file: figures_tex/diff_fMapSize_normalized_faust_direct.tex
\definecolor{mycolor1}{rgb}{0.00000,0.44706,0.74118}%
\definecolor{mycolor2}{rgb}{0.92941,0.69412,0.12549}%
\definecolor{mycolor3}{rgb}{0.85098,0.32549,0.09804}%
\pgfplotsset{
compat=1.11,
legend image code/.code={
\draw[mark repeat=2,mark phase=2]
plot coordinates {
(0cm,0cm)
(0.0cm,0cm)        
(0.3cm,0cm)         
};%
}
}
\hspace{-3pt}
\begin{tikzpicture}
\begin{axis}[%
width=0.17\linewidth,
height=0.12\linewidth,
at={(2.706in,0.598in)},
scale only axis,
xmin=20,
xmax=250,
every x tick label/.append style={font=\color{black}, font=\tiny},
every y tick label/.append style={font=\color{black}, font=\tiny},
xlabel style={font=\color{white!15!black}},
xlabel={\small $k$},
ymin=0.09,
ymax=0.35,
ylabel style={font=\color{white!15!black}},
ylabel={},
axis background/.style={fill=white},
title style={font=\bfseries,align=center},
title={\small FAUST \\ \tiny direct measure},
yticklabel style={
    /pgf/number format/fixed,
    /pgf/number format/precision=5
},
scaled y ticks=false,
xtick={50,100,150,200,250},
xmajorgrids,
ymajorgrids,
legend style={at={(0.95,0.8)}, anchor=south east, legend cell align=left, align=left, draw=white!15!black},
legend style={draw=none, legend columns=2, fill opacity=0, text opacity = 1, draw opacity=1},
legend style={inner sep=0pt}
]
\addlegendimage{line width=1.5pt}
\addlegendimage{dashed, line width=1.5pt}
\addplot [color=mycolor1, line width=1.5pt]
  table[row sep=crcr]{%
20	0.21841\\
30	0.23739\\
40	0.21447\\
50	0.21438\\
60	0.21234\\
70	0.23984\\
80	0.25874\\
90	0.28187\\
100	0.28575\\
110	0.26702\\
120	0.2283\\
130	0.22006\\
140	0.22727\\
150	0.23452\\
160	0.24298\\
170	0.25568\\
180	0.27207\\
190	0.2861\\
200	0.29887\\
210	0.30798\\
220	0.30381\\
230	0.30493\\
240	0.30613\\
250	0.30659\\
};
\addlegendentry{\scriptsize Initialization}

\addplot [color=mycolor2, line width=1.5pt]
  table[row sep=crcr]{%
20	0.21239\\
30	0.23045\\
40	0.23656\\
50	0.23704\\
60	0.23953\\
70	0.25413\\
80	0.26339\\
90	0.2715\\
100	0.27182\\
110	0.26331\\
120	0.25333\\
130	0.24824\\
140	0.25061\\
150	0.25769\\
160	0.2574\\
170	0.25833\\
180	0.25949\\
190	0.26258\\
200	0.26771\\
210	0.27241\\
220	0.27356\\
230	0.27484\\
240	0.27785\\
250	0.27992\\
};

\addplot [color=mycolor3, line width=1.5pt]
  table[row sep=crcr]{%
20	0.18566\\
30	0.20301\\
40	0.17135\\
50	0.16692\\
60	0.15312\\
70	0.16891\\
80	0.17541\\
90	0.17749\\
100	0.17611\\
110	0.17215\\
120	0.16191\\
130	0.15722\\
140	0.15444\\
150	0.15425\\
160	0.15483\\
170	0.15349\\
180	0.15311\\
190	0.15464\\
200	0.15447\\
210	0.14901\\
220	0.1465\\
230	0.14416\\
240	0.14324\\
250	0.14279\\
};

\addplot [color=mycolor1, dashed, line width=1.5pt]
  table[row sep=crcr]{%
20	0.18045\\
30	0.20138\\
40	0.16657\\
50	0.15667\\
60	0.14667\\
70	0.15946\\
80	0.18254\\
90	0.20111\\
100	0.20844\\
110	0.18118\\
120	0.14769\\
130	0.14169\\
140	0.14836\\
150	0.15584\\
160	0.1611\\
170	0.17393\\
180	0.18567\\
190	0.20193\\
200	0.20761\\
210	0.21203\\
220	0.20658\\
230	0.20571\\
240	0.20123\\
250	0.20082\\
};
\addlegendentry{\scriptsize with ICP}

\addplot [color=mycolor2, dashed, line width=1.5pt]
  table[row sep=crcr]{%
20	0.17296\\
30	0.18965\\
40	0.18798\\
50	0.18427\\
60	0.18405\\
70	0.19949\\
80	0.20981\\
90	0.23242\\
100	0.2302\\
110	0.21647\\
120	0.20777\\
130	0.20485\\
140	0.21092\\
150	0.2221\\
160	0.21702\\
170	0.22234\\
180	0.22248\\
190	0.22359\\
200	0.22397\\
210	0.23347\\
220	0.23507\\
230	0.23459\\
240	0.23831\\
250	0.2403\\
};

\addplot [color=mycolor3, dashed, line width=1.5pt]
  table[row sep=crcr]{%
20	0.15513\\
30	0.17826\\
40	0.144\\
50	0.13095\\
60	0.10723\\
70	0.11607\\
80	0.11883\\
90	0.12265\\
100	0.11805\\
110	0.11174\\
120	0.10271\\
130	0.099894\\
140	0.10157\\
150	0.10501\\
160	0.10467\\
170	0.10716\\
180	0.10995\\
190	0.11286\\
200	0.10979\\
210	0.10559\\
220	0.1064\\
230	0.10546\\
240	0.10268\\
250	0.10032\\
};
\end{axis}
\end{tikzpicture}%

%% file: figures_tex/diff_fMapSize_normalized_tosca.tex
%
%
\definecolor{mycolor1}{rgb}{0.00000,0.44706,0.74118}%
\definecolor{mycolor2}{rgb}{0.92941,0.69412,0.12549}%
\definecolor{mycolor3}{rgb}{0.85098,0.32549,0.09804}%
\pgfplotsset{
compat=1.11,
legend image code/.code={
\draw[mark repeat=2,mark phase=2]
plot coordinates {
(0cm,0cm)
(0.0cm,0cm)        
(0.3cm,0cm)         
};%
}
}
\begin{tikzpicture}
\begin{axis}[%
width=0.17\linewidth,
height=0.12\linewidth,
at={(2.706in,0.598in)},
scale only axis,
xmin=20,
xmax=250,
every x tick label/.append style={font=\color{black}, font=\tiny},
every y tick label/.append style={font=\color{black}, font=\tiny},
xlabel style={font=\color{white!15!black}},
xlabel={\small $k$},
ymin=0.04,
ymax=0.125,
ylabel style={font=\color{white!15!black}},
ylabel={},
axis background/.style={fill=white},
title style={font=\bfseries,align=center},
title={\small TOSCA\\ \tiny per-vertex measure },
yticklabel style={
    /pgf/number format/fixed,
    /pgf/number format/precision=5
},
scaled y ticks=false,
xtick={50,100,150,200,250},
xmajorgrids,
ymajorgrids,
legend style={at={(1.03,0.5)}, anchor=west, legend cell align=left, align=left, draw=white!15!black}
]
\addplot [color=mycolor1, line width=1.5pt]
  table[row sep=crcr]{%
20	0.07145\\
30	0.06837\\
40	0.061771\\
50	0.062594\\
60	0.064414\\
70	0.067376\\
80	0.071365\\
90	0.076168\\
100	0.080084\\
110	0.083664\\
120	0.083722\\
130	0.085124\\
140	0.087442\\
150	0.088349\\
160	0.091537\\
170	0.094074\\
180	0.09824\\
190	0.10189\\
200	0.10304\\
210	0.10563\\
220	0.10864\\
230	0.11169\\
240	0.11341\\
250	0.11762\\
};
\addlegendentry{standard}

\addplot [color=mycolor2, line width=1.5pt]
  table[row sep=crcr]{%
20	0.081275\\
30	0.074875\\
40	0.069179\\
50	0.066949\\
60	0.068705\\
70	0.069811\\
80	0.070836\\
90	0.071131\\
100	0.070884\\
110	0.071149\\
120	0.06979\\
130	0.068732\\
140	0.067632\\
150	0.069573\\
160	0.069235\\
170	0.070246\\
180	0.070497\\
190	0.070421\\
200	0.071533\\
210	0.071759\\
220	0.071394\\
230	0.071834\\
240	0.073398\\
250	0.073398\\
};
\addlegendentry{slanted}

\addplot [color=mycolor3, line width=1.5pt]
  table[row sep=crcr]{%
20	0.067581\\
30	0.062745\\
40	0.056208\\
50	0.055288\\
60	0.053641\\
70	0.05419\\
80	0.054057\\
90	0.053792\\
100	0.053466\\
110	0.053645\\
120	0.052645\\
130	0.052429\\
140	0.052082\\
150	0.052396\\
160	0.052712\\
170	0.053182\\
180	0.053375\\
190	0.053638\\
200	0.053576\\
210	0.053876\\
220	0.053815\\
230	0.053879\\
240	0.054392\\
250	0.05492\\
};
\addlegendentry{complRes}

\addplot [color=mycolor1, dashed, line width=1.5pt]
  table[row sep=crcr]{%
20	0.051502\\
30	0.046661\\
40	0.045737\\
50	0.045872\\
60	0.048941\\
70	0.052236\\
80	0.056364\\
90	0.062976\\
100	0.063684\\
110	0.068729\\
120	0.065568\\
130	0.064699\\
140	0.067308\\
150	0.066838\\
160	0.071679\\
170	0.074291\\
180	0.078634\\
190	0.082335\\
200	0.084474\\
210	0.086197\\
220	0.089602\\
230	0.092042\\
240	0.095811\\
250	0.10008\\
};
\addlegendentry{standard + ICP}

\addplot [color=mycolor2, dashed, line width=1.5pt]
  table[row sep=crcr]{%
20	0.057426\\
30	0.049342\\
40	0.049119\\
50	0.048379\\
60	0.051774\\
70	0.054613\\
80	0.05716\\
90	0.058193\\
100	0.058177\\
110	0.059723\\
120	0.058177\\
130	0.058035\\
140	0.056428\\
150	0.05825\\
160	0.059113\\
170	0.060781\\
180	0.061864\\
190	0.061403\\
200	0.062654\\
210	0.063629\\
220	0.063321\\
230	0.064448\\
240	0.067089\\
250	0.066009\\
};
\addlegendentry{slanted + ICP}

\addplot [color=mycolor3, dashed, line width=1.5pt]
  table[row sep=crcr]{%
20	0.051953\\
30	0.043822\\
40	0.04293\\
50	0.041658\\
60	0.040921\\
70	0.041914\\
80	0.041902\\
90	0.043106\\
100	0.042263\\
110	0.04404\\
120	0.04164\\
130	0.041846\\
140	0.041213\\
150	0.041557\\
160	0.042141\\
170	0.042467\\
180	0.042953\\
190	0.04405\\
200	0.044024\\
210	0.044597\\
220	0.044843\\
230	0.044784\\
240	0.045478\\
250	0.046112\\
};
\addlegendentry{complRes + ICP}
\legend{}
\end{axis}
\end{tikzpicture}%

%% file: figures_tex/diff_fMapSize_normalized_tosca_direct.tex
\definecolor{mycolor1}{rgb}{0.00000,0.44706,0.74118}%
\definecolor{mycolor2}{rgb}{0.92941,0.69412,0.12549}%
\definecolor{mycolor3}{rgb}{0.85098,0.32549,0.09804}%
\pgfplotsset{
compat=1.11,
legend image code/.code={
\draw[mark repeat=2,mark phase=2]
plot coordinates {
(0cm,0cm)
(0.0cm,0cm)        
(0.3cm,0cm)         
};%
}
}
\begin{tikzpicture}
\begin{axis}[%
width=0.17\linewidth,
height=0.12\linewidth,
at={(2.706in,0.598in)},
scale only axis,
xmin=20,
xmax=250,
every x tick label/.append style={font=\color{black}, font=\tiny},
every y tick label/.append style={font=\color{black}, font=\tiny},
xlabel style={font=\color{white!15!black}},
xlabel={\small $k$},
ymin=0.09,
ymax=0.29,
ylabel style={font=\color{white!15!black}},
ylabel={},
axis background/.style={fill=white},
title style={font=\bfseries,align=center},
title={\small TOSCA\\ \tiny direct measure},
yticklabel style={
    /pgf/number format/fixed,
    /pgf/number format/precision=5
},
scaled y ticks=false,
xtick={50,100,150,200,250},
xmajorgrids,
ymajorgrids,
legend style={at={(1.03,0.5)}, anchor=west, legend cell align=left, align=left, draw=white!15!black,
legend style={row sep=0.1pt},
font=\scriptsize}
]
\addplot [color=mycolor1, line width=2.0pt]
  table[row sep=crcr]{%
20	0.13143\\
30	0.12206\\
40	0.12984\\
50	0.14884\\
60	0.173\\
70	0.19365\\
80	0.20675\\
90	0.21895\\
100	0.22624\\
110	0.23242\\
120	0.22884\\
130	0.22625\\
140	0.22605\\
150	0.22827\\
160	0.23366\\
170	0.23667\\
180	0.24136\\
190	0.24517\\
200	0.24816\\
210	0.25282\\
220	0.25651\\
230	0.26013\\
240	0.26231\\
250	0.26576\\
};
\addlegendentry{\scriptsize standard}

\addplot [color=mycolor2, line width=2.0pt]
  table[row sep=crcr]{%
20	0.16626\\
30	0.17048\\
40	0.18457\\
50	0.19684\\
60	0.21216\\
70	0.22233\\
80	0.22914\\
90	0.22752\\
100	0.23309\\
110	0.23778\\
120	0.23496\\
130	0.23055\\
140	0.22605\\
150	0.22862\\
160	0.23002\\
170	0.23002\\
180	0.23111\\
190	0.22989\\
200	0.23464\\
210	0.23774\\
220	0.23452\\
230	0.2393\\
240	0.24273\\
250	0.24172\\
};
\addlegendentry{\scriptsize slanted}

\addplot [color=mycolor3, line width=2.0pt]
  table[row sep=crcr]{%
20	0.12934\\
30	0.11647\\
40	0.11823\\
50	0.12318\\
60	0.13204\\
70	0.1427\\
80	0.14662\\
90	0.15334\\
100	0.15466\\
110	0.15668\\
120	0.1579\\
130	0.15173\\
140	0.14968\\
150	0.15036\\
160	0.15236\\
170	0.15468\\
180	0.15489\\
190	0.15651\\
200	0.15752\\
210	0.15928\\
220	0.15992\\
230	0.16069\\
240	0.1629\\
250	0.16536\\
};
\addlegendentry{\scriptsize complRes}

\addplot [color=mycolor1, dashed, line width=2.0pt]
  table[row sep=crcr]{%
20	0.10955\\
30	0.091946\\
40	0.1021\\
50	0.11378\\
60	0.13129\\
70	0.14826\\
80	0.1626\\
90	0.17774\\
100	0.17988\\
110	0.17849\\
120	0.17183\\
130	0.16505\\
140	0.16993\\
150	0.16844\\
160	0.17389\\
170	0.18094\\
180	0.1866\\
190	0.19212\\
200	0.19458\\
210	0.19856\\
220	0.20297\\
230	0.20612\\
240	0.21095\\
250	0.21526\\
};
\addlegendentry{\scriptsize standard + ICP}

\addplot [color=mycolor2, dashed, line width=2.0pt]
  table[row sep=crcr]{%
20	0.13378\\
30	0.1357\\
40	0.14156\\
50	0.15841\\
60	0.17383\\
70	0.18512\\
80	0.19716\\
90	0.19839\\
100	0.20007\\
110	0.20701\\
120	0.2035\\
130	0.1984\\
140	0.19473\\
150	0.19833\\
160	0.20156\\
170	0.20252\\
180	0.20532\\
190	0.20348\\
200	0.2065\\
210	0.20882\\
220	0.20695\\
230	0.21153\\
240	0.21608\\
250	0.21461\\
};
\addlegendentry{\scriptsize slanted + ICP}

\addplot [color=mycolor3, dashed, line width=2.0pt]
  table[row sep=crcr]{%
20	0.1128\\
30	0.096427\\
40	0.10354\\
50	0.10254\\
60	0.10893\\
70	0.11365\\
80	0.11265\\
90	0.11741\\
100	0.11468\\
110	0.1194\\
120	0.11761\\
130	0.10946\\
140	0.10936\\
150	0.11215\\
160	0.11339\\
170	0.11601\\
180	0.11664\\
190	0.11947\\
200	0.11972\\
210	0.12221\\
220	0.12084\\
230	0.12106\\
240	0.12168\\
250	0.12343\\
};
\addlegendentry{\scriptsize complRes + ICP}
\legend{}
\end{axis}
\end{tikzpicture}%

%% file: figures_tex/summaryCurve_norm_faust_combined.tex
%
%
\definecolor{mycolor1}{rgb}{0.00000,0.44706,0.74118}%
\definecolor{mycolor2}{rgb}{0.92941,0.69412,0.12549}%
\definecolor{mycolor3}{rgb}{0.85098,0.32549,0.09804}%
\pgfplotsset{scaled x ticks=false}
\pgfplotsset{
compat=1.11,
legend image code/.code={
\draw[mark repeat=2,mark phase=2]
plot coordinates {
(0cm,0cm)
(0.15cm,0cm)        
(0.3cm,0cm)         
};%
}
}
\begin{tikzpicture}

\begin{axis}[%
width=0.2\linewidth,
height=0.2\linewidth,
every x tick label/.append style={font=\color{black}, font=\tiny},
every y tick label/.append style={font=\color{black}, font=\tiny},
at={(7.226in,0.894in)},
scale only axis,
xmin=0,
xmax=0.12,
xlabel style={font=\color{white!15!black}},
xlabel={\footnotesize{Geodesic Error ($\times 10^{-3}$)}},
ymin=0,
ymax=100.5,
ylabel style={font=\color{white!15!black}},
ylabel={\% correspondence},
axis background/.style={fill=white},
title style={font=\bfseries},
title={FAUST},
xmajorgrids,
ymajorgrids,
xtick={0,0.02,0.04,0.06,0.08,0.10, 0.12},
xticklabels={0,20,40,60,80,100,120},
legend style={at={(1.8,0)}, anchor=south east, legend cell align=left, align=left, draw=white!15!black}
]
\addplot [color=mycolor1, dotted, line width=1.5pt]
  table[row sep=crcr]{%
0	7.70016666666667\\
0.002	7.70046666666667\\
0.004	7.7941\\
0.006	8.0782\\
0.008	8.52763333333333\\
0.01	9.31506666666667\\
0.012	10.4421666666667\\
0.014	12.1757\\
0.016	14.8473333333333\\
0.018	18.4492\\
0.02	22.2889333333333\\
0.022	26.0455666666667\\
0.024	29.6529333333333\\
0.026	32.9144333333333\\
0.028	35.7437333333333\\
0.03	38.2455\\
0.032	40.484\\
0.034	42.5455\\
0.036	44.4329\\
0.038	46.1442333333333\\
0.04	47.7811\\
0.042	49.3592333333333\\
0.044	50.8903\\
0.046	52.3750666666667\\
0.048	53.7858\\
0.05	55.1086666666667\\
0.052	56.3478333333333\\
0.054	57.5113\\
0.056	58.6279666666667\\
0.058	59.6658333333333\\
0.06	60.6174666666667\\
0.062	61.5178666666667\\
0.064	62.336\\
0.066	63.123\\
0.068	63.8696\\
0.07	64.6009333333333\\
0.072	65.2934666666667\\
0.074	65.9346333333333\\
0.076	66.5615666666667\\
0.078	67.1354\\
0.08	67.6921\\
0.082	68.2019666666667\\
0.084	68.6856666666667\\
0.086	69.1348333333333\\
0.088	69.5412333333333\\
0.09	69.9198\\
0.092	70.2725666666667\\
0.094	70.5928666666667\\
0.096	70.8926333333333\\
0.098	71.1645\\
0.1	71.4088\\
0.102	71.6444666666667\\
0.104	71.8579666666667\\
0.106	72.0531333333333\\
0.108	72.2397666666667\\
0.11	72.4136\\
0.112	72.5692666666667\\
0.114	72.7032666666667\\
0.116	72.8323\\
0.118	72.9446\\
0.12	73.0585666666667\\
0.122	73.1625\\
0.124	73.2607\\
0.126	73.3528333333333\\
0.128	73.4403666666667\\
0.13	73.5168\\
0.132	73.5907\\
0.134	73.6633666666667\\
0.136	73.7342666666667\\
0.138	73.8044666666667\\
0.14	73.8763666666667\\
0.142	73.9424666666667\\
0.144	74.0053333333333\\
0.146	74.0680666666667\\
0.148	74.1305666666667\\
0.15	74.1897333333333\\
};
\addlegendentry{\scriptsize standard (1) : 201.7}

\addplot [color=mycolor2, dotted, line width=1.5pt]
  table[row sep=crcr]{%
0	5.6517\\
0.002	5.65183333333333\\
0.004	5.7299\\
0.006	5.95343333333333\\
0.008	6.3086\\
0.01	6.914\\
0.012	7.804\\
0.014	9.22346666666667\\
0.016	11.528\\
0.018	14.6969666666667\\
0.02	18.0771666666667\\
0.022	21.3601666666667\\
0.024	24.4666\\
0.026	27.2811333333333\\
0.028	29.8062666666667\\
0.03	32.1626\\
0.032	34.4306\\
0.034	36.5733666666667\\
0.036	38.6226666666667\\
0.038	40.5421333333333\\
0.04	42.3673333333333\\
0.042	44.1315333333333\\
0.044	45.8445666666667\\
0.046	47.5028\\
0.048	49.0904666666667\\
0.05	50.5739333333333\\
0.052	51.9504666666667\\
0.054	53.2389\\
0.056	54.438\\
0.058	55.5564333333333\\
0.06	56.6095\\
0.062	57.5851\\
0.064	58.4782\\
0.066	59.3161333333333\\
0.068	60.1200666666667\\
0.07	60.8798\\
0.072	61.6017\\
0.074	62.2753666666667\\
0.076	62.9266666666667\\
0.078	63.5225333333333\\
0.08	64.0973333333333\\
0.082	64.6461333333333\\
0.084	65.1520333333333\\
0.086	65.6473\\
0.088	66.1032\\
0.09	66.5430333333333\\
0.092	66.9705\\
0.094	67.3642666666667\\
0.096	67.7433\\
0.098	68.1042333333333\\
0.1	68.4460333333333\\
0.102	68.7729666666667\\
0.104	69.0821\\
0.106	69.3694333333333\\
0.108	69.6424333333333\\
0.11	69.8956666666667\\
0.112	70.1373333333333\\
0.114	70.3532666666667\\
0.116	70.5699\\
0.118	70.7658\\
0.12	70.9483666666667\\
0.122	71.1326333333333\\
0.124	71.3024333333333\\
0.126	71.4675\\
0.128	71.6210333333333\\
0.13	71.7570333333333\\
0.132	71.888\\
0.134	72.0160666666667\\
0.136	72.1325666666667\\
0.138	72.2496\\
0.14	72.3588\\
0.142	72.4696666666667\\
0.144	72.5756333333333\\
0.146	72.6795666666667\\
0.148	72.7782\\
0.15	72.8696666666667\\
};
\addlegendentry{\scriptsize slanted (2) : 188.2}

\addplot [color=mycolor3, dotted, line width=1.5pt]
  table[row sep=crcr]{%
0	10.0962\\
0.002	10.0964333333333\\
0.004	10.2101\\
0.006	10.5510333333333\\
0.008	11.1449\\
0.01	12.1911333333333\\
0.012	13.7292\\
0.014	16.1242\\
0.016	19.8328666666667\\
0.018	24.6831\\
0.02	29.6859\\
0.022	34.4170666666667\\
0.024	38.8262333333333\\
0.026	42.7823333333333\\
0.028	46.2343666666667\\
0.03	49.2851666666667\\
0.032	52.0288666666667\\
0.034	54.5110333333333\\
0.036	56.7673333333333\\
0.038	58.8302\\
0.04	60.7526666666667\\
0.042	62.5737666666667\\
0.044	64.3038333333333\\
0.046	65.9627\\
0.048	67.5036\\
0.05	68.9278666666667\\
0.052	70.2408333333333\\
0.054	71.4493666666667\\
0.056	72.6080666666667\\
0.058	73.6950333333333\\
0.06	74.6804333333333\\
0.062	75.5938\\
0.064	76.4216\\
0.066	77.1936333333333\\
0.068	77.9307\\
0.07	78.6479666666667\\
0.072	79.3249\\
0.074	79.9423666666667\\
0.076	80.5416\\
0.078	81.0824333333333\\
0.08	81.5993666666667\\
0.082	82.066\\
0.084	82.5073333333333\\
0.086	82.9118\\
0.088	83.2729\\
0.09	83.6020333333333\\
0.092	83.9092666666667\\
0.094	84.1788\\
0.096	84.4314333333333\\
0.098	84.6541333333333\\
0.1	84.85\\
0.102	85.0258666666667\\
0.104	85.1892\\
0.106	85.3364333333333\\
0.108	85.4786\\
0.11	85.6069333333333\\
0.112	85.7151\\
0.114	85.8103333333333\\
0.116	85.8969666666667\\
0.118	85.971\\
0.12	86.0408666666667\\
0.122	86.1089333333333\\
0.124	86.1679333333333\\
0.126	86.2261\\
0.128	86.2785\\
0.13	86.3255\\
0.132	86.3691333333333\\
0.134	86.4090333333333\\
0.136	86.446\\
0.138	86.4831333333333\\
0.14	86.5191333333333\\
0.142	86.5538333333333\\
0.144	86.5880666666667\\
0.146	86.6191333333333\\
0.148	86.6499\\
0.15	86.6794333333333\\
};
\addlegendentry{\scriptsize \textbf{ours (3) : 112.8}}

\addplot [color=mycolor1, dashed, line width=1.5pt]
  table[row sep=crcr]{%
0	17.1088\\
0.002	17.1094666666667\\
0.004	17.2573666666667\\
0.006	18.3133666666667\\
0.008	20.0343\\
0.01	22.316\\
0.012	25.6096\\
0.014	29.7552666666667\\
0.016	34.6014333333333\\
0.018	39.879\\
0.02	44.8637333333333\\
0.022	49.3866333333333\\
0.024	53.5202666666667\\
0.026	57.0656333333333\\
0.028	59.9513333333333\\
0.03	62.3058666666667\\
0.032	64.3175666666667\\
0.034	66\\
0.036	67.4727333333333\\
0.038	68.7716333333333\\
0.04	69.9516\\
0.042	71.0645666666667\\
0.044	72.1237666666667\\
0.046	73.1308333333333\\
0.048	74.0827666666667\\
0.05	74.9576333333333\\
0.052	75.7224666666667\\
0.054	76.4236\\
0.056	77.0575\\
0.058	77.6342666666667\\
0.06	78.1775\\
0.062	78.6919333333333\\
0.064	79.1529\\
0.066	79.5846666666667\\
0.068	79.9849333333333\\
0.07	80.3649666666667\\
0.072	80.7195333333333\\
0.074	81.0470333333333\\
0.076	81.3515333333333\\
0.078	81.6331333333333\\
0.08	81.8935\\
0.082	82.1326666666667\\
0.084	82.3624666666667\\
0.086	82.5740333333333\\
0.088	82.766\\
0.09	82.9516\\
0.092	83.1211666666667\\
0.094	83.2886333333333\\
0.096	83.4491\\
0.098	83.5979333333333\\
0.1	83.7353\\
0.102	83.8662\\
0.104	83.9872333333333\\
0.106	84.0984\\
0.108	84.2093\\
0.11	84.3145\\
0.112	84.4112333333333\\
0.114	84.5027333333333\\
0.116	84.5876\\
0.118	84.6727\\
0.12	84.7532666666667\\
0.122	84.8263666666667\\
0.124	84.8977333333333\\
0.126	84.9714333333333\\
0.128	85.0369\\
0.13	85.1\\
0.132	85.163\\
0.134	85.2263666666667\\
0.136	85.2873\\
0.138	85.3483\\
0.14	85.4055666666667\\
0.142	85.4623\\
0.144	85.5165666666667\\
0.146	85.5683333333333\\
0.148	85.6185\\
0.15	85.6682666666667\\
};
\addlegendentry{\scriptsize (1) + ICP : 116.5}

\addplot [color=mycolor2, dashed, line width=1.5pt]
  table[row sep=crcr]{%
0	14.9030333333333\\
0.002	14.9038666666667\\
0.004	15.0373666666667\\
0.006	15.9675\\
0.008	17.4972\\
0.01	19.5776333333333\\
0.012	22.5759333333333\\
0.014	26.3760333333333\\
0.016	30.9032333333333\\
0.018	35.8508666666667\\
0.02	40.5085666666667\\
0.022	44.7265666666667\\
0.024	48.5015\\
0.026	51.7413\\
0.028	54.3907333333333\\
0.03	56.5786666666667\\
0.032	58.4817\\
0.034	60.1398333333333\\
0.036	61.5998333333333\\
0.038	62.8915333333333\\
0.04	64.0766666666667\\
0.042	65.172\\
0.044	66.2199333333333\\
0.046	67.2201\\
0.048	68.1636\\
0.05	69.0203\\
0.052	69.7929333333333\\
0.054	70.4891333333333\\
0.056	71.1018333333333\\
0.058	71.6731\\
0.06	72.1975666666667\\
0.062	72.6864333333333\\
0.064	73.1416\\
0.066	73.5658666666667\\
0.068	73.9603666666667\\
0.07	74.3386\\
0.072	74.6877\\
0.074	75.0120333333333\\
0.076	75.3213333333333\\
0.078	75.6092666666667\\
0.08	75.8739\\
0.082	76.1203\\
0.084	76.3533\\
0.086	76.5774666666667\\
0.088	76.7826666666667\\
0.09	76.9745666666667\\
0.092	77.1600333333333\\
0.094	77.3327\\
0.096	77.4965333333333\\
0.098	77.6468333333333\\
0.1	77.7898333333333\\
0.102	77.9221\\
0.104	78.0409\\
0.106	78.1610333333333\\
0.108	78.2782\\
0.11	78.3859333333333\\
0.112	78.4896333333333\\
0.114	78.5865\\
0.116	78.6821666666667\\
0.118	78.7659\\
0.12	78.8512333333333\\
0.122	78.9343\\
0.124	79.0108333333333\\
0.126	79.0861666666667\\
0.128	79.16\\
0.13	79.2276333333333\\
0.132	79.2982666666667\\
0.134	79.367\\
0.136	79.4327666666667\\
0.138	79.4985\\
0.14	79.5635666666667\\
0.142	79.6245666666667\\
0.144	79.6835\\
0.146	79.7460333333333\\
0.148	79.8025333333333\\
0.15	79.8577333333333\\
};
\addlegendentry{\scriptsize (2) + ICP : 141.3}

\addplot [color=mycolor3, dashed, line width=1.5pt]
  table[row sep=crcr]{%
0	19.2054333333333\\
0.002	19.2060333333333\\
0.004	19.3741333333333\\
0.006	20.5751333333333\\
0.008	22.5441\\
0.01	25.1863\\
0.012	29.0012333333333\\
0.014	33.7899666666667\\
0.016	39.3743333333333\\
0.018	45.3588333333333\\
0.02	50.9713333333333\\
0.022	56.0309666666667\\
0.024	60.6083666666667\\
0.026	64.4939666666667\\
0.028	67.6555\\
0.03	70.2339666666667\\
0.032	72.4085666666667\\
0.034	74.2494666666667\\
0.036	75.8263333333333\\
0.038	77.1990666666667\\
0.04	78.4391666666667\\
0.042	79.5847333333333\\
0.044	80.6625\\
0.046	81.6807333333333\\
0.048	82.6239\\
0.05	83.4714\\
0.052	84.2213\\
0.054	84.8896333333333\\
0.056	85.4958333333333\\
0.058	86.0444666666667\\
0.06	86.5436666666667\\
0.062	87.0046666666667\\
0.064	87.4197666666667\\
0.066	87.8032666666667\\
0.068	88.1542666666667\\
0.07	88.4757666666667\\
0.072	88.7731666666667\\
0.074	89.0408\\
0.076	89.2923333333333\\
0.078	89.5226666666667\\
0.08	89.7367666666667\\
0.082	89.9327333333333\\
0.084	90.1108333333333\\
0.086	90.2765333333333\\
0.088	90.4289666666667\\
0.09	90.5713666666667\\
0.092	90.7027\\
0.094	90.8203\\
0.096	90.9349\\
0.098	91.0424333333333\\
0.1	91.1412333333333\\
0.102	91.2311333333333\\
0.104	91.3148\\
0.106	91.3891\\
0.108	91.4639666666667\\
0.11	91.5302666666667\\
0.112	91.5909333333333\\
0.114	91.6469\\
0.116	91.6973666666667\\
0.118	91.7458333333333\\
0.12	91.7916666666667\\
0.122	91.8347333333333\\
0.124	91.8777\\
0.126	91.9172333333333\\
0.128	91.9561333333333\\
0.13	91.9909333333333\\
0.132	92.0255666666667\\
0.134	92.0568\\
0.136	92.0869333333333\\
0.138	92.1159\\
0.14	92.1452\\
0.142	92.1738666666667\\
0.144	92.1987333333333\\
0.146	92.226\\
0.148	92.2514333333333\\
0.15	92.2749666666667\\
};
\addlegendentry{\scriptsize \textbf{(3) + ICP : 66.0}}

\addplot [color=mycolor1, line width=1.5pt]
  table[row sep=crcr]{%
0	17.7408666666667\\
0.002	17.742\\
0.004	17.9853666666667\\
0.006	19.3442\\
0.008	21.3780666666667\\
0.01	24.0850666666667\\
0.012	27.9826333333333\\
0.014	32.9366333333333\\
0.016	38.6294666666667\\
0.018	44.5402\\
0.02	49.8690666666667\\
0.022	54.4864\\
0.024	58.5420333333333\\
0.026	62.0486666666667\\
0.028	64.9997666666667\\
0.03	67.5068666666667\\
0.032	69.6394333333333\\
0.034	71.4750666666667\\
0.036	73.0425\\
0.038	74.4138333333333\\
0.04	75.6884666666667\\
0.042	76.8968\\
0.044	78.0479333333333\\
0.046	79.1707\\
0.048	80.2171333333333\\
0.05	81.1813666666667\\
0.052	82.0476\\
0.054	82.8296666666667\\
0.056	83.5274666666667\\
0.058	84.1499666666667\\
0.06	84.7297333333333\\
0.062	85.2651666666667\\
0.064	85.7531333333333\\
0.066	86.2003666666667\\
0.068	86.6270666666667\\
0.07	87.0106666666667\\
0.072	87.379\\
0.074	87.7172666666667\\
0.076	88.0295666666667\\
0.078	88.3190333333333\\
0.08	88.5853\\
0.082	88.8368333333333\\
0.084	89.0737666666667\\
0.086	89.2969\\
0.088	89.5085666666667\\
0.09	89.709\\
0.092	89.8978\\
0.094	90.0785666666667\\
0.096	90.2557333333333\\
0.098	90.4188333333333\\
0.1	90.5796333333333\\
0.102	90.7318\\
0.104	90.8783\\
0.106	91.0139333333333\\
0.108	91.1446333333333\\
0.11	91.2732666666667\\
0.112	91.3968666666667\\
0.114	91.5123\\
0.116	91.6228666666667\\
0.118	91.729\\
0.12	91.8331\\
0.122	91.9321\\
0.124	92.0291\\
0.126	92.1154\\
0.128	92.2004\\
0.13	92.2860333333333\\
0.132	92.3684333333333\\
0.134	92.444\\
0.136	92.5190666666667\\
0.138	92.5947\\
0.14	92.6673666666667\\
0.142	92.737\\
0.144	92.8048\\
0.146	92.873\\
0.148	92.9401333333333\\
0.15	93.0034333333333\\
};
\addlegendentry{\scriptsize (1) + BCICP : 56.9}

\addplot [color=mycolor2, line width=1.5pt]
  table[row sep=crcr]{%
0	16.1030333333333\\
0.002	16.1039333333333\\
0.004	16.3385333333333\\
0.006	17.6093666666667\\
0.008	19.5109666666667\\
0.01	22.0491\\
0.012	25.6735666666667\\
0.014	30.1981333333333\\
0.016	35.3420333333333\\
0.018	40.7005\\
0.02	45.5567666666667\\
0.022	49.7720333333333\\
0.024	53.4326333333333\\
0.026	56.5350666666667\\
0.028	59.1142666666667\\
0.03	61.2948666666667\\
0.032	63.1397\\
0.034	64.7105666666667\\
0.036	66.0696333333333\\
0.038	67.2636\\
0.04	68.3561666666667\\
0.042	69.3739333333333\\
0.044	70.3434666666667\\
0.046	71.2713666666667\\
0.048	72.1328\\
0.05	72.918\\
0.052	73.6122333333333\\
0.054	74.2350666666667\\
0.056	74.7787333333333\\
0.058	75.2641666666667\\
0.06	75.7084333333333\\
0.062	76.1068666666667\\
0.064	76.4687333333333\\
0.066	76.8100666666667\\
0.068	77.1270666666667\\
0.07	77.4227666666667\\
0.072	77.7016666666667\\
0.074	77.9656333333333\\
0.076	78.2144333333333\\
0.078	78.4433666666667\\
0.08	78.6572\\
0.082	78.8586333333333\\
0.084	79.0523333333333\\
0.086	79.2407\\
0.088	79.4166666666667\\
0.09	79.5799\\
0.092	79.7434333333333\\
0.094	79.9036333333333\\
0.096	80.0559666666667\\
0.098	80.2019666666667\\
0.1	80.3461666666667\\
0.102	80.4855\\
0.104	80.6193666666667\\
0.106	80.7478333333333\\
0.108	80.8727\\
0.11	80.9848\\
0.112	81.0966666666667\\
0.114	81.2076666666667\\
0.116	81.3135666666667\\
0.118	81.4180333333333\\
0.12	81.5217666666667\\
0.122	81.6179333333333\\
0.124	81.7088666666667\\
0.126	81.8044333333333\\
0.128	81.8934666666667\\
0.13	81.9783666666667\\
0.132	82.0631333333333\\
0.134	82.1462\\
0.136	82.2298\\
0.138	82.3132\\
0.14	82.3939333333333\\
0.142	82.4761333333333\\
0.144	82.5608333333333\\
0.146	82.6383666666667\\
0.148	82.7160666666667\\
0.15	82.7926666666667\\
};
\addlegendentry{\scriptsize (2) + BCICP : 116.5}

\addplot [color=mycolor3, line width=1.5pt]
  table[row sep=crcr]{%
0	19.51\\
0.002	19.5112333333333\\
0.004	19.7794666666667\\
0.006	21.2821\\
0.008	23.5337\\
0.01	26.5168333333333\\
0.012	30.7868\\
0.014	36.1642333333333\\
0.016	42.3031666666667\\
0.018	48.6910333333333\\
0.02	54.4194\\
0.022	59.4184\\
0.024	63.7822\\
0.026	67.4820333333333\\
0.028	70.5462333333333\\
0.03	73.1038666666667\\
0.032	75.2602666666667\\
0.034	77.0831\\
0.036	78.6250666666667\\
0.038	79.9487333333333\\
0.04	81.1589666666667\\
0.042	82.2935333333333\\
0.044	83.3555\\
0.046	84.3853333333333\\
0.048	85.3344\\
0.05	86.1972333333333\\
0.052	86.9601333333333\\
0.054	87.6267333333333\\
0.056	88.2134333333333\\
0.058	88.7121666666667\\
0.06	89.1693333333333\\
0.062	89.5836666666667\\
0.064	89.9552666666667\\
0.066	90.2948\\
0.068	90.6029333333333\\
0.07	90.8870666666667\\
0.072	91.1528\\
0.074	91.3899\\
0.076	91.6130666666667\\
0.078	91.8144\\
0.08	91.9961\\
0.082	92.1634333333333\\
0.084	92.3162666666667\\
0.086	92.4661666666667\\
0.088	92.5977666666667\\
0.09	92.7264666666667\\
0.092	92.8478666666667\\
0.094	92.9597\\
0.096	93.0670333333333\\
0.098	93.1728666666667\\
0.1	93.2675666666667\\
0.102	93.3625333333333\\
0.104	93.4497666666667\\
0.106	93.5297333333333\\
0.108	93.6053666666667\\
0.11	93.6799\\
0.112	93.7507666666667\\
0.114	93.8177666666667\\
0.116	93.8805333333333\\
0.118	93.9392333333333\\
0.12	93.9941333333333\\
0.122	94.0429\\
0.124	94.0897333333333\\
0.126	94.1362666666667\\
0.128	94.1792\\
0.13	94.2220666666667\\
0.132	94.2636333333333\\
0.134	94.3022333333333\\
0.136	94.3424\\
0.138	94.3804\\
0.14	94.4186333333333\\
0.142	94.4539333333333\\
0.144	94.4887\\
0.146	94.5204666666667\\
0.148	94.5538333333333\\
0.15	94.5862\\
};
\addlegendentry{\scriptsize \textbf{(3) + BCICP : 47.6}}
\end{axis}
\end{tikzpicture}%

%% file: figures_tex/summaryCurve_norm_tosca_new.tex
%
%
\definecolor{mycolor1}{rgb}{0.00000,0.44706,0.74118}%
\definecolor{mycolor2}{rgb}{0.92941,0.69412,0.12549}%
\definecolor{mycolor3}{rgb}{0.85098,0.32549,0.09804}%
\hspace{12pt}
\begin{tikzpicture}
\begin{axis}[%
width=0.2\linewidth,
height=0.2\linewidth,
every x tick label/.append style={font=\color{black}, font=\tiny},
every y tick label/.append style={font=\color{black}, font=\tiny},
at={(7.226in,0.894in)},
scale only axis,
xmin=0,
xmax=0.12,
xlabel style={font=\color{white!15!black}},
xlabel={\footnotesize{Geodesic Error ($\times 10^{-3}$)}},
ymin=0,
ymax=100.5,
ylabel style={font=\color{white!15!black}},
ylabel={\% correspondence},
axis background/.style={fill=white},
title style={font=\bfseries},
title={TOSCA},
xtick={0,0.02,0.04,0.06,0.08,0.10,0.12},
xticklabels={0,20,40,60,80,100,120},
xmajorgrids,
ymajorgrids,
legend style={at={(1.8,0)}, anchor=south east, legend cell align=left, align=left, draw=white!15!black}
]
\addplot [color=mycolor1, dotted, line width=1.5pt]
  table[row sep=crcr]{%
0	5.4100352112676\\
0.002	5.41355633802817\\
0.004	5.42961267605634\\
0.006	5.5268661971831\\
0.008	5.71489436619718\\
0.01	6.06091549295775\\
0.012	6.65862676056338\\
0.014	7.62052816901408\\
0.016	8.9625\\
0.018	10.741161971831\\
0.02	12.8639436619718\\
0.022	15.1135211267606\\
0.024	17.5117605633803\\
0.026	19.9489788732394\\
0.028	22.2608802816901\\
0.03	24.4944718309859\\
0.032	26.5658802816901\\
0.034	28.5311971830986\\
0.036	30.3412323943662\\
0.038	32.0659507042254\\
0.04	33.6868309859155\\
0.042	35.2190492957746\\
0.044	36.7185211267606\\
0.046	38.1864788732394\\
0.048	39.5976408450704\\
0.05	40.9734507042254\\
0.052	42.2899295774648\\
0.054	43.5688732394366\\
0.056	44.7980633802817\\
0.058	45.9573591549296\\
0.06	47.0579577464789\\
0.062	48.1289436619718\\
0.064	49.1503873239437\\
0.066	50.1127464788732\\
0.068	51.0519366197183\\
0.07	51.9282746478873\\
0.072	52.741514084507\\
0.074	53.503661971831\\
0.076	54.2491197183099\\
0.078	54.9469014084507\\
0.08	55.6153169014084\\
0.082	56.2547887323944\\
0.084	56.8344366197183\\
0.086	57.390985915493\\
0.088	57.9171478873239\\
0.09	58.4083450704225\\
0.092	58.8574295774648\\
0.094	59.2912323943662\\
0.096	59.7062676056338\\
0.098	60.1003521126761\\
0.1	60.4780281690141\\
0.102	60.8432394366197\\
0.104	61.1784154929577\\
0.106	61.5058098591549\\
0.108	61.8212323943662\\
0.11	62.1241901408451\\
0.112	62.4080633802817\\
0.114	62.6732394366197\\
0.116	62.9391197183099\\
0.118	63.1964084507042\\
0.12	63.4399295774648\\
};
\addlegendentry{\scriptsize standard (1) : 213.0}

\addplot [color=mycolor2, dotted, line width=1.5pt]
  table[row sep=crcr]{%
0	3.76007042253521\\
0.002	3.76323943661972\\
0.004	3.77633802816901\\
0.006	3.87183098591549\\
0.008	4.0393661971831\\
0.01	4.3456338028169\\
0.012	4.8637323943662\\
0.014	5.68669014084507\\
0.016	6.8805985915493\\
0.018	8.42422535211267\\
0.02	10.2415492957746\\
0.022	12.1595422535211\\
0.024	14.2177464788732\\
0.026	16.2797887323944\\
0.028	18.2732746478873\\
0.03	20.1821126760563\\
0.032	22.0400352112676\\
0.034	23.826161971831\\
0.036	25.5225704225352\\
0.038	27.1792957746479\\
0.04	28.7746126760563\\
0.042	30.3149295774648\\
0.044	31.8197183098592\\
0.046	33.3203521126761\\
0.048	34.759823943662\\
0.05	36.1817957746479\\
0.052	37.5558450704225\\
0.054	38.8851056338028\\
0.056	40.1569366197183\\
0.058	41.3681338028169\\
0.06	42.5298591549296\\
0.062	43.6809507042254\\
0.064	44.7617253521127\\
0.066	45.8004577464789\\
0.068	46.8105281690141\\
0.07	47.7763028169014\\
0.072	48.6786971830986\\
0.074	49.5575352112676\\
0.076	50.3992957746479\\
0.078	51.189823943662\\
0.08	51.9562323943662\\
0.082	52.6893661971831\\
0.084	53.3788028169014\\
0.086	54.0334507042254\\
0.088	54.6577112676056\\
0.09	55.2433802816901\\
0.092	55.7974647887324\\
0.094	56.3347887323944\\
0.096	56.8555985915493\\
0.098	57.3456338028169\\
0.1	57.8280281690141\\
0.102	58.2806690140845\\
0.104	58.7092605633803\\
0.106	59.1239436619718\\
0.108	59.5288732394366\\
0.11	59.917676056338\\
0.112	60.2840492957746\\
0.114	60.6379929577465\\
0.116	60.9854929577465\\
0.118	61.3196830985915\\
0.12	61.6362676056338\\
};
\addlegendentry{\scriptsize slanted (2) : 190.6}

\addplot [color=mycolor3, dotted, line width=1.5pt]
  table[row sep=crcr]{%
0	8.4212323943662\\
0.002	8.42528169014085\\
0.004	8.44429577464789\\
0.006	8.59778169014085\\
0.008	8.92443661971831\\
0.01	9.5244014084507\\
0.012	10.5168661971831\\
0.014	12.1003521126761\\
0.016	14.3147535211268\\
0.018	17.131161971831\\
0.02	20.2991901408451\\
0.022	23.6238732394366\\
0.024	27.053485915493\\
0.026	30.4323943661972\\
0.028	33.5794014084507\\
0.03	36.5446126760563\\
0.032	39.247323943662\\
0.034	41.7561971830986\\
0.036	44.0571126760563\\
0.038	46.2090845070423\\
0.04	48.2097183098592\\
0.042	50.092676056338\\
0.044	51.8794718309859\\
0.046	53.5819718309859\\
0.048	55.2052816901409\\
0.05	56.7818309859155\\
0.052	58.2657746478873\\
0.054	59.667676056338\\
0.056	60.9972887323944\\
0.058	62.2579929577465\\
0.06	63.4330281690141\\
0.062	64.5750352112676\\
0.064	65.6347535211268\\
0.066	66.638661971831\\
0.068	67.5739436619718\\
0.07	68.4688732394366\\
0.072	69.2918309859155\\
0.074	70.0533098591549\\
0.076	70.7912676056338\\
0.078	71.4618661971831\\
0.08	72.105\\
0.082	72.699823943662\\
0.084	73.2321830985915\\
0.086	73.7482746478873\\
0.088	74.2147535211268\\
0.09	74.6394366197183\\
0.092	75.0437676056338\\
0.094	75.4292605633803\\
0.096	75.7965492957747\\
0.098	76.1244366197183\\
0.1	76.4425704225352\\
0.102	76.7536267605634\\
0.104	77.0353521126761\\
0.106	77.3035915492958\\
0.108	77.5721478873239\\
0.11	77.8139436619718\\
0.112	78.0442957746479\\
0.114	78.2600704225352\\
0.116	78.4677816901409\\
0.118	78.6707042253521\\
0.12	78.8643661971831\\
};
\addlegendentry{\scriptsize \textbf{ours (3) : 125.7}}

\addplot [color=mycolor1, dashed, line width=1.5pt]
  table[row sep=crcr]{%
0	13.3526408450704\\
0.002	13.355176056338\\
0.004	13.3975352112676\\
0.006	13.8255985915493\\
0.008	14.7747887323944\\
0.01	16.0454577464789\\
0.012	17.745985915493\\
0.014	20.0226408450704\\
0.016	22.8311971830986\\
0.018	25.9996830985915\\
0.02	29.3761971830986\\
0.022	32.7446478873239\\
0.024	36.0320422535211\\
0.026	39.128661971831\\
0.028	41.9484507042254\\
0.03	44.5094718309859\\
0.032	46.7442957746479\\
0.034	48.7241901408451\\
0.036	50.4702464788732\\
0.038	52.0216901408451\\
0.04	53.4295070422535\\
0.042	54.7003169014084\\
0.044	55.886338028169\\
0.046	57.0167253521127\\
0.048	58.087323943662\\
0.05	59.1137323943662\\
0.052	60.0753169014084\\
0.054	60.9708098591549\\
0.056	61.8244014084507\\
0.058	62.6144718309859\\
0.06	63.3580281690141\\
0.062	64.060985915493\\
0.064	64.7159154929577\\
0.066	65.3380633802817\\
0.068	65.9315845070422\\
0.07	66.4879577464789\\
0.072	67.0187676056338\\
0.074	67.5127816901408\\
0.076	67.9913028169014\\
0.078	68.4444014084507\\
0.08	68.8850704225352\\
0.082	69.2968661971831\\
0.084	69.6833450704225\\
0.086	70.0640845070423\\
0.088	70.4235563380282\\
0.09	70.760985915493\\
0.092	71.0779225352113\\
0.094	71.380985915493\\
0.096	71.6818661971831\\
0.098	71.9538732394366\\
0.1	72.2198943661972\\
0.102	72.4834154929577\\
0.104	72.7486971830986\\
0.106	73.0017957746479\\
0.108	73.2351408450704\\
0.11	73.4639436619718\\
0.112	73.689823943662\\
0.114	73.9059154929577\\
0.116	74.115\\
0.118	74.3208802816901\\
0.12	74.5209507042254\\
};
\addlegendentry{\scriptsize (1) + ICP : 156.5}

\addplot [color=mycolor2, dashed, line width=1.5pt]
  table[row sep=crcr]{%
0	10.8090845070423\\
0.002	10.8111267605634\\
0.004	10.8567957746479\\
0.006	11.2957394366197\\
0.008	12.215\\
0.01	13.4290492957746\\
0.012	15.0475\\
0.014	17.2117605633803\\
0.016	19.7907042253521\\
0.018	22.6399647887324\\
0.02	25.618838028169\\
0.022	28.5377464788732\\
0.024	31.4010563380282\\
0.026	34.1097887323944\\
0.028	36.586161971831\\
0.03	38.8027816901408\\
0.032	40.779823943662\\
0.034	42.5311971830986\\
0.036	44.1155281690141\\
0.038	45.5364788732394\\
0.04	46.8469366197183\\
0.042	48.0729929577465\\
0.044	49.2299647887324\\
0.046	50.3613732394366\\
0.048	51.4306338028169\\
0.05	52.4467957746479\\
0.052	53.4066197183099\\
0.054	54.325176056338\\
0.056	55.1764084507042\\
0.058	55.9896830985915\\
0.06	56.755985915493\\
0.062	57.4947535211268\\
0.064	58.17\\
0.066	58.8322183098592\\
0.068	59.4397535211268\\
0.07	60.0164084507042\\
0.072	60.5836971830986\\
0.074	61.1071478873239\\
0.076	61.6062323943662\\
0.078	62.0929929577465\\
0.08	62.5503873239437\\
0.082	62.9828169014085\\
0.084	63.3902816901408\\
0.086	63.7806338028169\\
0.088	64.1496830985916\\
0.09	64.5085563380282\\
0.092	64.8548591549296\\
0.094	65.1799647887324\\
0.096	65.4922887323944\\
0.098	65.7935915492958\\
0.1	66.0828873239437\\
0.102	66.3641549295775\\
0.104	66.6317957746479\\
0.106	66.8969366197183\\
0.108	67.1502816901408\\
0.11	67.4022887323944\\
0.112	67.6407042253521\\
0.114	67.8721126760563\\
0.116	68.1009154929578\\
0.118	68.3205633802817\\
0.12	68.5407394366197\\
};
\addlegendentry{\scriptsize (2) + ICP : 163.2}

\addplot [color=mycolor3, dashed, line width=1.5pt]
  table[row sep=crcr]{%
0	16.4060563380282\\
0.002	16.4098943661972\\
0.004	16.4673591549296\\
0.006	17.0533802816901\\
0.008	18.3321830985916\\
0.01	20.0301408450704\\
0.012	22.3018661971831\\
0.014	25.3587323943662\\
0.016	29.038485915493\\
0.018	33.1011971830986\\
0.02	37.2430281690141\\
0.022	41.2646126760563\\
0.024	45.1391549295775\\
0.026	48.7736971830986\\
0.028	52.0493661971831\\
0.03	54.9858802816901\\
0.032	57.5321126760563\\
0.034	59.7572183098591\\
0.036	61.6953873239437\\
0.038	63.4078521126761\\
0.04	64.9420422535211\\
0.042	66.3221830985916\\
0.044	67.6108098591549\\
0.046	68.8200352112676\\
0.048	69.9582042253521\\
0.05	71.0393661971831\\
0.052	72.0332746478873\\
0.054	72.9578521126761\\
0.056	73.8186267605634\\
0.058	74.6126408450704\\
0.06	75.3541197183099\\
0.062	76.052323943662\\
0.064	76.6969718309859\\
0.066	77.3054225352113\\
0.068	77.8721126760563\\
0.07	78.4018309859155\\
0.072	78.9048943661972\\
0.074	79.3682746478873\\
0.076	79.8065492957746\\
0.078	80.2388028169014\\
0.08	80.6365492957747\\
0.082	81.0105985915493\\
0.084	81.3583098591549\\
0.086	81.6853169014085\\
0.088	81.9950352112676\\
0.09	82.2732394366197\\
0.092	82.5384507042254\\
0.094	82.7987323943662\\
0.096	83.040985915493\\
0.098	83.2719014084507\\
0.1	83.496338028169\\
0.102	83.7126056338028\\
0.104	83.915\\
0.106	84.1111971830986\\
0.108	84.2996126760563\\
0.11	84.4783450704225\\
0.112	84.6504929577465\\
0.114	84.8141549295775\\
0.116	84.9720422535211\\
0.118	85.1237323943662\\
0.12	85.2646830985915\\
};
\addlegendentry{\scriptsize \textbf{(3) + ICP : 90.2}}

\addplot [color=mycolor1, line width=1.5pt]
  table[row sep=crcr]{%
0	15.1117605633803\\
0.002	15.1167253521127\\
0.004	15.1847887323944\\
0.006	15.8272535211268\\
0.008	17.1460211267606\\
0.01	18.9053521126761\\
0.012	21.300176056338\\
0.014	24.4943661971831\\
0.016	28.265985915493\\
0.018	32.3170774647887\\
0.02	36.3896478873239\\
0.022	40.2953169014084\\
0.024	44.0471126760563\\
0.026	47.539823943662\\
0.028	50.7148591549296\\
0.03	53.5956690140845\\
0.032	56.1337323943662\\
0.034	58.3829577464789\\
0.036	60.3835915492958\\
0.038	62.1687323943662\\
0.04	63.8144366197183\\
0.042	65.3079929577465\\
0.044	66.7062676056338\\
0.046	68.0216549295775\\
0.048	69.301338028169\\
0.05	70.4833450704225\\
0.052	71.5884507042254\\
0.054	72.6316197183099\\
0.056	73.6018309859155\\
0.058	74.5016901408451\\
0.06	75.3405633802817\\
0.062	76.1341549295775\\
0.064	76.8651056338028\\
0.066	77.5618661971831\\
0.068	78.2199295774648\\
0.07	78.8249295774648\\
0.072	79.4072535211268\\
0.074	79.9645422535211\\
0.076	80.4685563380282\\
0.078	80.958661971831\\
0.08	81.4202464788732\\
0.082	81.8513732394366\\
0.084	82.2505985915493\\
0.086	82.6361971830986\\
0.088	82.990985915493\\
0.09	83.3356690140845\\
0.092	83.6479225352113\\
0.094	83.9454929577465\\
0.096	84.2416197183099\\
0.098	84.5181338028169\\
0.1	84.7861971830986\\
0.102	85.0397535211268\\
0.104	85.287323943662\\
0.106	85.5195070422535\\
0.108	85.7457394366197\\
0.11	85.9555985915493\\
0.112	86.1583450704225\\
0.114	86.3563732394366\\
0.116	86.5497183098592\\
0.118	86.7375\\
0.12	86.9153873239437\\
};
\addlegendentry{\scriptsize (1) + BCICP : 81.8}

\addplot [color=mycolor2, line width=1.5pt]
  table[row sep=crcr]{%
0	13.1003521126761\\
0.002	13.1036971830986\\
0.004	13.1649647887324\\
0.006	13.741161971831\\
0.008	14.9232394366197\\
0.01	16.4937323943662\\
0.012	18.6328169014085\\
0.014	21.5034507042254\\
0.016	24.8584154929577\\
0.018	28.4342957746479\\
0.02	32.0092253521127\\
0.022	35.4307042253521\\
0.024	38.700985915493\\
0.026	41.7442605633803\\
0.028	44.4987676056338\\
0.03	46.9974295774648\\
0.032	49.1916901408451\\
0.034	51.1223943661972\\
0.036	52.8555633802817\\
0.038	54.3998591549296\\
0.04	55.8110915492958\\
0.042	57.1132746478873\\
0.044	58.3166197183099\\
0.046	59.4724647887324\\
0.048	60.5725\\
0.05	61.6045070422535\\
0.052	62.5679225352113\\
0.054	63.4733450704225\\
0.056	64.3345070422535\\
0.058	65.1450704225352\\
0.06	65.8934507042254\\
0.062	66.5817957746479\\
0.064	67.2249647887324\\
0.066	67.8268309859155\\
0.068	68.4102464788732\\
0.07	68.9555633802817\\
0.072	69.4591549295775\\
0.074	69.9371126760563\\
0.076	70.3857746478873\\
0.078	70.8140492957747\\
0.08	71.2125352112676\\
0.082	71.5967253521127\\
0.084	71.9479929577465\\
0.086	72.2946126760563\\
0.088	72.6209507042254\\
0.09	72.9191901408451\\
0.092	73.2090492957746\\
0.094	73.4833802816901\\
0.096	73.750176056338\\
0.098	73.9981690140845\\
0.1	74.2365140845071\\
0.102	74.4706690140845\\
0.104	74.6972887323944\\
0.106	74.9153169014084\\
0.108	75.1342605633803\\
0.11	75.3419014084507\\
0.112	75.5455281690141\\
0.114	75.7489436619718\\
0.116	75.9406690140845\\
0.118	76.1275\\
0.12	76.3146830985916\\
};
\addlegendentry{\scriptsize (2) + BCICP}

\addplot [color=mycolor3, line width=1.5pt]
  table[row sep=crcr]{%
0	17.772323943662\\
0.002	17.7781338028169\\
0.004	17.8508802816901\\
0.006	18.5774647887324\\
0.008	20.0769718309859\\
0.01	22.0627464788732\\
0.012	24.7708098591549\\
0.014	28.3721830985915\\
0.016	32.6573943661972\\
0.018	37.2538732394366\\
0.02	41.8337676056338\\
0.022	46.2029929577465\\
0.024	50.3511267605634\\
0.026	54.2004929577465\\
0.028	57.6595070422535\\
0.03	60.8017605633803\\
0.032	63.5354225352113\\
0.034	65.9369366197183\\
0.036	68.0327816901408\\
0.038	69.8869718309859\\
0.04	71.5686971830986\\
0.042	73.0927464788732\\
0.044	74.5174295774648\\
0.046	75.8344014084507\\
0.048	77.0765845070423\\
0.05	78.2386267605634\\
0.052	79.2954577464789\\
0.054	80.2773591549296\\
0.056	81.1992605633803\\
0.058	82.0366901408451\\
0.06	82.8182042253521\\
0.062	83.5363732394366\\
0.064	84.1738028169014\\
0.066	84.7778873239437\\
0.068	85.3470422535211\\
0.07	85.8669366197183\\
0.072	86.3484507042254\\
0.074	86.7907394366197\\
0.076	87.2005633802817\\
0.078	87.5985915492958\\
0.08	87.9734507042254\\
0.082	88.313485915493\\
0.084	88.624823943662\\
0.086	88.9173591549296\\
0.088	89.1944366197183\\
0.09	89.4443661971831\\
0.092	89.6739084507042\\
0.094	89.8987676056338\\
0.096	90.1013028169014\\
0.098	90.2942957746479\\
0.1	90.4759154929578\\
0.102	90.6516901408451\\
0.104	90.8203521126761\\
0.106	90.9747535211268\\
0.108	91.1245774647887\\
0.11	91.263661971831\\
0.112	91.3951408450704\\
0.114	91.5148591549296\\
0.116	91.6333098591549\\
0.118	91.7476408450704\\
0.12	91.8551408450704\\
};
\addlegendentry{\scriptsize \textbf{(3) + BCICP : 61.5}}
\end{axis}
\end{tikzpicture}%

%% file: figures_tex/diff_desc_faust_direct.tex
%
%
\definecolor{mycolor1}{rgb}{0.00000,0.44706,0.74118}%
\definecolor{mycolor2}{rgb}{0.92941,0.69412,0.12549}%
\definecolor{mycolor3}{rgb}{0.85098,0.32549,0.09804}%
\pgfplotsset{scaled x ticks=false}
\pgfplotsset{
compat=1.11,
legend image code/.code={
\draw[mark repeat=2,mark phase=2]
plot coordinates {
(0cm,0cm)
(0.15cm,0cm)        
(0.3cm,0cm)         
};%
}
}
\begin{tikzpicture}

\begin{axis}[%
width=0.6\linewidth,
height=0.4\linewidth,
at={(1.472in,1.026in)},
scale only axis,
xmin=1,
xmax=5,
ymode=log,
xtick={1,2,3,4,5},
xticklabels={{1},{2},{5},{10},{20}},
ytick={0.05, 0.1},
yticklabels={0.05,0.1},
xlabel style={font=\color{white!15!black}},
xlabel={\# descriptors},
ymin=0,
ymax=0.25,
ylabel style={font=\color{white!15!black}},
ylabel={Average direct error},
axis background/.style={fill=white},
xmajorgrids,
ymajorgrids,
legend style={at={(1.42,0.03)}, anchor=south east, legend cell align=left, align=left, draw=white!15!black, fill=white, fill opacity=0.6, draw opacity=1, text opacity=1},
]
\addplot [color=mycolor1, line width=1.5 pt]
  table[row sep=crcr]{%
1	0.20173\\
2	0.070664\\
3	0.054841\\
4	0.053226\\
5	0.051437\\
};
\addlegendentry{\scriptsize standard (1)}

\addplot [color=mycolor2, line width=1.5 pt]
  table[row sep=crcr]{%
1	0.18815\\
2	0.13763\\
3	0.10923\\
4	0.10755\\
5	0.1063\\
};
\addlegendentry{\scriptsize slanted (2)}

\addplot [color=mycolor3, line width=1.5 pt]
  table[row sep=crcr]{%
1	0.11276\\
2	0.056063\\
3	0.045678\\
4	0.042554\\
5	0.041212\\
};
\addlegendentry{\scriptsize ours (3)}

\addplot [color=mycolor1, dashed, line width=1.5 pt]
  table[row sep=crcr]{%
1	0.11649\\
2	0.037365\\
3	0.031668\\
4	0.035089\\
5	0.032601\\
};
\addlegendentry{\scriptsize (1) + ICP}

\addplot [color=mycolor2, dashed, line width=1.5 pt]
  table[row sep=crcr]{%
1	0.14135\\
2	0.091973\\
3	0.074516\\
4	0.072879\\
5	0.073227\\
};
\addlegendentry{\scriptsize (2) + ICP}

\addplot [color=mycolor3, dashed, line width=1.5 pt]
  table[row sep=crcr]{%
1	0.065958\\
2	0.031577\\
3	0.029414\\
4	0.028862\\
5	0.02752\\
};
\addlegendentry{\scriptsize (3) + ICP}

\end{axis}
\end{tikzpicture}%

%% file: figures_tex/diff_sigma_normalized_pervtx.tex
%
%
\definecolor{mycolor1}{rgb}{0.00000,0.44700,0.74100}%
\definecolor{mycolor2}{rgb}{0.85098,0.32549,0.09804}%
\pgfplotsset{
compat=1.11,
legend image code/.code={
\draw[mark repeat=2,mark phase=2]
plot coordinates {
(0cm,0cm)
(0.0cm,0cm)        
(0.3cm,0cm)         
};%
}
}
\begin{tikzpicture}

\begin{axis}[%
width=0.35\linewidth,
height=0.3\linewidth,
at={(6.441in,0.894in)},
scale only axis,
xmin=0,
xmax=2,
xlabel style={font=\color{white!15!black}},
xlabel={$\gamma$},
ymin=0.02,
ymax=0.16,
every x tick label/.append style={font=\color{black}, font=\tiny},
every y tick label/.append style={font=\color{black}, font=\tiny},
ylabel style={font=\color{white!15!black}},
ylabel={Average error},
axis background/.style={fill=white},
title style={font=\bfseries, align=center},
title={\tiny Different $\gamma$\\ \tiny $\Vert M_{Re}(\gamma) \odot C \Vert_F^2 + \Vert M_{Im}(\gamma) \odot C\Vert_F^2$},
yticklabel style={
    /pgf/number format/fixed,
    /pgf/number format/precision=5
},
scaled y ticks=false,
xmajorgrids,
ymajorgrids,
legend style={at={(0,0.7)}, anchor=west, legend cell align=left, align=left, draw=none, fill opacity=0, text opacity=1},
legend style={inner sep=0pt}
]
\addplot [color=mycolor1, line width=2.0pt]
  table[row sep=crcr]{%
0	0.0623232778110984\\
2	0.0623232778110984\\
};
\addlegendentry{\scriptsize standard}

\addplot [color=mycolor1, dashed, line width=2.0pt]
  table[row sep=crcr]{%
0	0.0369725128322993\\
2	0.0369725128322993\\
};
\addlegendentry{\scriptsize standard + ICP}

\addplot [color=mycolor2, line width=2.0pt]
  table[row sep=crcr]{%
0.025	0.0583152368421981\\
0.05	0.0521026272527965\\
0.075	0.0479194291687966\\
0.1	0.0443987808623975\\
0.125	0.0414164434206974\\
0.15	0.0390856300281968\\
0.175	0.0378017628686979\\
0.2	0.0369727399182985\\
0.225	0.0365386788215979\\
0.25	0.0362931853709975\\
0.275	0.0362342441398975\\
0.3	0.0361314092052974\\
0.325	0.0362422869789977\\
0.35	0.0362734897058975\\
0.375	0.0363799626078973\\
0.4	0.0364435168204975\\
0.425	0.0366122288034978\\
0.45	0.036759590376898\\
0.475	0.0369067219745983\\
0.5	0.0371222591543984\\
0.525	0.0372981830416987\\
0.55	0.0375146266414985\\
0.575	0.0377592913342988\\
0.6	0.0379744867329989\\
0.7	0.0394444530496004\\
0.8	0.0413471053966815\\
0.9	0.0434959573925015\\
1	0.0483048387134006\\
1.1	0.0613125641601987\\
1.2	0.0699588645933976\\
1.3	0.0811066329265983\\
1.4	0.0906580449131988\\
1.5	0.0996321135478975\\
1.6	0.110989262251892\\
1.7	0.121583089024591\\
1.8	0.13235149037609\\
1.9	0.143169955066397\\
2	0.154796697072591\\
};
\addlegendentry{\scriptsize \textbf{ours}}

\addplot [color=mycolor2, dashed, line width=2.0pt]
  table[row sep=crcr]{%
0.025	0.0383934680210978\\
0.05	0.033326151640699\\
0.075	0.0307348240462399\\
0.1	0.029224551507539\\
0.125	0.0272981572234594\\
0.15	0.0263066336647999\\
0.175	0.0254669636449993\\
0.2	0.0249498077236993\\
0.225	0.0245837765132995\\
0.25	0.0243806315742994\\
0.275	0.0242751123744992\\
0.3	0.0241783891945991\\
0.325	0.0242675997474994\\
0.35	0.0241152016695994\\
0.375	0.0241979016097999\\
0.4	0.0241361919454995\\
0.425	0.0242285151607993\\
0.45	0.0244394103092991\\
0.475	0.0247862654553993\\
0.5	0.0246673034437999\\
0.525	0.0249318941739996\\
0.55	0.0249694644020994\\
0.575	0.0251315978918998\\
0.6	0.0254560058948996\\
0.7	0.0260108274538997\\
0.8	0.0268856713556994\\
0.9	0.0283567982500997\\
1	0.0305497878734997\\
1.1	0.0359991089231993\\
1.2	0.0414407196311997\\
1.3	0.0451297832754989\\
1.4	0.0487173973353972\\
1.5	0.0530933142122965\\
1.6	0.0579054850009979\\
1.7	0.0628788951940965\\
1.8	0.067025554674599\\
1.9	0.0718594538182979\\
2	0.0742803619985962\\
};
\addlegendentry{\scriptsize \textbf{ours + ICP}}

\end{axis}
\end{tikzpicture}%

%% file: figures_tex/diff_weight_normalized_pervtx.tex
%
%
\definecolor{mycolor1}{rgb}{0.00000,0.44700,0.74100}%
\definecolor{mycolor2}{rgb}{0.85098,0.32549,0.09804}%
\pgfplotsset{
compat=1.11,
legend image code/.code={
\draw[mark repeat=2,mark phase=2]
plot coordinates {
(0cm,0cm)
(0.0cm,0cm)        
(0.3cm,0cm)         
};%
}
}
\hspace{-6pt}
\begin{tikzpicture}
\begin{axis}[%
width=0.35\linewidth,
height=0.3\linewidth,
at={(6.441in,0.894in)},
scale only axis,
xmin=0,
xmax=1,
xlabel style={font=\color{white!15!black}},
xlabel={$w$},
ymin=0.02,
ymax=0.065,
ylabel style={font=\color{white!15!black}},
axis background/.style={fill=white},
every x tick label/.append style={font=\color{black}, font=\tiny},
every y tick label/.append style={font=\color{black}, font=\tiny},
title style={font=\bfseries, align=center},
title={\tiny Different weight $w$ \\ \tiny $M_{res} = w M_{Im}^2 + (1-w)M_{Re}^2$},
yticklabel style={
    /pgf/number format/fixed,
    /pgf/number format/precision=5
},
scaled y ticks=false,
xmajorgrids,
ymajorgrids,
legend style={at={(1.03,0.5)}, anchor=west, legend cell align=left, align=left, draw=white!15!black}
]
\addplot [color=mycolor1, line width=2.0pt]
  table[row sep=crcr]{%
0	0.0624537966860984\\
1	0.0624537966860984\\
};
\addlegendentry{standard}

\addplot [color=mycolor1, dashed, line width=2.0pt]
  table[row sep=crcr]{%
0	0.0369103798652996\\
1	0.0369103798652996\\
};
\addlegendentry{standard + ICP}

\addplot [color=mycolor2, line width=2.0pt]
  table[row sep=crcr]{%
0	0.0419940313321026\\
0.05	0.0407565551437019\\
0.1	0.0398833051566007\\
0.15	0.0392642518791006\\
0.2	0.0387585556273003\\
0.25	0.0382881544195996\\
0.3	0.0380229839208986\\
0.35	0.0377277810931984\\
0.4	0.0374842654237985\\
0.45	0.0372936728917986\\
0.5	0.0371117676156985\\
0.55	0.0369603658568981\\
0.6	0.0367906834243979\\
0.65	0.0366871721588974\\
0.7	0.0366280014479973\\
0.75	0.0366138638427982\\
0.8	0.036572327902198\\
0.85	0.0365928149390981\\
0.9	0.0367247735554983\\
0.95	0.0371088450001986\\
1	0.0381972682613983\\
};
\addlegendentry{complRes}

\addplot [color=mycolor2, dashed, line width=2.0pt]
  table[row sep=crcr]{%
0	0.0271738263983991\\
0.05	0.0267111949965993\\
0.1	0.0260802357187996\\
0.15	0.0259784179599999\\
0.2	0.0255595201852998\\
0.25	0.0255742205973997\\
0.3	0.0253323837697991\\
0.35	0.0253102468341993\\
0.4	0.0249880831742996\\
0.45	0.0246980628605993\\
0.5	0.0246723603679996\\
0.55	0.024796041372\\
0.6	0.0243854913708995\\
0.65	0.0244769438813994\\
0.7	0.0244786891499993\\
0.75	0.0244737875036996\\
0.8	0.0245104325869996\\
0.85	0.0246135441591997\\
0.9	0.0248716994911994\\
0.95	0.0252750233820993\\
1	0.0259257198668997\\
};
\addlegendentry{complRes + ICP}
\legend{}
\end{axis}
\end{tikzpicture}%

%% file: figures_tex/maskWeight_pervtx.tex
%
%
\definecolor{mycolor1}{rgb}{0.00000,0.44700,0.74100}%
\definecolor{mycolor2}{rgb}{0.92941,0.69412,0.12549}%
\definecolor{mycolor3}{rgb}{0.85098,0.32549,0.09804}%
\begin{tikzpicture}

\begin{axis}[%
width=0.4\linewidth,
height=0.4\linewidth,
at={(1.904in,0.693in)},
scale only axis,
xmode=log,
xmin=0,
xmax=10,
xminorticks=true,
ymin=0.04,
ymax=0.35,
every x tick label/.append style={font=\color{black}, font=\tiny},
every y tick label/.append style={font=\color{black}, font=\tiny},
axis background/.style={fill=white},
xmajorgrids,
xminorgrids,
ymajorgrids,
title={per-vertex measure},
legend style={at={(0.671,0.75)}, anchor=south west, legend cell align=left, align=left, draw=white!15!black}
]
\addplot [color=mycolor1, line width=1.0pt, mark size=0.6pt, mark=*, mark options={solid, mycolor1}]
  table[row sep=crcr]{%
0	0.34308\\
1e-06	0.32713\\
1.1e-05	0.25562\\
2.1e-05	0.2276\\
3.1e-05	0.21102\\
4.1e-05	0.20053\\
5.1e-05	0.19134\\
6.1e-05	0.18499\\
7.1e-05	0.17866\\
8.1e-05	0.17444\\
9.1e-05	0.17056\\
0.0002	0.15132\\
0.0003	0.14224\\
0.0004	0.1363\\
0.0005	0.13142\\
0.0006	0.12735\\
0.0007	0.12517\\
0.0008	0.12297\\
0.0009	0.12111\\
0.001	0.11933\\
0.002	0.10925\\
0.003	0.1033\\
0.004	0.098591\\
0.005	0.094639\\
0.006	0.092207\\
0.007	0.090654\\
0.008	0.089904\\
0.009	0.089209\\
0.01	0.08852\\
0.011	0.088641\\
0.012	0.087717\\
0.013	0.087634\\
0.014	0.087789\\
0.015	0.087778\\
0.016	0.087556\\
0.017	0.088057\\
0.018	0.087831\\
0.019	0.088255\\
0.02	0.088653\\
0.025	0.089001\\
0.03	0.090034\\
0.035	0.091244\\
0.04	0.092079\\
0.045	0.093776\\
0.05	0.094782\\
0.055	0.095461\\
0.06	0.096072\\
0.065	0.096751\\
0.07	0.097437\\
0.075	0.097903\\
0.08	0.098633\\
0.085	0.099156\\
0.09	0.09894\\
0.095	0.10018\\
0.1	0.1004\\
0.2	0.10651\\
0.3	0.11154\\
0.4	0.11458\\
0.5	0.11748\\
0.6	0.12026\\
0.7	0.12289\\
0.8	0.12515\\
0.9	0.12755\\
1	0.12965\\
3	0.15837\\
5	0.17705\\
7	0.19205\\
9	0.20542\\
};
\addlegendentry{standard}

\addplot [color=mycolor2, line width=1.0pt, mark size=0.6pt, mark=*, mark options={solid, mycolor2}]
  table[row sep=crcr]{%
0	0.34308\\
1e-06	0.16491\\
1.1e-05	0.076483\\
2.1e-05	0.068257\\
3.1e-05	0.065503\\
4.1e-05	0.063896\\
5.1e-05	0.063175\\
6.1e-05	0.062809\\
7.1e-05	0.062622\\
8.1e-05	0.062525\\
9.1e-05	0.062664\\
0.0002	0.063749\\
0.0003	0.065356\\
0.0004	0.066886\\
0.0005	0.068222\\
0.0006	0.06935\\
0.0007	0.07024\\
0.0008	0.071061\\
0.0009	0.071795\\
0.001	0.072481\\
0.002	0.077635\\
0.003	0.080687\\
0.004	0.083029\\
0.005	0.08475\\
0.006	0.086282\\
0.007	0.087607\\
0.008	0.088741\\
0.009	0.089595\\
0.01	0.090383\\
0.011	0.091041\\
0.012	0.091671\\
0.013	0.092225\\
0.014	0.092846\\
0.015	0.093353\\
0.016	0.093934\\
0.017	0.094481\\
0.018	0.094979\\
0.019	0.09543\\
0.02	0.095872\\
0.025	0.098219\\
0.03	0.099914\\
0.035	0.10156\\
0.04	0.10325\\
0.045	0.10461\\
0.05	0.10594\\
0.055	0.10725\\
0.06	0.10853\\
0.065	0.10989\\
0.07	0.111\\
0.075	0.11191\\
0.08	0.11288\\
0.085	0.11394\\
0.09	0.115\\
0.095	0.11587\\
0.1	0.11683\\
0.2	0.12879\\
0.3	0.13622\\
0.4	0.14148\\
0.5	0.14589\\
0.6	0.14932\\
0.7	0.15228\\
0.8	0.1546\\
0.9	0.15679\\
1	0.15887\\
3	0.17912\\
5	0.18949\\
7	0.19804\\
9	0.20487\\
};
\addlegendentry{slanted}

\addplot [color=mycolor3, line width=1.0pt, mark size=0.6pt, mark=*, mark options={solid, mycolor3}]
  table[row sep=crcr]{%
0	0.34308\\
1e-06	0.276\\
1.1e-05	0.16213\\
2.1e-05	0.13131\\
3.1e-05	0.11233\\
4.1e-05	0.10165\\
5.1e-05	0.09377\\
6.1e-05	0.088277\\
7.1e-05	0.083857\\
8.1e-05	0.080124\\
9.1e-05	0.077433\\
0.0002	0.064282\\
0.0003	0.060367\\
0.0004	0.058478\\
0.0005	0.057341\\
0.0006	0.056573\\
0.0007	0.056102\\
0.0008	0.055851\\
0.0009	0.055781\\
0.001	0.055666\\
0.002	0.056156\\
0.003	0.057119\\
0.004	0.058081\\
0.005	0.059144\\
0.006	0.060144\\
0.007	0.060997\\
0.008	0.061577\\
0.009	0.062325\\
0.01	0.062966\\
0.011	0.063606\\
0.012	0.064305\\
0.013	0.06488\\
0.014	0.065433\\
0.015	0.065909\\
0.016	0.066318\\
0.017	0.066743\\
0.018	0.067156\\
0.019	0.06754\\
0.02	0.067925\\
0.025	0.069548\\
0.03	0.070897\\
0.035	0.072003\\
0.04	0.072993\\
0.045	0.073958\\
0.05	0.074819\\
0.055	0.075643\\
0.06	0.076229\\
0.065	0.07699\\
0.07	0.077612\\
0.075	0.078144\\
0.08	0.07858\\
0.085	0.079141\\
0.09	0.079592\\
0.095	0.080096\\
0.1	0.080457\\
0.2	0.086736\\
0.3	0.090847\\
0.4	0.093558\\
0.5	0.095952\\
0.6	0.098167\\
0.7	0.10041\\
0.8	0.10282\\
0.9	0.10496\\
1	0.10705\\
3	0.13565\\
5	0.15615\\
7	0.17221\\
9	0.18407\\
};
\addlegendentry{ours}
\legend{}
\end{axis}
\end{tikzpicture}%

%% file: figures_tex/maskWeight_direct.tex
%
%
\definecolor{mycolor1}{rgb}{0.00000,0.44700,0.74100}%
\definecolor{mycolor2}{rgb}{0.92941,0.69412,0.12549}%
\definecolor{mycolor3}{rgb}{0.85098,0.32549,0.09804}%
\begin{tikzpicture}

\begin{axis}[%
width=0.4\linewidth,
height=0.4\linewidth,
at={(1.904in,0.693in)},
scale only axis,
xmode=log,
xmin=0,
xmax=10,
xminorticks=true,
ymin=0.1,
ymax=0.42,
every x tick label/.append style={font=\color{black}, font=\tiny},
every y tick label/.append style={font=\color{black}, font=\tiny},
axis background/.style={fill=white},
xmajorgrids,
xminorgrids,
ymajorgrids,
title = {direct measure},
legend style={legend cell align=left, align=left}
]
\addplot [color=mycolor1, line width=1.0pt, mark size=0.6pt, mark=*, mark options={solid, mycolor1}]
  table[row sep=crcr]{%
0	0.40556\\
1e-06	0.39152\\
1.1e-05	0.32949\\
2.1e-05	0.30423\\
3.1e-05	0.28819\\
4.1e-05	0.27771\\
5.1e-05	0.26849\\
6.1e-05	0.26157\\
7.1e-05	0.25497\\
8.1e-05	0.25047\\
9.1e-05	0.24601\\
0.0002	0.22338\\
0.0003	0.21223\\
0.0004	0.20486\\
0.0005	0.19882\\
0.0006	0.19437\\
0.0007	0.19156\\
0.0008	0.18896\\
0.0009	0.18658\\
0.001	0.18467\\
0.002	0.17266\\
0.003	0.16557\\
0.004	0.16027\\
0.005	0.15571\\
0.006	0.15307\\
0.007	0.15172\\
0.008	0.1513\\
0.009	0.15075\\
0.01	0.1504\\
0.011	0.1513\\
0.012	0.15099\\
0.013	0.15119\\
0.014	0.15217\\
0.015	0.15262\\
0.016	0.15311\\
0.017	0.15452\\
0.018	0.15477\\
0.019	0.15577\\
0.02	0.15692\\
0.025	0.15962\\
0.03	0.1627\\
0.035	0.16555\\
0.04	0.16796\\
0.045	0.17104\\
0.05	0.17302\\
0.055	0.17461\\
0.06	0.17617\\
0.065	0.1775\\
0.07	0.17883\\
0.075	0.17989\\
0.08	0.18112\\
0.085	0.18224\\
0.09	0.1825\\
0.095	0.18425\\
0.1	0.18463\\
0.2	0.19584\\
0.3	0.20297\\
0.4	0.20715\\
0.5	0.21088\\
0.6	0.21422\\
0.7	0.21736\\
0.8	0.21993\\
0.9	0.22237\\
1	0.22465\\
3	0.25256\\
5	0.27025\\
7	0.28437\\
9	0.29655\\
};
\addlegendentry{Standard}

\addplot [color=mycolor2, line width=1.0pt, mark size=0.6pt, mark=*, mark options={solid, mycolor2}]
  table[row sep=crcr]{%
0	0.40556\\
1e-06	0.25369\\
1.1e-05	0.15364\\
2.1e-05	0.1422\\
3.1e-05	0.13821\\
4.1e-05	0.13662\\
5.1e-05	0.13597\\
6.1e-05	0.13607\\
7.1e-05	0.13643\\
8.1e-05	0.13658\\
9.1e-05	0.1375\\
0.0002	0.1439\\
0.0003	0.14904\\
0.0004	0.15298\\
0.0005	0.15601\\
0.0006	0.15841\\
0.0007	0.16023\\
0.0008	0.16199\\
0.0009	0.16358\\
0.001	0.16492\\
0.002	0.17317\\
0.003	0.17774\\
0.004	0.1813\\
0.005	0.18378\\
0.006	0.1862\\
0.007	0.18803\\
0.008	0.1897\\
0.009	0.19102\\
0.01	0.19204\\
0.011	0.19307\\
0.012	0.19388\\
0.013	0.19474\\
0.014	0.19551\\
0.015	0.19621\\
0.016	0.19685\\
0.017	0.19739\\
0.018	0.19808\\
0.019	0.19845\\
0.02	0.19898\\
0.025	0.20177\\
0.03	0.20376\\
0.035	0.20577\\
0.04	0.20761\\
0.045	0.20908\\
0.05	0.21046\\
0.055	0.21164\\
0.06	0.21291\\
0.065	0.21425\\
0.07	0.21529\\
0.075	0.21618\\
0.08	0.21707\\
0.085	0.21811\\
0.09	0.21901\\
0.095	0.21991\\
0.1	0.22082\\
0.2	0.23227\\
0.3	0.23955\\
0.4	0.24467\\
0.5	0.24916\\
0.6	0.25247\\
0.7	0.2553\\
0.8	0.25736\\
0.9	0.25926\\
1	0.26133\\
3	0.27922\\
5	0.28853\\
7	0.29592\\
9	0.30177\\
};
\addlegendentry{Slanted}

\addplot [color=mycolor3, line width=1.0pt, mark size=0.6pt, mark=*, mark options={solid, mycolor3}]
  table[row sep=crcr]{%
0	0.40556\\
1e-06	0.3546\\
1.1e-05	0.24532\\
2.1e-05	0.20842\\
3.1e-05	0.18527\\
4.1e-05	0.17105\\
5.1e-05	0.16023\\
6.1e-05	0.15277\\
7.1e-05	0.14654\\
8.1e-05	0.14144\\
9.1e-05	0.13719\\
0.0002	0.11764\\
0.0003	0.11121\\
0.0004	0.10745\\
0.0005	0.10531\\
0.0006	0.10412\\
0.0007	0.10332\\
0.0008	0.10285\\
0.0009	0.10269\\
0.001	0.10261\\
0.002	0.10492\\
0.003	0.10864\\
0.004	0.1124\\
0.005	0.11625\\
0.006	0.11986\\
0.007	0.12309\\
0.008	0.12574\\
0.009	0.12818\\
0.01	0.1304\\
0.011	0.13259\\
0.012	0.13442\\
0.013	0.13605\\
0.014	0.13758\\
0.015	0.13904\\
0.016	0.14036\\
0.017	0.1416\\
0.018	0.14267\\
0.019	0.14377\\
0.02	0.14491\\
0.025	0.14915\\
0.03	0.15257\\
0.035	0.15539\\
0.04	0.15807\\
0.045	0.16035\\
0.05	0.16241\\
0.055	0.16412\\
0.06	0.16561\\
0.065	0.16712\\
0.07	0.16859\\
0.075	0.16994\\
0.08	0.17093\\
0.085	0.17209\\
0.09	0.17316\\
0.095	0.17418\\
0.1	0.17503\\
0.2	0.18583\\
0.3	0.19203\\
0.4	0.19572\\
0.5	0.19835\\
0.6	0.20077\\
0.7	0.20332\\
0.8	0.20526\\
0.9	0.20673\\
1	0.20825\\
3	0.23389\\
5	0.25204\\
7	0.26713\\
9	0.27761\\
};
\addlegendentry{Ours}

\end{axis}
\end{tikzpicture}%

%% file: figures_tex/res_shrec_aveErr.tex
%
%
\definecolor{mycolor1}{rgb}{0.00000,0.44700,0.74100}%
\definecolor{mycolor2}{rgb}{0.92941,0.69412,0.12549}%
\definecolor{mycolor3}{rgb}{0.85098,0.32549,0.09804}%
\pgfplotsset{
compat=1.11,
legend image code/.code={
\draw[mark repeat=2,mark phase=2]
plot coordinates {
(0cm,0cm)
(0.0cm,0cm)        
(0.3cm,0cm)         
};%
}
}
\begin{tikzpicture}

\begin{axis}[%
width=0.6\linewidth,
height=0.35\linewidth,
at={(7.226in,0.894in)},
scale only axis,
xmin=0,
xmax=380,
xlabel style={font=\color{white!15!black}},
xlabel={\# pair ID},
ymin=0,
ymax=0.35,
xmajorgrids,
ymajorgrids,
ylabel style={font=\color{white!15!black}},
ylabel={Average error},
axis background/.style={fill=white},
every x tick label/.append style={font=\color{black}, font=\scriptsize},
every y tick label/.append style={font=\color{black}, font=\scriptsize},
legend style={at={(0.05,0.5)}, anchor=south west, legend cell align=left, align=left, draw=white!15!black}
]
\addlegendimage{only marks, mark=*,  mark options={solid, mycolor1}}
\addlegendimage{only marks, mark=*,  mark options={solid, mycolor2}}
\addlegendimage{only marks, mark=*,  mark options={solid, mycolor3}}
\addplot [color=mycolor1, line width=0pt, draw=none, mark size=0.4 pt, mark=*, mark options={solid, mycolor1}]
  table[row sep=crcr]{%
1	0.0282848333333333\\
2	0.0294788571428571\\
3	0.0348190714285714\\
4	0.0356227619047619\\
5	0.0363190476190476\\
6	0.0641478666666667\\
7	0.0414624285714286\\
8	0.0354795714285714\\
9	0.0423198571428571\\
10	0.0442517476190476\\
11	0.0465084761904762\\
12	0.0431375238095238\\
13	0.0416933333333333\\
14	0.0440340952380952\\
15	0.042558580952381\\
16	0.0430638095238095\\
17	0.0423040952380952\\
18	0.0510496666666667\\
19	0.0434782619047619\\
20	0.0553071428571429\\
21	0.0437580428571429\\
22	0.0454522380952381\\
23	0.0438751428571429\\
24	0.0410570476190476\\
25	0.0656392095238095\\
26	0.0425976904761905\\
27	0.0476474285714286\\
28	0.0456056666666667\\
29	0.0408678571428571\\
30	0.0542727142857143\\
31	0.0501457\\
32	0.0481097380952381\\
33	0.0453261428571429\\
34	0.0461931428571429\\
35	0.101931380952381\\
36	0.052206\\
37	0.0514144285714286\\
38	0.0575487619047619\\
39	0.110188238095238\\
40	0.0471620952380952\\
41	0.0537368571428571\\
42	0.0543735714285714\\
43	0.0553264285714286\\
44	0.0492133285714286\\
45	0.0576759047619048\\
46	0.0551670714285714\\
47	0.055529\\
48	0.0503337619047619\\
49	0.0890453619047619\\
50	0.050956\\
51	0.0624224761904762\\
52	0.0539130952380952\\
53	0.0596613333333333\\
54	0.0551992142857143\\
55	0.0568724761904762\\
56	0.0535229523809524\\
57	0.0571807619047619\\
58	0.0619250952380952\\
59	0.0581142380952381\\
60	0.057318\\
61	0.055415619047619\\
62	0.0576416761904762\\
63	0.0562064761904762\\
64	0.0553872857142857\\
65	0.0843920714285714\\
66	0.0601905714285714\\
67	0.0604198571428572\\
68	0.058355619047619\\
69	0.0592459666666667\\
70	0.0588363333333333\\
71	0.0594928571428571\\
72	0.0700826666666667\\
73	0.0524638333333333\\
74	0.0558456666666666\\
75	0.0598558095238095\\
76	0.0617877619047619\\
77	0.0604020714285714\\
78	0.0646542380952381\\
79	0.0614498571428571\\
80	0.0709215714285714\\
81	0.116846761904762\\
82	0.0635780904761905\\
83	0.06837\\
84	0.0674512380952381\\
85	0.0882355714285714\\
86	0.0581782380952381\\
87	0.0585204761904762\\
88	0.0678593333333333\\
89	0.0623573333333333\\
90	0.0684971904761905\\
91	0.0918296666666667\\
92	0.0555555714285714\\
93	0.0791765714285714\\
94	0.0965145714285714\\
95	0.0628069428571429\\
96	0.0608143571428572\\
97	0.0655548571428571\\
98	0.0782311904761905\\
99	0.0693007142857143\\
100	0.0559944285714286\\
101	0.0671668095238095\\
102	0.0729920476190476\\
103	0.0584851904761905\\
104	0.10803519047619\\
105	0.0734618095238095\\
106	0.0850201142857143\\
107	0.0679018095238095\\
108	0.0600792857142857\\
109	0.0674687619047619\\
110	0.0736301428571429\\
111	0.0684403333333333\\
112	0.0713479523809524\\
113	0.0564833333333333\\
114	0.136873761904762\\
115	0.0753091428571429\\
116	0.0830064761904762\\
117	0.0779455714285714\\
118	0.0652427142857143\\
119	0.0720909047619048\\
120	0.126822142857143\\
121	0.0713465238095238\\
122	0.102341714285714\\
123	0.0725178095238095\\
124	0.0677007142857143\\
125	0.135048857142857\\
126	0.0847420285714286\\
127	0.0751507619047619\\
128	0.0804299523809524\\
129	0.0802655238095238\\
130	0.080325619047619\\
131	0.0786438571428572\\
132	0.0821194333333333\\
133	0.0784171428571429\\
134	0.0874220476190476\\
135	0.0982108571428572\\
136	0.0773989095238095\\
137	0.100221571428571\\
138	0.0763065666666667\\
139	0.0650585714285714\\
140	0.080353619047619\\
141	0.0804966190476191\\
142	0.0842282380952381\\
143	0.113916714285714\\
144	0.156501142857143\\
145	0.077298380952381\\
146	0.0837294761904762\\
147	0.0812937142857143\\
148	0.0798087619047619\\
149	0.0755554285714286\\
150	0.0827741904761905\\
151	0.0862850952380953\\
152	0.0913413238095238\\
153	0.0924272857142857\\
154	0.0824870952380952\\
155	0.0848018571428571\\
156	0.092137\\
157	0.0790394904761905\\
158	0.0871866666666667\\
159	0.0987431428571429\\
160	0.101993523809524\\
161	0.0889711904761905\\
162	0.0983625714285714\\
163	0.0729728571428571\\
164	0.0829175714285714\\
165	0.0847861904761905\\
166	0.0903180428571429\\
167	0.0814825238095238\\
168	0.0871852857142857\\
169	0.0913015238095238\\
170	0.0869541428571429\\
171	0.112002380952381\\
172	0.083263480952381\\
173	0.0885715238095238\\
174	0.0852399523809524\\
175	0.088469619047619\\
176	0.102788523809524\\
177	0.0927164285714286\\
178	0.0883074285714286\\
179	0.0890785714285714\\
180	0.0897031904761905\\
181	0.106001595238095\\
182	0.0903402857142857\\
183	0.107960333333333\\
184	0.0956506142857143\\
185	0.0958313809523809\\
186	0.093044\\
187	0.0825076190476191\\
188	0.0986516095238095\\
189	0.0905034761904762\\
190	0.107303952380952\\
191	0.0946968571428571\\
192	0.102249523809524\\
193	0.0894529047619048\\
194	0.126157914285714\\
195	0.101952571428571\\
196	0.114156257142857\\
197	0.0895721428571428\\
198	0.0994584285714286\\
199	0.0889247142857143\\
200	0.0965701428571429\\
201	0.0986901904761905\\
202	0.098035380952381\\
203	0.101704285714286\\
204	0.153499142857143\\
205	0.135765285714286\\
206	0.090793380952381\\
207	0.100505095238095\\
208	0.110726571428571\\
209	0.116866685714286\\
210	0.0989622380952381\\
211	0.0951077380952381\\
212	0.133107614285714\\
213	0.098035\\
214	0.0869394761904762\\
215	0.0986199523809524\\
216	0.101276142857143\\
217	0.122035976190476\\
218	0.0999184761904762\\
219	0.113989\\
220	0.101987476190476\\
221	0.100160952380952\\
222	0.102167133333333\\
223	0.104891857142857\\
224	0.110442095238095\\
225	0.10542480952381\\
226	0.111842142857143\\
227	0.125596471428571\\
228	0.129405095238095\\
229	0.0953608571428572\\
230	0.132903233333333\\
231	0.110417047619048\\
232	0.163821857142857\\
233	0.0965406666666667\\
234	0.127377\\
235	0.117443\\
236	0.107573738095238\\
237	0.138324166666667\\
238	0.120667419047619\\
239	0.102032\\
240	0.114976761904762\\
241	0.0984424761904762\\
242	0.106758666666667\\
243	0.118886142857143\\
244	0.0976714761904762\\
245	0.134493571428571\\
246	0.11797019047619\\
247	0.115666952380952\\
248	0.142377195238095\\
249	0.119499357142857\\
250	0.145527576190476\\
251	0.112667523809524\\
252	0.127031380952381\\
253	0.128633452380952\\
254	0.138546285714286\\
255	0.143312904761905\\
256	0.131102952380952\\
257	0.13161719047619\\
258	0.116997195238095\\
259	0.12473680952381\\
260	0.129028476190476\\
261	0.142083\\
262	0.131490333333333\\
263	0.129572242857143\\
264	0.121130904761905\\
265	0.107304180952381\\
266	0.145048780952381\\
267	0.122775771428571\\
268	0.137867647619048\\
269	0.116455519047619\\
270	0.111297238095238\\
271	0.12619880952381\\
272	0.173863333333333\\
273	0.143155476190476\\
274	0.134480857142857\\
275	0.140910257142857\\
276	0.155912714285714\\
277	0.137452142857143\\
278	0.139283714285714\\
279	0.139471833333333\\
280	0.148015952380952\\
281	0.142199380952381\\
282	0.148941428571429\\
283	0.13251979047619\\
284	0.142271095238095\\
285	0.135174923809524\\
286	0.146061852380952\\
287	0.129781857142857\\
288	0.135769095238095\\
289	0.150791\\
290	0.162387285714286\\
291	0.155066428571429\\
292	0.143415780952381\\
293	0.154535133333333\\
294	0.152714428571429\\
295	0.161646904761905\\
296	0.139359714285714\\
297	0.158821142857143\\
298	0.150864380952381\\
299	0.138750428571429\\
300	0.164916\\
301	0.151804523809524\\
302	0.178909\\
303	0.155425523809524\\
304	0.178409447619048\\
305	0.129637714285714\\
306	0.151714952380952\\
307	0.184142023809524\\
308	0.151299904761905\\
309	0.159825523809524\\
310	0.176924047619048\\
311	0.129272333333333\\
312	0.163286523809524\\
313	0.157361619047619\\
314	0.171970628571429\\
315	0.165780142857143\\
316	0.160448571428571\\
317	0.162708095238095\\
318	0.147135671428571\\
319	0.183787904761905\\
320	0.151657666666667\\
321	0.181420690476191\\
322	0.177635280952381\\
323	0.146230714285714\\
324	0.174690285714286\\
325	0.150759857142857\\
326	0.287831047619048\\
327	0.187090857142857\\
328	0.173634\\
329	0.175806047619048\\
330	0.214070619047619\\
331	0.139134857142857\\
332	0.188294476190476\\
333	0.161734333333333\\
334	0.182134761904762\\
335	0.191425619047619\\
336	0.161850619047619\\
337	0.189043142857143\\
338	0.197572142857143\\
339	0.161901547619048\\
340	0.198061\\
341	0.172313285714286\\
342	0.243124333333333\\
343	0.153031285714286\\
344	0.214564380952381\\
345	0.197711052380952\\
346	0.197518333333333\\
347	0.195810380952381\\
348	0.172210285714286\\
349	0.221996047619048\\
350	0.205385285714286\\
351	0.254013523809524\\
352	0.187815238095238\\
353	0.196164876190476\\
354	0.212427619047619\\
355	0.222729666666667\\
356	0.234823333333333\\
357	0.344295523809524\\
358	0.2217406\\
359	0.234070904761905\\
360	0.290742523809524\\
361	0.211081619047619\\
362	0.201279904761905\\
363	0.222417904761905\\
364	0.239457571428571\\
365	0.220126857142857\\
366	0.213971714285714\\
367	0.308602714285714\\
368	0.311117280952381\\
369	0.242201761904762\\
370	0.243392957142857\\
371	0.276737285714286\\
372	0.262311257142857\\
373	0.316273142857143\\
374	0.239055190476191\\
375	0.223734857142857\\
376	0.237742714285714\\
377	0.269547123809524\\
378	0.252082571428571\\
379	0.271491142857143\\
380	0.254549904761905\\
};
\addlegendentry{ Standard}

\addplot [color=mycolor2, line width=0pt, draw=none, mark size=0.4 pt, mark=*, mark options={solid, mycolor2}]
  table[row sep=crcr]{%
1	0.0428515\\
2	0.0421319523809524\\
3	0.103717976190476\\
4	0.0411044761904762\\
5	0.0561591619047619\\
6	0.0375804380952381\\
7	0.11223819047619\\
8	0.11114680952381\\
9	0.0770499523809524\\
10	0.0463028571428572\\
11	0.0614077619047619\\
12	0.100231285714286\\
13	0.0654046190476191\\
14	0.0491735238095238\\
15	0.0481856761904762\\
16	0.0678902857142857\\
17	0.0711544761904762\\
18	0.0817876190476191\\
19	0.0655653095238095\\
20	0.0838074285714286\\
21	0.0544209\\
22	0.066770619047619\\
23	0.0559230476190476\\
24	0.109208619047619\\
25	0.0401497333333333\\
26	0.0917513571428571\\
27	0.0920822857142857\\
28	0.0501031428571429\\
29	0.0471211428571428\\
30	0.0570805238095238\\
31	0.0679364761904762\\
32	0.0745096904761905\\
33	0.117336047619048\\
34	0.0774517142857143\\
35	0.142053666666667\\
36	0.0974631857142857\\
37	0.0576962380952381\\
38	0.0491996666666667\\
39	0.0825142380952381\\
40	0.0514367142857143\\
41	0.0559375714285714\\
42	0.0547382380952381\\
43	0.121437380952381\\
44	0.0816800476190476\\
45	0.122023238095238\\
46	0.113137214285714\\
47	0.133249666666667\\
48	0.065392619047619\\
49	0.0810460952380953\\
50	0.0660815714285714\\
51	0.0528888095238095\\
52	0.074865380952381\\
53	0.0761010952380952\\
54	0.0609555476190476\\
55	0.069764619047619\\
56	0.0674169523809524\\
57	0.0740664285714286\\
58	0.0737797619047619\\
59	0.0963579047619048\\
60	0.0678302857142857\\
61	0.0558536666666667\\
62	0.0852272952380953\\
63	0.087870380952381\\
64	0.0960160476190476\\
65	0.0885864523809524\\
66	0.0551631904761905\\
67	0.0645667142857143\\
68	0.0741520476190476\\
69	0.100888561904762\\
70	0.0814693809523809\\
71	0.0610399047619048\\
72	0.0883113809523809\\
73	0.0841433571428572\\
74	0.0839578095238095\\
75	0.141754428571429\\
76	0.0868476666666667\\
77	0.136782357142857\\
78	0.085707\\
79	0.0563276666666667\\
80	0.0821199047619048\\
81	0.0980181428571428\\
82	0.0511747571428571\\
83	0.0657547142857143\\
84	0.0762264285714286\\
85	0.08353\\
86	0.0679396666666667\\
87	0.0759495714285714\\
88	0.0901322380952381\\
89	0.0792161904761905\\
90	0.0700370952380952\\
91	0.0829645238095238\\
92	0.092715619047619\\
93	0.12400680952381\\
94	0.0780224285714286\\
95	0.0684346190476191\\
96	0.0720187857142857\\
97	0.121975904761905\\
98	0.0689157142857143\\
99	0.0903503809523809\\
100	0.0729500476190476\\
101	0.0595786666666667\\
102	0.0900388571428571\\
103	0.0896470952380952\\
104	0.0770984285714286\\
105	0.105797142857143\\
106	0.119111261904762\\
107	0.141748\\
108	0.105541328571429\\
109	0.101258238095238\\
110	0.0988505714285714\\
111	0.0944720476190476\\
112	0.129765619047619\\
113	0.117359285714286\\
114	0.144963333333333\\
115	0.115103142857143\\
116	0.0973081904761905\\
117	0.142566380952381\\
118	0.079937619047619\\
119	0.119894619047619\\
120	0.115770233333333\\
121	0.103856142857143\\
122	0.140299333333333\\
123	0.0824644761904762\\
124	0.107412571428571\\
125	0.0829241904761905\\
126	0.200276904761905\\
127	0.0791630476190476\\
128	0.100583380952381\\
129	0.0529185238095238\\
130	0.164159857142857\\
131	0.111881138095238\\
132	0.10651029047619\\
133	0.120131285714286\\
134	0.0724782761904762\\
135	0.0834621428571429\\
136	0.0620515761904762\\
137	0.166285666666667\\
138	0.0798688476190476\\
139	0.128160904761905\\
140	0.097902619047619\\
141	0.0828815714285714\\
142	0.105187666666667\\
143	0.105645238095238\\
144	0.123744523809524\\
145	0.0964015714285714\\
146	0.103952238095238\\
147	0.111368904761905\\
148	0.12357380952381\\
149	0.150575476190476\\
150	0.116911428571429\\
151	0.149826142857143\\
152	0.12890809047619\\
153	0.121219380952381\\
154	0.115374476190476\\
155	0.082838380952381\\
156	0.131569904761905\\
157	0.0940675714285714\\
158	0.0820525714285714\\
159	0.117080619047619\\
160	0.118665238095238\\
161	0.118782471428571\\
162	0.190416238095238\\
163	0.119607571428571\\
164	0.115632142857143\\
165	0.125163952380952\\
166	0.0925149952380953\\
167	0.0958997142857143\\
168	0.0986317142857143\\
169	0.07955\\
170	0.165834523809524\\
171	0.111976380952381\\
172	0.0772256714285714\\
173	0.123975857142857\\
174	0.0977297619047619\\
175	0.148268523809524\\
176	0.11751200952381\\
177	0.120340852380952\\
178	0.0996000476190476\\
179	0.109581476190476\\
180	0.0887097619047619\\
181	0.0946189285714286\\
182	0.0950666666666667\\
183	0.102875857142857\\
184	0.0950438523809524\\
185	0.0942379523809524\\
186	0.12058380952381\\
187	0.0836286666666667\\
188	0.089528180952381\\
189	0.0900960952380952\\
190	0.164383142857143\\
191	0.131651333333333\\
192	0.121699952380952\\
193	0.163026619047619\\
194	0.0931419619047619\\
195	0.0969795238095238\\
196	0.138343447619048\\
197	0.152150904761905\\
198	0.0904594761904762\\
199	0.177811238095238\\
200	0.13966119047619\\
201	0.112468952380952\\
202	0.176204142857143\\
203	0.125722571428571\\
204	0.116553047619048\\
205	0.0984389047619048\\
206	0.107747476190476\\
207	0.102582714285714\\
208	0.131207714285714\\
209	0.128124161904762\\
210	0.109590476190476\\
211	0.104625976190476\\
212	0.118448519047619\\
213	0.157691952380952\\
214	0.124438976190476\\
215	0.157122761904762\\
216	0.105121714285714\\
217	0.165400833333333\\
218	0.110927238095238\\
219	0.084758519047619\\
220	0.109130380952381\\
221	0.0981132380952381\\
222	0.15089760952381\\
223	0.124093714285714\\
224	0.115213523809524\\
225	0.100454819047619\\
226	0.11312919047619\\
227	0.10172499047619\\
228	0.137762333333333\\
229	0.108283380952381\\
230	0.0895459952380952\\
231	0.159538857142857\\
232	0.106838333333333\\
233	0.103841857142857\\
234	0.172837142857143\\
235	0.176426142857143\\
236	0.156528809523809\\
237	0.0986206428571429\\
238	0.117939233333333\\
239	0.128253380952381\\
240	0.0871467142857143\\
241	0.154952285714286\\
242	0.14799480952381\\
243	0.100801142857143\\
244	0.0987234761904762\\
245	0.105141428571429\\
246	0.124151047619048\\
247	0.0929572380952381\\
248	0.115501052380952\\
249	0.106133352380952\\
250	0.113984242857143\\
251	0.160408476190476\\
252	0.0945293333333333\\
253	0.120259214285714\\
254	0.133967571428571\\
255	0.12256980952381\\
256	0.13129090952381\\
257	0.132770761904762\\
258	0.122434623809524\\
259	0.151243333333333\\
260	0.116314238095238\\
261	0.159916333333333\\
262	0.143203428571429\\
263	0.1139611\\
264	0.122401142857143\\
265	0.0959212809523809\\
266	0.174273619047619\\
267	0.1386172\\
268	0.169966704761905\\
269	0.166158180952381\\
270	0.177905523809524\\
271	0.149970857142857\\
272	0.168981523809524\\
273	0.145105190476191\\
274	0.137864904761905\\
275	0.130436523809524\\
276	0.162828619047619\\
277	0.11393919047619\\
278	0.137119333333333\\
279	0.124957261904762\\
280	0.169734333333333\\
281	0.131221380952381\\
282	0.111019142857143\\
283	0.132011380952381\\
284	0.144199380952381\\
285	0.1584704\\
286	0.156595090476191\\
287	0.147566095238095\\
288	0.169892904761905\\
289	0.165065380952381\\
290	0.14416879047619\\
291	0.12426980952381\\
292	0.189688447619048\\
293	0.139576419047619\\
294	0.168635666666667\\
295	0.135636380952381\\
296	0.113803238095238\\
297	0.164300757142857\\
298	0.160379333333333\\
299	0.176525047619048\\
300	0.153788142857143\\
301	0.150310095238095\\
302	0.173442714285714\\
303	0.157651\\
304	0.109580476190476\\
305	0.146104619047619\\
306	0.157095714285714\\
307	0.157664642857143\\
308	0.143237723809524\\
309	0.141218714285714\\
310	0.180167380952381\\
311	0.146167857142857\\
312	0.157298047619048\\
313	0.157794523809524\\
314	0.134916438095238\\
315	0.134995095238095\\
316	0.165811376190476\\
317	0.142562285714286\\
318	0.169016523809524\\
319	0.147349476190476\\
320	0.158258333333333\\
321	0.172189976190476\\
322	0.125862185714286\\
323	0.169892952380952\\
324	0.162219285714286\\
325	0.123222428571429\\
326	0.202862380952381\\
327	0.194559904761905\\
328	0.117703238095238\\
329	0.155433380952381\\
330	0.171762571428571\\
331	0.175589428571429\\
332	0.166743142857143\\
333	0.165115666666667\\
334	0.237394285714286\\
335	0.166854\\
336	0.172090238095238\\
337	0.173958047619048\\
338	0.152883380952381\\
339	0.144179452380952\\
340	0.177765523809524\\
341	0.171018857142857\\
342	0.221742666666667\\
343	0.160741380952381\\
344	0.175646666666667\\
345	0.151505623809524\\
346	0.195739666666667\\
347	0.19733619047619\\
348	0.205340714285714\\
349	0.157893904761905\\
350	0.198651523809524\\
351	0.18205219047619\\
352	0.169853\\
353	0.19646620952381\\
354	0.168830904761905\\
355	0.199961285714286\\
356	0.211307523809524\\
357	0.177591\\
358	0.167565523809524\\
359	0.218104333333333\\
360	0.208510333333333\\
361	0.165296904761905\\
362	0.181897428571429\\
363	0.187952657142857\\
364	0.237492952380952\\
365	0.211400857142857\\
366	0.212756476190476\\
367	0.178799142857143\\
368	0.240096614285714\\
369	0.197069571428571\\
370	0.185242766666667\\
371	0.210261142857143\\
372	0.245955352380952\\
373	0.207146571428571\\
374	0.225110333333333\\
375	0.200609185714286\\
376	0.203696333333333\\
377	0.204681142857143\\
378	0.229804523809524\\
379	0.252432523809524\\
380	0.249708523809524\\
};
\addlegendentry{Slanted}

\addplot [color=mycolor3, line width=1.5pt]
  table[row sep=crcr]{%
1	0.0279941666666667\\
2	0.0289090952380952\\
3	0.0351133095238095\\
4	0.0351328571428571\\
5	0.0376296047619048\\
6	0.0387476761904762\\
7	0.0393878095238095\\
8	0.039409619047619\\
9	0.0399582857142857\\
10	0.0403066523809524\\
11	0.0409088571428571\\
12	0.0411362380952381\\
13	0.0412086666666667\\
14	0.0419072857142857\\
15	0.0419210571428571\\
16	0.041936\\
17	0.0424050952380952\\
18	0.042583619047619\\
19	0.0428355952380952\\
20	0.0429088571428571\\
21	0.0430975666666667\\
22	0.0434587619047619\\
23	0.0435369523809524\\
24	0.0435421666666667\\
25	0.0438554\\
26	0.0439706904761905\\
27	0.0440723333333333\\
28	0.0442427142857143\\
29	0.044742\\
30	0.045321619047619\\
31	0.0462154285714286\\
32	0.0462531666666667\\
33	0.0463539047619048\\
34	0.0467382380952381\\
35	0.0473219523809524\\
36	0.047796380952381\\
37	0.0483089523809524\\
38	0.0495484761904762\\
39	0.0495799047619048\\
40	0.0498656476190476\\
41	0.0501861904761905\\
42	0.050463380952381\\
43	0.0505327619047619\\
44	0.0510161904761905\\
45	0.0518432380952381\\
46	0.0518956428571428\\
47	0.0522685238095238\\
48	0.0523567619047619\\
49	0.0528359380952381\\
50	0.0529580476190476\\
51	0.0531613333333333\\
52	0.0535586666666667\\
53	0.0540421428571429\\
54	0.0547312619047619\\
55	0.0548135714285714\\
56	0.054921\\
57	0.0550624238095238\\
58	0.0556424285714286\\
59	0.0558948571428571\\
60	0.0560092857142857\\
61	0.0560519523809524\\
62	0.0562993428571429\\
63	0.0563555714285714\\
64	0.056393619047619\\
65	0.0564187333333333\\
66	0.0566829380952381\\
67	0.0567610952380952\\
68	0.05713\\
69	0.0571311619047619\\
70	0.0573851428571429\\
71	0.0573996666666667\\
72	0.0577669523809524\\
73	0.0587019285714286\\
74	0.0590519047619048\\
75	0.0591591428571428\\
76	0.0594174761904762\\
77	0.0594506904761905\\
78	0.0594655238095238\\
79	0.0596736666666666\\
80	0.0596890952380952\\
81	0.0599349523809524\\
82	0.0602025190476191\\
83	0.0605661428571429\\
84	0.0608417142857143\\
85	0.0611474761904762\\
86	0.0611637142857143\\
87	0.0613678095238095\\
88	0.0615454285714286\\
89	0.0619050476190476\\
90	0.0620964285714286\\
91	0.0627949523809524\\
92	0.0632066666666667\\
93	0.0632622380952381\\
94	0.0633141904761905\\
95	0.0636399476190476\\
96	0.0639506904761905\\
97	0.0645675714285714\\
98	0.0648239047619047\\
99	0.0650213333333333\\
100	0.0650627619047619\\
101	0.0658040476190476\\
102	0.0659577142857143\\
103	0.0664318571428572\\
104	0.0667257619047619\\
105	0.0669751904761905\\
106	0.0675610190476191\\
107	0.0679965238095238\\
108	0.0681605714285714\\
109	0.0684187619047619\\
110	0.0686159047619048\\
111	0.0686176666666667\\
112	0.0688430476190476\\
113	0.0688536190476191\\
114	0.0693620952380952\\
115	0.0694293333333333\\
116	0.0696187142857143\\
117	0.0698921428571429\\
118	0.0704229523809524\\
119	0.0706589047619048\\
120	0.0707089\\
121	0.0710075238095238\\
122	0.0715787619047619\\
123	0.0716710476190476\\
124	0.0720238571428571\\
125	0.0726579047619048\\
126	0.0729555523809524\\
127	0.0730394761904762\\
128	0.0733592380952381\\
129	0.0735663333333333\\
130	0.073651619047619\\
131	0.0740138571428572\\
132	0.0743439761904762\\
133	0.0744509047619048\\
134	0.0747227619047619\\
135	0.075258380952381\\
136	0.0755448619047619\\
137	0.0757296666666667\\
138	0.0763636142857143\\
139	0.0766294761904762\\
140	0.0767493952380953\\
141	0.0767734761904762\\
142	0.0779135714285714\\
143	0.0781732857142857\\
144	0.0784268095238095\\
145	0.0784704285714286\\
146	0.0786566190476191\\
147	0.0788777619047619\\
148	0.0789825238095238\\
149	0.0790377619047619\\
150	0.0792060952380952\\
151	0.0793991428571429\\
152	0.0799166571428571\\
153	0.0802257619047619\\
154	0.080357380952381\\
155	0.0806291904761905\\
156	0.081238\\
157	0.0813798238095238\\
158	0.0815027333333333\\
159	0.0819143809523809\\
160	0.082088619047619\\
161	0.0826740476190476\\
162	0.082815\\
163	0.0829025714285714\\
164	0.0829990476190476\\
165	0.0831294761904762\\
166	0.0835141809523809\\
167	0.0836491904761905\\
168	0.0838474285714285\\
169	0.0841548095238095\\
170	0.0843931904761905\\
171	0.0847924761904762\\
172	0.0851436238095238\\
173	0.0856957142857143\\
174	0.0858835238095238\\
175	0.0859000952380952\\
176	0.0863728095238095\\
177	0.0872465714285714\\
178	0.0878109523809524\\
179	0.0884018095238095\\
180	0.0884433333333333\\
181	0.088645119047619\\
182	0.0895728571428571\\
183	0.0897276666666667\\
184	0.0898186619047619\\
185	0.0899461904761905\\
186	0.0899961428571429\\
187	0.0903373333333333\\
188	0.0904690380952381\\
189	0.0905353333333333\\
190	0.0918090952380952\\
191	0.0918350476190476\\
192	0.0922626190476191\\
193	0.0937176666666667\\
194	0.0938961047619048\\
195	0.0943198095238095\\
196	0.0943716714285714\\
197	0.0944836666666667\\
198	0.0946101904761905\\
199	0.0946283333333333\\
200	0.0947495238095238\\
201	0.0951687619047619\\
202	0.0955973333333333\\
203	0.0958813952380952\\
204	0.0962778571428571\\
205	0.0962811904761905\\
206	0.0968245714285714\\
207	0.0969862857142857\\
208	0.0972036190476191\\
209	0.0975012095238095\\
210	0.0975498095238095\\
211	0.0988180714285714\\
212	0.0989536619047619\\
213	0.0990709047619048\\
214	0.0990756666666667\\
215	0.0997007619047619\\
216	0.10044019047619\\
217	0.100457119047619\\
218	0.100740476190476\\
219	0.101393857142857\\
220	0.101465904761905\\
221	0.102015714285714\\
222	0.102221514285714\\
223	0.103425857142857\\
224	0.105073142857143\\
225	0.106090238095238\\
226	0.106659380952381\\
227	0.1077469\\
228	0.107880761904762\\
229	0.108049761904762\\
230	0.108369614285714\\
231	0.108387142857143\\
232	0.108426095238095\\
233	0.109026095238095\\
234	0.109500333333333\\
235	0.10972780952381\\
236	0.110954971428571\\
237	0.111337214285714\\
238	0.111418847619048\\
239	0.112943761904762\\
240	0.113147095238095\\
241	0.113147142857143\\
242	0.113947857142857\\
243	0.114422523809524\\
244	0.114646523809524\\
245	0.114664238095238\\
246	0.114953428571429\\
247	0.115328619047619\\
248	0.115543195238095\\
249	0.11609330952381\\
250	0.116881528571429\\
251	0.118645476190476\\
252	0.119950285714286\\
253	0.120680923809524\\
254	0.12135419047619\\
255	0.122107761904762\\
256	0.122139857142857\\
257	0.122214476190476\\
258	0.122768480952381\\
259	0.123919047619048\\
260	0.124816428571429\\
261	0.125257904761905\\
262	0.125375571428571\\
263	0.125907566666667\\
264	0.126115857142857\\
265	0.126343276190476\\
266	0.127335828571429\\
267	0.129228485714286\\
268	0.129716071428571\\
269	0.130172376190476\\
270	0.131180333333333\\
271	0.131429285714286\\
272	0.132127571428571\\
273	0.133607714285714\\
274	0.134978380952381\\
275	0.135978495238095\\
276	0.136740523809524\\
277	0.13763780952381\\
278	0.137942761904762\\
279	0.138658976190476\\
280	0.140325904761905\\
281	0.140390666666667\\
282	0.140960380952381\\
283	0.140985476190476\\
284	0.142677666666667\\
285	0.142862161904762\\
286	0.142896276190476\\
287	0.143510714285714\\
288	0.144405571428571\\
289	0.144582476190476\\
290	0.144583095238095\\
291	0.145052666666667\\
292	0.145075971428571\\
293	0.145118466666667\\
294	0.145211380952381\\
295	0.145493857142857\\
296	0.147270619047619\\
297	0.149241047619048\\
298	0.150008047619048\\
299	0.150487285714286\\
300	0.152210952380952\\
301	0.152300952380952\\
302	0.153492047619048\\
303	0.153915714285714\\
304	0.154602257142857\\
305	0.154995666666667\\
306	0.155095666666667\\
307	0.156636261904762\\
308	0.156830238095238\\
309	0.157227571428571\\
310	0.15745719047619\\
311	0.157553047619048\\
312	0.158574666666667\\
313	0.158617428571429\\
314	0.158986438095238\\
315	0.159605619047619\\
316	0.161523285714286\\
317	0.162072666666667\\
318	0.163024476190476\\
319	0.163249\\
320	0.164690095238095\\
321	0.165049161904762\\
322	0.167360847619048\\
323	0.168442523809524\\
324	0.171240714285714\\
325	0.17180080952381\\
326	0.173182428571429\\
327	0.174738714285714\\
328	0.175719714285714\\
329	0.176210047619048\\
330	0.176322904761905\\
331	0.176457761904762\\
332	0.176600619047619\\
333	0.178084952380952\\
334	0.180299504761905\\
335	0.18085880952381\\
336	0.180985619047619\\
337	0.182570904761905\\
338	0.182621666666667\\
339	0.183434642857143\\
340	0.186020952380952\\
341	0.186404761904762\\
342	0.187841285714286\\
343	0.189235571428571\\
344	0.191200571428571\\
345	0.19194029047619\\
346	0.193985142857143\\
347	0.195943904761905\\
348	0.198614619047619\\
349	0.199598142857143\\
350	0.201182904761905\\
351	0.205015\\
352	0.206928428571429\\
353	0.207162161904762\\
354	0.215250095238095\\
355	0.216881571428571\\
356	0.219228285714286\\
357	0.221100761904762\\
358	0.223168333333333\\
359	0.223460476190476\\
360	0.224015904761905\\
361	0.22506719047619\\
362	0.225416380952381\\
363	0.225982761904762\\
364	0.226506619047619\\
365	0.229558380952381\\
366	0.22967380952381\\
367	0.237553904761905\\
368	0.239790471428571\\
369	0.243528\\
370	0.249247576190476\\
371	0.250989314285714\\
372	0.254152876190476\\
373	0.255708190476191\\
374	0.257122333333333\\
375	0.257185571428571\\
376	0.257793857142857\\
377	0.26545219047619\\
378	0.27039780952381\\
379	0.278804238095238\\
380	0.284303142857143\\
};
\addlegendentry{Ours}

\end{axis}
\end{tikzpicture}%

%% file: sec_limitation_futureWork_conclusion.tex
\section{Conclusion, Limitations and Future Work}
In this paper, we proposed a new regularizer, the resolvent Laplacian commutativity, for the functional map pipeline.  We first analyzed the limitations of the original Laplacian commutativity term and theoretically justified the effectiveness of our proposed new term. This new regularizer can significantly improve the quality of the computed functional maps and the corresponding recovered point-wise maps before or after refinement.

However, our method also has several limitations that we would like to overcome in future work.  First, our proposed regularizer is well justified on isometric and non-isometric shape pairs, but not on partial shapes, where the ground-truth functional maps can have a different structure. Therefore, it would be interesting to extend the analysis to partial shape pairs.
Second, besides the funnel pattern, in our experiments, we also observed the slanted-diagonal structure of the ground-truth functional map of some non-isometric shape pairs as discussed in~\cite{rodola2017partial}. It would be interesting to consider this feature into the mask constructions for further improvement for non-isometric datasets. 
\hl{
Thirdly, the role of the (a, b) introduced in Eq.~(5) is not well studied, and we would like to leave this exploration as future work. 
}
Finally, we believe that it would be interesting to study and potentially apply data-driven techniques to 
learn the optimal operators for enforcing commutativity across diverse shape pairs.

%% file: sec_appendix_HilbertSchmidt.tex
\section{Pullbacks are Bounded}\label{appendixPullback}

\begin{lemma}
Let $\Phi:\mathcal{M} \to \mathcal{N}$ be a diffeomorphism between connected compact oriented Riemannian manifolds. Then, the associated pullback $\Phi^{*}$ is bounded as a map $L_{2}(\mathcal{N}) \to L_{2}(\mathcal{M})$.
\end{lemma}
\begin{proof}
Let $d\mathcal{M}$ and $d\mathcal{N}$ denote the volume forms of $\mathcal{M}$ and $\mathcal{N}$, respectively. Since $C^{\infty}$ is dense in $L_{2}$, it is enough to show that there exists a constant $B > 0$ such that, for any $f \in C^{\infty}(\mathcal{N})$,

\begin{equation}
    \int_{\mathcal{M}} \left( \Phi^{*} f \right)^{2} d \mathcal{M} \leq B \int_{\mathcal{N}} f^{2} d \mathcal{N}~.
\end{equation}

\noindent By virtue of being a diffeomorphism, $\Phi$ is invertible. Using the pullback of $\Phi^{-1}$, we can express the left-hand side of the desired inequality as an integral over $\mathcal{N}$.

Begin by considering $\left[ \left(\Phi^{-1} \right)^{*}d \mathcal{M} \right]$, the pullback of the volume form of $\mathcal{M}$. Since volume forms are top degree forms, there exists $u \in C^{\infty}(\mathcal{N})$ such that

\begin{equation}
    \left[ \left(\Phi^{-1} \right)^{*}d \mathcal{M} \right] = u d \mathcal{N}
\end{equation}

\noindent $\Phi$ is either orientation preserving or orientation reversing. In the former case, $u > 0$. In the latter, $u < 0$. In either case, the left-hand side of the desired inequality can be recast as

\begin{equation}
\begin{split}
    \int_{\mathcal{M}} \left( \Phi^{*} f \right)^{2} d \mathcal{M} &= \int_{\mathcal{N}} \left( \left(\Phi^{-1} \right)^{*} \Phi^{*} f \right)^{2} |u| d\mathcal{N}\\
    &= \int_{\mathcal{N}}  f^{2} |u| d\mathcal{N}~.
\end{split}
\end{equation}

\noindent It remains to bound this expression.

\begin{equation}
        \int_{\mathcal{N}}  f^{2} |u| d\mathcal{N} \leq \sup_{x \in \mathcal{N}} |u(x)| \int_{\mathcal{N}}  f^{2} d\mathcal{N}
\end{equation}

\noindent Since $\mathcal{N}$ is compact and $u$ is continuous, the supremum is achieved and is finite. This concludes the proof.
\end{proof}

\section{Bounded Frobenius Norm for $\gamma > 1/2$} \label{appendixHilbertSchmidt}

\noindent In this appendix, we prove a sufficient condition for the Frobenius norm based energy to be defined in the continuous case. We begin by introducing an infinite dimensional analog of the Frobenius norm, which is provided by the Hilbert-Schmidt norm.

\begin{definition}[Hilbert-Schmidt Norm]
Let $A: \mathcal{H}_1 \to \mathcal{H}_2$ be a linear operator between Hilbert spaces. Let $A^{\dagger}$ denote the adjoint of the operator $A$. Then, the Hilbert-Schmidt norm of $A$ is given by:
\begin{equation}
    \| A \|^{2}_{HS} = \text{Tr}\left( A^{\dagger} A \right) = \sum_{i=1}^{\infty} \langle e_{i} , A^{\dagger} A e_{i}  \rangle~,
\end{equation}

\noindent where $\{e_{i}\}_{i=1}^{\infty}$ is any orthonormal basis of $\mathcal{H}_1$. This norm is also known as the Schatten $2$-norm.
\end{definition}

\noindent The following lemma is the main result of this appendix.

\begin{lemma} \label{LemmaBoundedEnergy}
Let $\Delta_1$ and $\Delta_2$ be Laplacians on compact, connected, oriented surfaces $\mathcal{M}_1$ and $\mathcal{M}_2$, respectively. Let $\matr{C}_{12}: L_2(\mathcal{M}_1) \to L_2(\mathcal{M}_2)$ be a bounded operator. If $\gamma > 1/2$, then:
\begin{equation}
    \left\| \matr{C}_{12} R_{\mu} \left( \Delta_1^{\gamma} \right) - R_{\mu} \left( \Delta_2^{\gamma} \right) \matr{C}_{12}  \right\|^{2}_{HS} < \infty~,
\end{equation}

\noindent where $\mu$ is any complex number not on the non-negative real line.
\end{lemma}

\begin{proof}
Operators with finite Hilbert-Schmidt norm are known as operators of Hilbert-Schmidt class. It can be shown (see \cite{reedsimon1}) that linear combinations of Hilbert-Schmidt class operators are of Hilbert-Schmidt class. Moreover, the product of a bounded operator and a Hilbert-Schmidt class operator is also of Hilbert-Schmidt class. Thus, it is sufficient to show that $R_{\mu}(\Delta^{\gamma})$ has finite Hilbert-Schmidt norm for $\gamma > 1/2$.

Denote the eigenfunctions and eigenvalues of $\Delta$ by $\{ \psi_{k} \}_{k=0}^\infty$ and $\{ \lambda_{k} \}_{k=0}^\infty$, respectively. We assume that the eigenvalues are numbered in the usual non-decreasing order.

$R_{\mu}(\Delta^{\gamma})$ is diagonal in the eigenbasis of $\Delta$ and has eigenvalues $1/(\lambda_{k}^{\gamma} - \mu )$. Thus, the Hilbert-Schmidt norm of $R_{\mu}(\Delta^{\gamma})$ is given by:
\begin{equation}
\begin{split}
    \| R_{\mu}(\Delta^{\gamma}) \|^{2}_{HS} &= \sum_{k = 0}^{\infty} \langle \psi_{k} , R_{\mu}(\Delta^{\gamma})^{\dagger} R_{\mu}(\Delta^{\gamma}) \psi_{k} \rangle\\
    &= \sum_{k = 0}^{\infty} \frac{1}{\left| \lambda_{k}^{\gamma} - \mu \right|^{2}}
\end{split}
\end{equation}

We will establish the convergence of this series for $\gamma > 1/2$ using the comparison test with the series $\sum_{k=1}^{\infty} \frac{1}{k^{p}}$.  This series is well-known to converge if and only if $p>1$.

Since the $\lambda_{k}$ are non-negative and increasing with $k$ towards $\infty$, the following inequality holds for all large enough $k$:

\begin{equation}
    \frac{1}{4 \lambda_{k}^{2\gamma}} \leq \frac{1}{\left| \lambda_{k}^{\gamma} - \mu \right|^{2}} \leq \frac{2}{ \lambda_{k}^{2 \gamma}}~.
\end{equation}

\noindent By Weyl's law, there exists a constant $B > 0$ such that $\lambda_{k} \sim B k$ for large $k$ \cite{minakshisundaram1949some}. Then, the inequality becomes

\begin{equation}
  \frac{1}{ 4 B^{2 \gamma}} \frac{1}{ k^{2 \gamma} } \leq  \frac{1}{\left| \lambda_{k}^{\gamma} - \mu \right|^{2}} \leq   \frac{2}{ B^{2 \gamma} } \frac{1}{ k^{2 \gamma}}~.
\end{equation}

\noindent Thus, by the comparison test, the series for $\| R_{\mu}(\Delta^{\gamma}) \|^{2}_{HS}$ converges if and only if $\gamma > 1/2$. This concludes the proof.
\end{proof}

We conclude this appendix with two remarks on the above result. First, note that we have shown a result slightly stronger than required by the statement of the lemma. In fact, we only needed to show that the series for $\| R_{\mu}(\Delta^{\gamma}) \|^{2}_{HS}$  converges if $\gamma > 1/2$. The "only if" part was optional. We have done this to illustrate that the above proof \emph{strategy} is guaranteed to fail for $\gamma \leq 1/2$. Specifically, it is no longer sufficient to assume that $C_{12}$ is merely a bounded operator. One also cannot simply assume $C_{12}$ to be of Hilbert-Schmidt class, as this rules out the important case of $C_{12} = Id$, which is bounded, but not Hilbert-Schmidt.

As a final remark, note that an analogous proof strategy can be applied to the Schatten $p$-norm, which can be seen as the $l_{p}$ generalization of the Hilbert-Schmidt norm. There, the $k^{\text{th}}$ term of the series would be $1/|\lambda_{k}^{\gamma} - \mu |^{p}$ and convergence would be guaranteed for $\gamma > 1/p$. Thus, one can find $p$ large enough so that $R_{\mu}(\Delta^{\gamma})$ has a well-defined Schatten $p$-norm for any given $\gamma > 0$.

%% file: sec_appendix_heatMask.tex
\section{Comparison to heat mask}\label{appendix:heatMask}

\begin{figure}[!ht]
\centering
\includegraphics[width = 0.98\columnwidth]{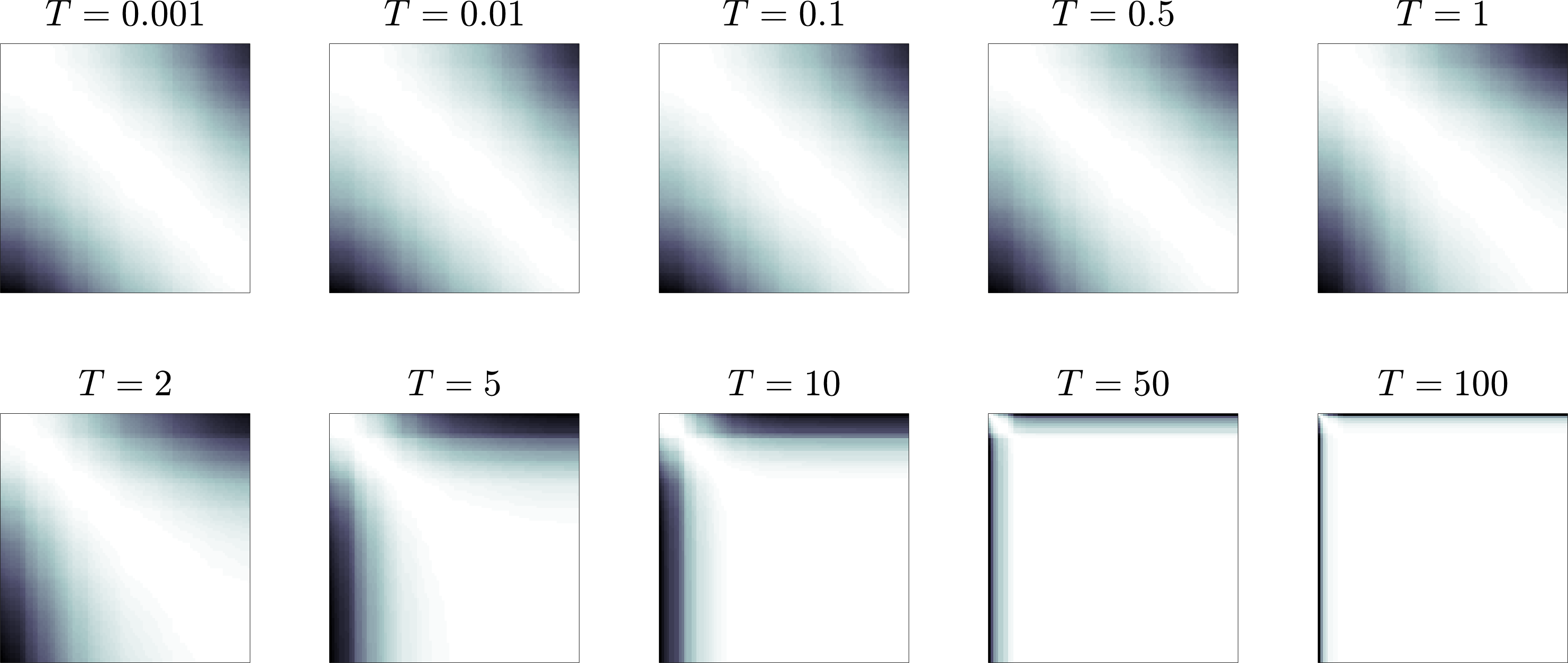}
\caption{Visualization of the heat mask with different time-scale parameter $T$.}
\label{fig:res:heatMask}
\end{figure}

\begin{figure}[!ht]
\centering
\input{figures_suppl/diff_heat_param.tex}
\caption{Changing the time-scale parameter $T$ of the heat mask. Here we test the performance of the heat mask (yellow curve) with different $T$ and report the average direct error of~50 random FAUST isometric shape pairs. The average error of the standard mask (blue dashed line) and the average error of our resolvent mask (red dashed line) are included for comparison.}
\label{fig:res:heatMask:curve}
\vspace{-12pt}
\end{figure}
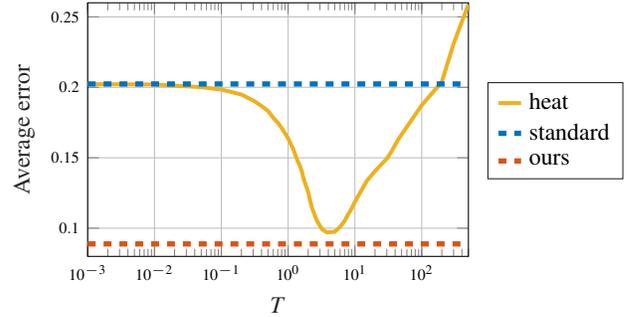

Besides our resolvent-based commutativity term, another natural bounded option would be to use the commutativity with the \emph{heat} operators, which are also bounded linear functional operators, which can act as regularizers on functional maps. This would lead to the following mask matrix:
\begin{equation}
\matr{M}_{heat}(i,j) = \exp(-T \Lambda_2(i)) - \exp(-T \Lambda_1(j)),
\end{equation}
where $T$ is the scalar time parameter. Fig.~\ref{fig:res:heatMask} visualizes the heat mask with different values of $T$ on an isometric shape pair. In Fig.~\ref{fig:res:heatMask:curve} we compared the performance of the heat mask with different values of $T$ ranging from $10^{-3}$ to 500 on FAUST isometric shape pairs.
We can see that, as a bounded operator, the heat mask has a better performance than the standard mask. 
When $T = 5$, it gives the best performance among the tested values. The heat mask with $T = 5$ (in Fig.~\ref{fig:res:heatMask}) has a similar funnel structure, but still does not achieve the quality of the results we obtain with the mask based on the resolvent operators.

%% file: figures_suppl/diff_heat_param.tex
%
%
\definecolor{mycolor1}{rgb}{0.92941,0.69412,0.12549}%
\definecolor{mycolor2}{rgb}{0.00000,0.44700,0.74100}%
\definecolor{mycolor3}{rgb}{0.85098,0.32549,0.09804}%
\pgfplotsset{scaled x ticks=false}
\pgfplotsset{
compat=1.11,
legend image code/.code={
\draw[mark repeat=2,mark phase=2]
plot coordinates {
(0cm,0cm)
(0.0cm,0cm)        
(0.3cm,0cm)         
};%
}
}
\begin{tikzpicture}
\begin{axis}[%
width=0.6\linewidth,
height=0.4\linewidth,
at={(6.441in,0.894in)},
scale only axis,
xmode=log,
xmin=0.001,
xmax=500,
xminorticks=true,
xlabel style={font=\color{white!15!black}},
xlabel={$T$},
ymin=0.08,
ymax=0.26,
ylabel style={font=\color{white!15!black}},
ylabel={Average error},
axis background/.style={fill=white},
xmajorgrids,
ymajorgrids,
scaled x ticks=false,
scaled y ticks=false,
every x tick label/.append style={font=\color{black}, font=\scriptsize},
every y tick label/.append style={font=\color{black}, font=\scriptsize},
legend style={at={(1.05,0.304)}, anchor=south west, legend cell align=left, align=left, draw=white!15!black}
]
\addplot [color=mycolor1, line width=1.5pt]
  table[row sep=crcr]{%
0.001	0.20225\\
0.01	0.20205\\
0.05	0.2001\\
0.1	0.19844\\
0.2	0.19488\\
0.3	0.19066\\
0.4	0.18675\\
0.5	0.18328\\
0.6	0.17833\\
0.7	0.17501\\
0.8	0.17116\\
0.9	0.1673\\
1	0.16374\\
1.1	0.15968\\
1.2	0.15594\\
1.3	0.15133\\
1.4	0.14725\\
1.5	0.14437\\
1.6	0.1405\\
1.7	0.13548\\
1.8	0.13197\\
1.9	0.12864\\
2	0.12543\\
2.25	0.11528\\
2.5	0.10928\\
2.75	0.10461\\
3	0.10165\\
3.25	0.09918\\
3.5	0.098165\\
3.75	0.097176\\
4	0.097112\\
5	0.097478\\
6	0.10088\\
7	0.10516\\
8	0.11018\\
9	0.11419\\
10	0.11872\\
15	0.13342\\
20	0.14053\\
25	0.14536\\
30	0.14937\\
35	0.15414\\
40	0.15912\\
45	0.16345\\
50	0.16647\\
100	0.18726\\
200	0.20391\\
300	0.23086\\
400	0.24668\\
500	0.2587\\
};
\addlegendentry{heat}

\addplot [color=mycolor2, dashed, line width=2.0pt]
  table[row sep=crcr]{%
0.001	0.20238\\
500	0.20238\\
};
\addlegendentry{standard}

\addplot [color=mycolor3, dashed, line width=2.0pt]
  table[row sep=crcr]{%
0.001	0.088847\\
500	0.088847\\
};
\addlegendentry{ours}

\end{axis}
\end{tikzpicture}%

%% file: sec_appendix_descSize.tex
\section{Change the descriptor size}\label{appendix:descSize}
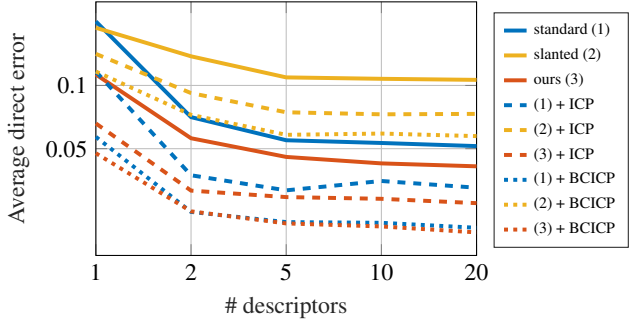
\begin{figure}[!h]
    \centering
    \input{figures_suppl/diff_desc_faust_direct_complete.tex}
    \caption{The average geodesic error v.s. the number of descriptors used for functional map estimation on the FAUST dataset. Note that our mask leads to better initialization, which results in better maps, even after the ICP and BCICP refinement.}
    \label{fig:append:descSize}
    \vspace{-12pt}
\end{figure}

Fig.~\ref{fig:append:descSize} shows the average direct error of 300 FAUST shapes with different number of input descriptors. The descriptors are used to optimize a 100-by-100 functional map for each test pair. The results with ICP/BCICP refinement are also included.

\begin{table}[!h]
\caption{We also compare our results to the exact setting of~\cite{ren2018continuous} on the same set of shape pairs of FAUST dataset. The average direct geodesic error over 300 shape pairs are reported and our mask leads to better results.
}
\label{table:appendix:compare_to_SGA}
\centering
\begin{tabular}{|c|ccc|}
\hline
\multirow{2}{*}{} & \multicolumn{3}{c|}{Avg. direct error \footnotesize{($\times 10^{-3}$)}} \\ \cline{2-4} 
 & Ini & +ICP & +BCICP \\ \hline
\cite{ren2018continuous} & 175.5 & 121.2 & 58.2 \\
Ours & 72.7 & 53.5 & 42.0 \\ \hline
\end{tabular}
\end{table}

We also compare to the results obtained directly using the code and the dataset provided by the authors of~\cite{ren2018continuous}
(see Table~{\ref{table:appendix:compare_to_SGA}}). The reproduced results that are reported in Table~{\ref{table:appendix:compare_to_SGA}}
are consistent with the values reported in the Table 1-2 ("WKS + directOp + BCICP") in the paper,
and Table 2-3 ("WKS + directOp") in the supplementary materials of~\cite{ren2018continuous}. Note that~\cite{ren2018continuous} split the dataset into isometric and non-isometric categories, and here we report the results altogether.

Recall that the total energy to optimize is 
$E\big(\matr{C}_{12}\big) = \alpha_1E_{\text{desc}} + \alpha_2 E_{\text{mult}} + \alpha_3 E_{\text{orient}} + \alpha_4 E_{\text{mask}}$. 
We would like to highlight the fact that the parameters in~\cite{ren2018continuous} are not set in the same way as in our work. In particular,
in~\cite{ren2018continuous}
the standard Laplacian mask is used, and the weight $\alpha_i$ are set to fixed values $\alpha_i^*$. While in our test,
we used the proposed resolvent Laplacian mask, and the weight $\alpha_i$ are set to $\alpha_i^* / E_i(C_{\text{ini}})$, where 
$E_i(C_{\text{ini}})$ is the corresponding energy term acting on the initial functional map $C_{\text{ini}}$. 

Our approach allows a better control over the relative contribution of the different terms in the energy and we observed that it typically works better in practice as well. 
Specifically, in our setting, the relative weight of each term $\alpha_i E_i(C)$ are fixed across different shape pairs.
In this case, if we change the mask construction, 
we can conclude that the improvement indeed come from our proposed resolvent mask. 
However, in the comparison to the exact setting of~\cite{ren2018continuous}, since only $\alpha_i^*$ is fixed over different test pairs, 
the relative weight of each term $\alpha_i^* E_i(C)$ can have different scale 
since different test pairs may have different scale of eigenvalues, descriptors and etc. 
Therefore, the improvement from Table~{\ref{table:appendix:compare_to_SGA}} is not obtained in a well controlled setting, 
since the improvement can also come from the change of the relative weight of different energy terms as well as from our resolvent mask.

Thus, in addition to the new Table~{\ref{table:appendix:compare_to_SGA}}, we also emphasize that in Fig.~{\ref{fig:append:descSize}} we provide 
a more fair and controlled comparison to~\cite{ren2018continuous}
in which 10 pairs of descriptors are used. Note that, in all the tests across the paper, 
we used the above discussed way to fix the relative weight of the mask term w.r.t. the rest terms to make sure the improvement indeed comes from our new resolvent mask.

%% file: figures_suppl/diff_desc_faust_direct_complete.tex
%
%
\definecolor{mycolor1}{rgb}{0.00000,0.44706,0.74118}%
\definecolor{mycolor2}{rgb}{0.92941,0.69412,0.12549}%
\definecolor{mycolor3}{rgb}{0.85098,0.32549,0.09804}%
\pgfplotsset{scaled x ticks=false}
\pgfplotsset{
compat=1.11,
legend image code/.code={
\draw[mark repeat=2,mark phase=2]
plot coordinates {
(0cm,0cm)
(0.15cm,0cm)        
(0.3cm,0cm)         
};%
}
}
\begin{tikzpicture}

\begin{axis}[%
width=0.6\linewidth,
height=0.4\linewidth,
at={(1.472in,1.026in)},
scale only axis,
xmin=1,
xmax=5,
ymode=log,
xtick={1,2,3,4,5},
xticklabels={{1},{2},{5},{10},{20}},
ytick={0.05, 0.1},
yticklabels={0.05,0.1},
xlabel style={font=\color{white!15!black}},
xlabel={\# descriptors},
ymin=0,
ymax=0.25,
ylabel style={font=\color{white!15!black}},
ylabel={Average direct error},
axis background/.style={fill=white},
xmajorgrids,
ymajorgrids,
legend style={at={(1.4,0.03)}, anchor=south east, legend cell align=left, align=left, draw=white!15!black, fill=white, fill opacity=0.6, draw opacity=1, text opacity=1},
]
\addplot [color=mycolor1, line width=1.5 pt]
  table[row sep=crcr]{%
1	0.20173\\
2	0.070664\\
3	0.054841\\
4	0.053226\\
5	0.051437\\
};
\addlegendentry{\scriptsize standard (1)}

\addplot [color=mycolor2, line width=1.5 pt]
  table[row sep=crcr]{%
1	0.18815\\
2	0.13763\\
3	0.10923\\
4	0.10755\\
5	0.1063\\
};
\addlegendentry{\scriptsize slanted (2)}

\addplot [color=mycolor3, line width=1.5 pt]
  table[row sep=crcr]{%
1	0.11276\\
2	0.056063\\
3	0.045678\\
4	0.042554\\
5	0.041212\\
};
\addlegendentry{\scriptsize ours (3)}

\addplot [color=mycolor1, dashed, line width=1.5 pt]
  table[row sep=crcr]{%
1	0.11649\\
2	0.037365\\
3	0.031668\\
4	0.035089\\
5	0.032601\\
};
\addlegendentry{\scriptsize (1) + ICP}

\addplot [color=mycolor2, dashed, line width=1.5 pt]
  table[row sep=crcr]{%
1	0.14135\\
2	0.091973\\
3	0.074516\\
4	0.072879\\
5	0.073227\\
};
\addlegendentry{\scriptsize (2) + ICP}

\addplot [color=mycolor3, dashed, line width=1.5 pt]
  table[row sep=crcr]{%
1	0.065958\\
2	0.031577\\
3	0.029414\\
4	0.028862\\
5	0.02752\\
};
\addlegendentry{\scriptsize (3) + ICP}

\addplot [color=mycolor1, dotted, line width=1.5pt]
  table[row sep=crcr]{%
1	0.056924\\
2	0.024996\\
3	0.022361\\
4	0.022187\\
5	0.02101\\
};
\addlegendentry{\scriptsize (1) + BCICP}

\addplot [color=mycolor2, dotted, line width=1.5pt]
  table[row sep=crcr]{%
1	0.11649\\
2	0.07245\\
3	0.058206\\
4	0.059018\\
5	0.057382\\
};
\addlegendentry{\scriptsize (2) + BCICP}

\addplot [color=mycolor3, dotted, line width=1.5pt]
  table[row sep=crcr]{%
1	0.04757\\
2	0.025163\\
3	0.022038\\
4	0.021319\\
5	0.019977\\
};
\addlegendentry{\scriptsize (3) + BCICP}

\end{axis}
\end{tikzpicture}%

%% file: sec_appendix_stability.tex
\section{Stability under remeshing and refinement}\label{appendix:stability}
\begin{figure}[!h]
\centering
  \begin{overpic}
  [trim=5cm -1cm 5cm 0cm,clip,width=1\columnwidth,grid=false]{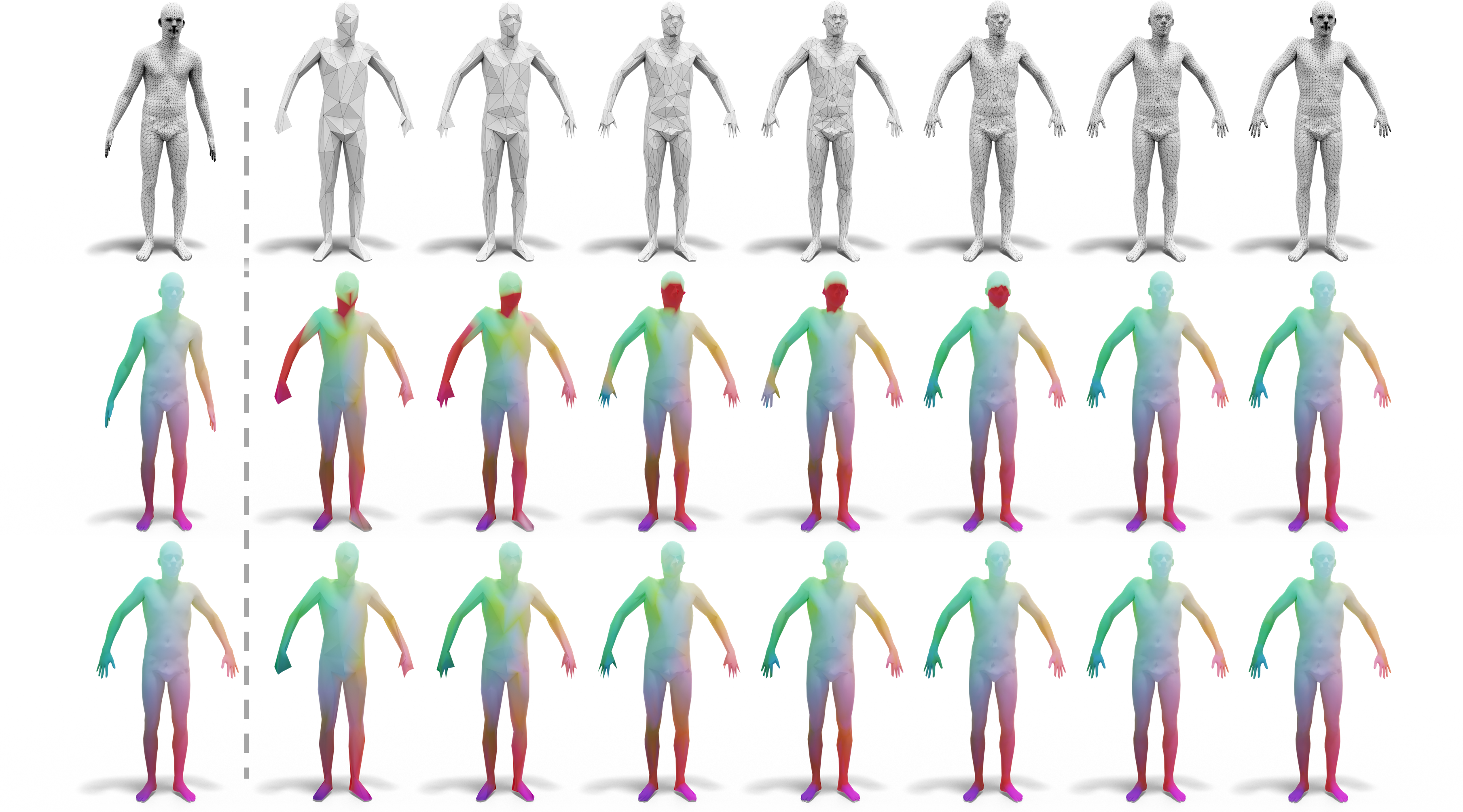}
  \put(1,64){\scriptsize{n = 6890}}
  \put(16.5,64){\scriptsize{n = 200}}
  \put(30,64){\scriptsize{n=300}}
  \put(42,64){\scriptsize{n=500}}
  \put(54,64){\scriptsize{n=1000}}
  \put(67,64){\scriptsize{n=3000}}
  \put(79,64){\scriptsize{n=5000}}
  \put(92,64){\scriptsize{n=6890}}
  \put(-1,40){\scriptsize{Source}}
  \put(-1,20){\scriptsize{Target}}
  \put(101,39){\footnotesize{\rotatebox{-90}{Standard}}}
  \put(101,16){\footnotesize{\rotatebox{-90}{\textbf{Ours}}}}
  \end{overpic}
  \caption{\hl{
  First row: The target shape is fixed, while the source shape is remeshed and downsampled to resolution ranging from 200 to 5000 (the original target shape without remeshing is shown in the last column); Second row: we use the standard mask to optimize a 100-by-100 functional map with 3 descriptors. The recovered pointwise maps are visualized on the corresponding target shape; Third row: similar to the second row but using our resolvent mask. 
  } 
  }
  \label{fig:appendix:stability}
\end{figure}
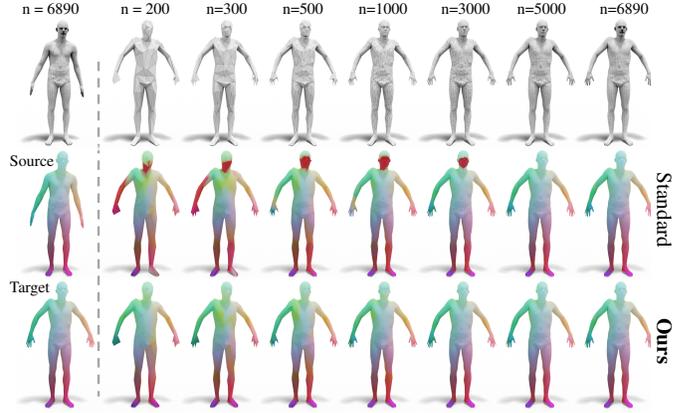

\hl{
Fig.~{\ref{fig:appendix:stability}} illustrate the stability of our resolvent mask under remeshing and refinement. Specifically, we fix the source shape, and remesh and downsample the target shape to different resolutions ranging from 200 to 5000. Note that the original source shape (first row, first column) and the original target shape (first row, last column) have the same triangulation. The downsampled meshes (using QSlim) are shown in the first row with the number of vertices reported in the above.

We then compute a 100-by-100 functional map between the source shape and the downsampled target shape using the standard Laplacian mask (the corresponding point maps are shown in the second row) and our resolvent mask (in the third row). We can see that, our resolvent mask is more stable than the standard mask across different mesh resolution and irregular/inconsistent triangulation.
}

%% file: main.bbl
\newcommand{\etalchar}[1]{$^{#1}$}
\begin{thebibliography}{\uppercase{VKZHCO11}}

\bibitem[ADK16]{aflalo2013spectral}
\textsc{Aflalo Y., Dubrovina A., Kimmel R.}:
\newblock Spectral generalized multi-dimensional scaling.
\newblock \emph{International Journal of Computer Vision 118}, 3 (2016),
  380--392.

\bibitem[ASC11]{aubry2011wave}
\textsc{Aubry M., Schlickewei U., Cremers D.}:
\newblock The wave kernel signature: A quantum mechanical approach to shape
  analysis.
\newblock In \emph{2011 IEEE international conference on computer vision
  workshops (ICCV workshops)} (2011), IEEE, pp.~1626--1633.

\bibitem[BBK08]{tosca}
\textsc{Bronstein A.~M., Bronstein M.~M., Kimmel R.}:
\newblock \emph{{N}umerical {G}eometry of {N}on-{R}igid {S}hapes}.
\newblock Springer Science \& Business Media, 2008.

\bibitem[BCBB16]{biasotti2016recent}
\textsc{Biasotti S., Cerri A., Bronstein A., Bronstein M.}:
\newblock Recent trends, applications, and perspectives in 3d shape similarity
  assessment.
\newblock In \emph{Computer Graphics Forum} (2016), vol.~35, pp.~87--119.

\bibitem[BDK17]{burghard2017embedding}
\textsc{Burghard O., Dieckmann A., Klein R.}:
\newblock Embedding shapes with {G}reen's functions for global shape matching.
\newblock \emph{Computers \& Graphics 68} (2017), 1--10.

\bibitem[BRLB14]{Bogo:CVPR:2014}
\textsc{Bogo F., Romero J., Loper M., Black M.~J.}:
\newblock {FAUST}: Dataset and evaluation for {3D} mesh registration.
\newblock In \emph{Proceedings IEEE Conf. on Computer Vision and Pattern
  Recognition (CVPR)} (Piscataway, NJ, USA, June 2014), IEEE.

\bibitem[CSBK18]{choukroun2018hamiltonian}
\textsc{Choukroun Y., Shtern A., Bronstein A.~M., Kimmel R.}:
\newblock Hamiltonian operator for spectral shape analysis.
\newblock \emph{IEEE transactions on visualization and computer graphics}
  (2018).

\bibitem[Cut13]{cuturi2013sinkhorn}
\textsc{Cuturi M.}:
\newblock Sinkhorn distances: Lightspeed computation of optimal transport.
\newblock In \emph{Advances in neural information processing systems} (2013),
  pp.~2292--2300.

\bibitem[Dod81]{dodziuk1981eigenvalues}
\textsc{Dodziuk J.}:
\newblock Eigenvalues of the {L}aplacian and the heat equation.
\newblock \emph{The American Mathematical Monthly 88}, 9 (1981), 686--695.

\bibitem[EBC17]{ezuz2017deblurring}
\textsc{Ezuz D., Ben-Chen M.}:
\newblock Deblurring and denoising of maps between shapes.
\newblock In \emph{Computer Graphics Forum} (2017), vol.~36, Wiley Online
  Library, pp.~165--174.

\bibitem[ERGB16]{eynard2016coupled}
\textsc{Eynard D., Rodola E., Glashoff K., Bronstein M.~M.}:
\newblock Coupled functional maps.
\newblock In \emph{3D Vision (3DV), 2016 Fourth International Conference on}
  (2016), IEEE, pp.~399--407.

\bibitem[GBKS18]{gehre2018interactive}
\textsc{Gehre A., Bronstein M., Kobbelt L., Solomon J.}:
\newblock Interactive curve constrained functional maps.
\newblock In \emph{Computer Graphics Forum} (2018), vol.~37, Wiley Online
  Library, pp.~1--12.

\bibitem[GBP]{giorgi2007shape}
\textsc{Giorgi D., Biasotti S., Paraboschi L.}:
\newblock Shape retrieval contest 2007: Watertight models track.

\bibitem[HO17]{huang2017adjoint}
\textsc{Huang R., Ovsjanikov M.}:
\newblock Adjoint map representation for shape analysis and matching.
\newblock In \emph{Computer Graphics Forum} (2017), vol.~36, Wiley Online
  Library, pp.~151--163.

\bibitem[HRS{\etalchar{*}}16]{heeren2016splines}
\textsc{Heeren B., Rumpf M., Schr{\"o}der P., Wardetzky M., Wirth B.}:
\newblock Splines in the space of shells.
\newblock In \emph{Computer Graphics Forum} (2016), vol.~35, Wiley Online
  Library, pp.~111--120.

\bibitem[HWG14]{huang2014functional}
\textsc{Huang Q., Wang F., Guibas L.}:
\newblock Functional map networks for analyzing and exploring large shape
  collections.
\newblock \emph{ACM Transactions on Graphics (TOG) 33}, 4 (2014), 36.

\bibitem[KBB{\etalchar{*}}13]{kovnatsky2013coupled}
\textsc{Kovnatsky A., Bronstein M.~M., Bronstein A.~M., Glashoff K., Kimmel
  R.}:
\newblock Coupled quasi-harmonic bases.
\newblock In \emph{Computer Graphics Forum} (2013), vol.~32, pp.~439--448.

\bibitem[KBBV15]{kovnatsky2015functional}
\textsc{Kovnatsky A., Bronstein M.~M., Bresson X., Vandergheynst P.}:
\newblock Functional correspondence by matrix completion.
\newblock In \emph{Proceedings of the IEEE conference on computer vision and
  pattern recognition} (2015), pp.~905--914.

\bibitem[KGB16]{kovnatsky2016madmm}
\textsc{Kovnatsky A., Glashoff K., Bronstein M.~M.}:
\newblock {MADMM:} a generic algorithm for non-smooth optimization on
  manifolds.
\newblock In \emph{Proc. ECCV} (2016), Springer, pp.~680--696.

\bibitem[LRB{\etalchar{*}}16]{litany2016non}
\textsc{Litany O., Rodol{\`a} E., Bronstein A.~M., Bronstein M.~M., Cremers
  D.}:
\newblock Non-rigid puzzles.
\newblock In \emph{Computer Graphics Forum} (2016), vol.~35, Wiley Online
  Library, pp.~135--143.

\bibitem[LRBB17]{litany2017fully}
\textsc{Litany O., Rodol{\`a} E., Bronstein A.~M., Bronstein M.~M.}:
\newblock Fully spectral partial shape matching.
\newblock In \emph{Computer Graphics Forum} (2017), vol.~36, Wiley Online
  Library, pp.~247--258.

\bibitem[MCSK{\etalchar{*}}17]{mandad2017variance}
\textsc{Mandad M., Cohen-Steiner D., Kobbelt L., Alliez P., Desbrun M.}:
\newblock Variance-minimizing transport plans for inter-surface mapping.
\newblock \emph{ACM Transactions on Graphics (TOG) 36} (2017), 14.

\bibitem[MP49]{minakshisundaram1949some}
\textsc{Minakshisundaram S., Pleijel A.}:
\newblock Some properties of the eigenfunctions of the {L}aplace-operator on
  {R}iemannian manifolds.
\newblock \emph{Canadian J. Math 1}, 8 (1949).

\bibitem[MRCB18]{melzi2018localized}
\textsc{Melzi S., Rodol{\`a} E., Castellani U., Bronstein M.~M.}:
\newblock Localized manifold harmonics for spectral shape analysis.
\newblock In \emph{Computer Graphics Forum} (2018), vol.~37, Wiley Online
  Library, pp.~20--34.

\bibitem[NMR{\etalchar{*}}18]{nogneng2018improved}
\textsc{Nogneng D., Melzi S., Rodol{\`a} E., Castellani U., Bronstein M.,
  Ovsjanikov M.}:
\newblock Improved functional mappings via product preservation.
\newblock In \emph{Computer Graphics Forum} (2018), vol.~37, Wiley Online
  Library, pp.~179--190.

\bibitem[NO17]{nogneng2017informative}
\textsc{Nogneng D., Ovsjanikov M.}:
\newblock Informative descriptor preservation via commutativity for shape
  matching.
\newblock \emph{Computer Graphics Forum 36}, 2 (2017), 259--267.

\bibitem[NVT{\etalchar{*}}14]{neumann2014compressed}
\textsc{Neumann T., Varanasi K., Theobalt C., Magnor M., Wacker M.}:
\newblock Compressed manifold modes for mesh processing.
\newblock In \emph{Computer Graphics Forum} (2014), vol.~33, Wiley Online
  Library, pp.~35--44.

\bibitem[OBCS{\etalchar{*}}12]{ovsjanikov2012functional}
\textsc{Ovsjanikov M., Ben-Chen M., Solomon J., Butscher A., Guibas L.}:
\newblock {F}unctional {M}aps: {A} {F}lexible {R}epresentation of {M}aps
  {B}etween {S}hapes.
\newblock \emph{ACM Transactions on Graphics (TOG) 31}, 4 (2012), 30.

\bibitem[OCB{\etalchar{*}}17]{ovsjanikov2017course}
\textsc{Ovsjanikov M., Corman E., Bronstein M., Rodol\`{a} E., Ben-Chen M.,
  Guibas L., Chazal F., Bronstein A.}:
\newblock Computing and processing correspondences with functional maps.
\newblock In \emph{ACM SIGGRAPH 2017 Courses} (2017), SIGGRAPH '17,
  pp.~5:1--5:62.

\bibitem[PSO18]{poulenard2018topological}
\textsc{Poulenard A., Skraba P., Ovsjanikov M.}:
\newblock Topological function optimization for continuous shape matching.
\newblock In \emph{Computer Graphics Forum} (2018), vol.~37, Wiley Online
  Library, pp.~13--25.

\bibitem[RCB{\etalchar{*}}17]{rodola2017partial}
\textsc{Rodol{\`a} E., Cosmo L., Bronstein M.~M., Torsello A., Cremers D.}:
\newblock Partial functional correspondence.
\newblock In \emph{Computer Graphics Forum} (2017), vol.~36, Wiley Online
  Library, pp.~222--236.

\bibitem[RMC15]{rodola-vmv15}
\textsc{Rodol\`{a} E., Moeller M., Cremers D.}:
\newblock Point-wise map recovery and refinement from functional
  correspondence.
\newblock In \emph{Proc. Vision, Modeling and Visualization (VMV)} (2015).

\bibitem[ROA{\etalchar{*}}13]{rustamov13}
\textsc{Rustamov R., Ovsjanikov M., Azencot O., Ben-Chen M., Chazal F., Guibas
  L.}:
\newblock Map-based exploration of intrinsic shape differences and variability.
\newblock \emph{ACM Transactions on Graphics (TOG) 32}, 4 (July 2013),
  72:1--72:12.

\bibitem[RPWO18]{ren2018continuous}
\textsc{Ren J., Poulenard A., Wonka P., Ovsjanikov M.}:
\newblock Continuous and orientation-preserving correspondences via functional
  maps.
\newblock \emph{ACM Transactions on Graphics (TOG) 37}, 6 (2018).

\bibitem[RS80]{reedsimon1}
\textsc{Reed M., Simon B.}:
\newblock \emph{Methods of Modern Mathematical Physics I: Functional Analysis},
  revised and enlarged~ed.
\newblock Academic Press, San Diego, 1980.

\bibitem[Sau06a]{sauvigny2006partial1}
\textsc{Sauvigny F.}:
\newblock \emph{Partial Differential Equations 1: Foundations and Integral
  Representations}.
\newblock Springer, Berlin, Heidelberg, 2006.

\bibitem[Sau06b]{sauvigny2006partial2}
\textsc{Sauvigny F.}:
\newblock \emph{Partial Differential Equations 2: Functional Analytic Methods}.
\newblock Springer, Berlin, Heidelberg, 2006.

\bibitem[SDGP{\etalchar{*}}15]{solomon2015convolutional}
\textsc{Solomon J., De~Goes F., Peyr{\'e} G., Cuturi M., Butscher A., Nguyen
  A., Du T., Guibas L.}:
\newblock Convolutional wasserstein distances: Efficient optimal transportation
  on geometric domains.
\newblock \emph{ACM Transactions on Graphics (TOG) 34}, 4 (2015), 66.

\bibitem[SPKS16]{solomon2016entropic}
\textsc{Solomon J., Peyr{\'e} G., Kim V.~G., Sra S.}:
\newblock Entropic metric alignment for correspondence problems.
\newblock \emph{ACM Transactions on Graphics (TOG) 35}, 4 (2016), 72.

\bibitem[TCL{\etalchar{*}}13]{tam2013registration}
\textsc{Tam G.~K., Cheng Z.-Q., Lai Y.-K., Langbein F.~C., Liu Y., Marshall D.,
  Martin R.~R., Sun X.-F., Rosin P.~L.}:
\newblock Registration of {3D} point clouds and meshes: a survey from rigid to
  nonrigid.
\newblock \emph{IEEE TVCG 19}, 7 (2013), 1199--1217.

\bibitem[VKZHCO11]{van2011survey}
\textsc{Van~Kaick O., Zhang H., Hamarneh G., Cohen-Or D.}:
\newblock A survey on shape correspondence.
\newblock In \emph{Computer Graphics Forum} (2011), vol.~30, pp.~1681--1707.

\bibitem[VLR{\etalchar{*}}17]{vestner2017product}
\textsc{Vestner M., Litman R., Rodol{\`a} E., Bronstein A., Cremers D.}:
\newblock Product manifold filter: Non-rigid shape correspondence via kernel
  density estimation in the product space.
\newblock In \emph{Proc. CVPR} (2017), pp.~6681--6690.

\bibitem[WGBS18]{wang2018kernel}
\textsc{Wang L., Gehre A., Bronstein M.~M., Solomon J.}:
\newblock Kernel functional maps.
\newblock In \emph{Computer Graphics Forum} (2018), vol.~37, Wiley Online
  Library, pp.~27--36.

\bibitem[WHG13]{wang2013image}
\textsc{Wang F., Huang Q., Guibas L.~J.}:
\newblock Image co-segmentation via consistent functional maps.
\newblock In \emph{Proceedings of the IEEE International Conference on Computer
  Vision} (2013), pp.~849--856.

\bibitem[WLZT18]{wang2018vector}
\textsc{Wang Y., Liu B., Zhou K., Tong Y.}:
\newblock Vector field map representation for near conformal surface
  correspondence.
\newblock In \emph{Computer Graphics Forum} (2018), vol.~37, Wiley Online
  Library, pp.~72--83.

\end{thebibliography}
